\documentclass[10pt]{article}%
\usepackage{graphicx}
\usepackage{amssymb}
\usepackage{amsfonts}
\usepackage{amsmath}
\usepackage{latexsym}
\usepackage{afterpage}
\usepackage{graphicx}
\usepackage[nohead]{geometry}
\usepackage[doublespacing]{setspace}
\usepackage[bottom]{footmisc}
\usepackage{indentfirst}
\usepackage{endnotes}
\usepackage{booktabs}
\usepackage{rotating}
\usepackage{comment}
\usepackage{lscape}
\usepackage{float}
\usepackage{longtable}%
\providecommand{\U}[1]{\protect\rule{.1in}{.1in}}
\newtheorem{theorem}{Theorem}

\newtheorem{assumption}{Assumption}

\newtheorem{corollary}{Corollary}

\newtheorem{example}{Example}

\newtheorem{lemma}{Lemma}

\newtheorem{remark}{Remark}

\newenvironment{proof}[1][Proof]{\noindent\textbf{#1.} }{\ \rule{0.5em}{0.5em}}
\makeatletter
\def\@biblabel#1{\hspace*{-\labelsep}}
\makeatother
\begin{document}

\title{Identifying the Distribution of Treatment Effects \\under Support Restrictions}
\author{Ju Hyun Kim\thanks{University of North Carolina at Chapel Hill,
juhkim@email.unc.edu. I am greatly indebted to my advisor Bernard Salani\'{e}
for numerous discussions and his support throughout this project. I am also
grateful to Jushan Bai, Pierre-Andr\'{e} Chiappori, Serena Ng, and Christoph
Rothe for their encouragement and helpful comments. This paper has benefited
from discussions with Andrew Chesher, Xavier D'Haultfoeuille, Alfred Galichon,
Kyle Jurado, Shakeeb Khan, Toru Kitagawa, Dennis Kristensen, Seunghoon Na,
Salvador Navarro, Byoung Park, and Charles Zheng. While writing this paper, I
was generously supported by a Wueller pre-dissertation award fellowship from
the Economics Department at Columbia University. All errors are mine.}}
\date{\today }
\maketitle

\begin{abstract}
The distribution of treatment effects (DTE) is often of interest in the
context of welfare policy evaluation. In this paper, I consider partial
identification of the DTE under known marginal distributions and support
restrictions on the potential outcomes. Examples of such support restrictions
include monotone treatment response, concave treatment response, convex
treatment response, and the Roy model of self-selection. To establish
informative bounds on the DTE, I formulate the problem as an optimal
transportation linear program and develop a new dual representation to
characterize the identification region with respect to the known marginal
distributions. I use this result to derive informative bounds for concrete
economic examples. I also propose an estimation procedure and illustrate the
usefulness of my approach in the context of an empirical analysis of the
effects of smoking on infant birth weight. The empirical results show that
monotone treatment response has a substantial identifying power for the DTE
when the marginal distributions of the potential outcomes are given.\newline

\end{abstract}

\section{Introduction}

In this paper, I study partial identification of the distribution of treatment
effects (DTE) under a broad class of restrictions on potential outcomes. The
DTE is defined as follows: for any fixed $\delta\in$ $\mathbb{R},$%
\[
F_{\Delta}\left(  \delta\right)  =\Pr\left(  \Delta\leq\delta\right)  ,
\]
with the treatment effect $\Delta=Y_{1}-Y_{0}$ where $Y_{0}$ and $Y_{1}$
denote the potential outcomes without and with some treatment, respectively.
The question that I am interested in is how treatment effects or program
benefits are distributed across the population.

In the context of welfare policy evaluation, distributional aspects of the
effects are often of interest, e.g. "which individuals are severely affected
by the program?" or "how are those benefits distributed across the
population?". As Heckman et al. (1997) pointed out, the DTE is particularly
important when treatments produce nontransferable and nonredistributable
benefits such as outcomes in health interventions, academic achievement in
educational programs, and occupational skills in job training programs or when
some individuals experience severe welfare changes at the tails of the impact distribution.

Although most empirical research on program evaluation has focused on average
treatment effects (ATE) or marginal distributions of potential outcomes, these
parameters are limited in their ability to capture heterogeneity of the
treatment effects at the individual level. For example, consider two projects
with the same average benefits, one of which concentrates benefits among a
small group of people, while the other distributes benefits evenly across the
population. ATE cannot differentiate between the two projects because it shows
only the central tendency of treatment effects as a location parameter,
whereas the DTE captures information about the entire distribution. Marginal
distributions of $Y_{0}$ and $Y_{1}$ are also uninformative about parameters
on the individual specific heterogeneity in treatment effects including the
fraction of the population that benefits from a program $\Pr\left(  Y_{1}\geq
Y_{0}\right)  ,$ the fraction of the population that has gains or losses in a
specific range $\Pr\left(  \delta^{L}\leq Y_{1}-Y_{0}\leq\delta^{U}\right)  $,
the $q$-quantile of the impact distribution $\inf\left\{  \delta:F_{\Delta
}\left(  \delta\right)  >q\right\}  $, etc. See, e.g. Heckman et al. (1997),
Abbring and Heckman (2007), and Firpo and Ridder (2008), among others for more details.

Despite the importance of these parameters in economics, related empirical
research has been hampered by difficulties associated with identifying the
entire distribution of effects. The central challenge arises from a missing
data problem: under mutually exclusive treatment participation,
econometricians can observe either a treated outcome or an untreated outcome,
but both potential outcomes $Y_{0}$ and $Y_{1}$ are never simultaneously
observed for each agent. Therefore, the joint distribution of $Y_{0}$ and
$Y_{1}$\ is not typically exactly identified, which\ complicates
identification of the DTE, which is point-identified only under strong
assumptions about\ each individual's rank across the treatment status or
specifications on the joint distribution of $Y_{0}$ and $Y_{1}$, which are
often not justified by economic theory or plausible priors.

This paper relies on partial identification to avoid strong assumptions and
remain cautious of assumption-driven conclusions. In the related literature,
Manski (1997) established bounds on the DTE under monotone treatment response
(MTR), which assumes that the treatment effects are nonnegative. Fan and Park
(2009, 2010) and Fan and Wu (2010) adopted results from copula theory to
establish bounds on the DTE, given marginal distributions. Unfortunately, both
approaches deliver bounds that are often too wide to be informative in
practice. Since these two conditions are often plausible in practice, a
natural way to tighten the bounds is considering both MTR and given marginal
distributions of potential outcomes. However, methods of establishing
informative bounds on the DTE under these two restrictions have remained
unanswered. Specifically, in the existing copula approach it is technically
challenging to find out the particular joint distributions that achieve the
best possible bounds on the DTE under the two restrictions.

In this paper, I propose a novel approach to circumvent these difficulties
associated with identifying the DTE under these two restrictions.
Methodologically, my approach involves formulating the problem as an optimal
transportation linear program and embedding support restrictions on the
potential outcomes including MTR into the cost function. A key feature of the
optimal transportation approach is that it admits a dual formulation. This
makes it possible to derive the best possible bounds from the optimization
problem with respect to given marginal distributions but not the joint
distribution, which is an advantage over the copula approach. Specifically,
the linearity of support restrictions in the entire joint distribution allows
for the penalty formulation. Since support restrictions hold with probability
one, the corresponding multiplier on those constraints should be infinite. To
the best of my knowledge, the dual representation of such an optimization
problem with an infinite Lagrange multiplier has not been derived in the
literature. In this paper, I develop a dual representation for $\{0,1,\infty
\}$-valued costs by extending the existing result on duality for
$\{0,1\}$-valued costs.

My approach applies to general support restrictions on the potential outcomes
as well as MTR. Such support restrictions encompass shape restrictions on the
treatment response function that can be written as $g\left(  Y_{0}%
,Y_{1}\right)  \leq0$ with probability one for any continuous function
$g:\mathbb{R}\rightarrow\mathbb{R}$, including MTR, concave treatment
response, and convex treatment response.\footnote{Let $Y_{d}=f\left(
t_{d}\right)  $ where $Y_{d}$ is a potential outcome and $t_{d}$ is a level of
inputs for multi-valued treatment status $d.$\ Concave treatment response and
convex treatment response assume that the treatment response function $f$ is
concave and convex, respectively.} Moreover, considering support restrictions
opens the way to identify the DTE in the Roy model of self-selection and the
DTE conditional on some sets of potential outcomes.

Numerous examples in applied economics fit into this setting because marginal
distributions are point or partially identified under weak conditions and
support restrictions are often implied by economic theory and plausible
priors. The marginal distributions of the potential outcomes are
point-identified in randomized experiments or under unconfoundedness. Even if
selection depends on unobservables, they are point-identified for compliers
under the local average treatment effects assumptions (Imbens and Rubin
(1997), Abadie (2002)) and are partially identified in the presence of
instrumental variables (Kitagawa (2009)). Also, MTR has been defended as a
plausible restriction in empirical studies of returns to education (Manski and
Pepper (2000)), the effect of funds for low-ability pupils (Haan (2012)), the
impact of the National School Lunch Program on children's health (Gundersen et
al. (2011)), and various medical treatments (Bhattacharya et al. (2005,
2012)). Researchers sometimes have plausible information on the shape of
treatment response functions from economic theory or from empirical results in
previous studies. For example, based on diminishing marginal returns to
production, one may find it plausible to assume that the marginal effect of
improved maize seed adoption on productivity diminishes as the level of
adoption increases, holding other inputs fixed. Also, one may want to assume
that the marginal adverse effect of an additional cigarette on infant birth
weight diminishes as the number of cigarettes increases as shown in Hoderlein
and Sasaki (2013). In the empirical literature, concave treatment response has
been assumed for returns to schooling (Okumura (2010)) and convex treatment
response for the effect of education on smoking (Boes (2010)).\footnote{All of
these studies considered ATE\ or marginal distributions of potential outcomes
only.}

A considerable amount of the literature has used the Roy model to describe
people's self-selection ranging from immigration to the U.S. (Borjas (1987))
to college entrance (Heckman et al. (2011)). Also, heterogeneity in treatment
effects for unobservable subgroups defined by particular sets of potential
outcomes has been of central interest in various empirical studies.
Heterogeneous peer effects and tracking impacts (Duflo et al. (2011)) and
heterogeneous class size effects (Ding and Lehrer (2008)) by the level of
students' performance, and the heterogeneity in smoking effects by potential
infant's birth weight (Hoderlein and Sasaki (2013)) have also been discussed
in the literature focusing on heterogeneous average effects.

I\ apply my method to an empirical analysis of the effects of smoking on
infant birth weight. I propose an estimation procedure and illustrate the
usefulness of my approach by showing that MTR has a substantial identifying
power for the distribution of smoking effects given marginal distributions. As
a support restriction, I\ assume that smoking has nonpositive effects on
infant birth weight. Smoking not only has a direct impact on infant birth
weight, but is also associated with unobservable factors that affect infant
birth weight. To overcome the endogenous selection problem, I\ make use of the
tax increase in Massachusetts in January 1993 as a source of exogenous
variation. I point-identify marginal distributions of potential infant birth
weight with and without smoking for compliers, which indicate pregnant women
who changed their smoking status from smoking to nonsmoking in response to
this tax shock. To estimate the marginal distributions of potential infant
birth weight, I\ use the instrumental variables (IV) method presented in
Abadie et al. (2002). Furthermore, I\ estimate the DTE\ bounds using plug-in
estimators based on the estimates of marginal distribution functions. As a
by-product, I find that the average adverse effect of smoking is more severe
for women with a higher tendency to smoke and that smoking women with some
college and college graduates are less likely to give births to low birth
weight infants than other smoking women.

In the next section, I\ give a formal description of the basic setup,
notation, terms and assumptions throughout this paper and present concrete
examples of support restrictions. I review the existing method of identifying
the DTE given marginal distributions without support restrictions to
demonstrate its limits in the presence of support restrictions. I also briefly
discuss the optimal transportation approach to describe the key idea of my
identification strategy. Section 3 formally characterizes the identification
region of the DTE under general support restrictions and derives informative
bounds for economic examples from the characterization. Section 4 provides
numerical examples to assess the informativeness of my new bounds\ and
analyzes sources of identification gains. Section 5 illustrates\ the
usefulness of these bounds by applying DTE bounds derived in Section 3 to an
empirical analysis of the impact distribution of smoking on infant birth
weight. Section 6 concludes and discusses interesting extensions.

\section{Basic Setup, DTE\ Bounds and Optimal Transportation Approach}

In this section, I\ present the potential outcomes setup that this study is
based on, the notation, and the assumptions used throughout this study.
I\ demonstrate that the bounds on the DTE established without support
restrictions are not the best possible bounds in the presence of support
restrictions. Then I\ propose a new method to derive sharp bounds on the DTE
based on the optimal transportation framework.

\subsection{Basic Setup}

The setup that I\ consider is as follows: the econometrician observes a
realized outcome variable $Y$ and a treatment participation indicator $D$ for
each individual$,$ where $D=1$ indicates treatment participation while $D=0$
nonparticipation. An observed outcome $Y$ can be written as $Y=DY_{1}%
+(1-D)Y_{0}$. Only $Y_{1}$ is observed for the individual who takes the
treatment while only $Y_{0}$ is observed for the individual who does not take
the treatment, where $Y_{0}$ and$\ Y_{1}$ are the potential outcome without
and with treatment, respectively. Treatment effects $\Delta$ are defined as
$\Delta=Y_{1}-Y_{0}$ the difference of potential outcomes. The objective of
this study is to identify the distribution function of treatment effects
$F_{\Delta}\left(  \delta\right)  =\Pr\left(  Y_{1}-Y_{0}\leq\delta\right)  $
from observed pairs $\left(  Y,D\right)  $ for fixed $\delta\in\mathbb{R}$ $.$

To avoid notational confusion, I\ differentiate between the
\emph{distribution} and the \emph{distribution function}. Let $\mu_{0}$,
$\mu_{1} $ and $\pi$ denote marginal distributions of $Y_{0}$ and $Y_{1}$, and
their joint distribution, respectively. That is, for any measurable set
$A_{d}$\ in $\mathbb{R}$, $\mu_{d}\left(  A_{d}\right)  =\Pr\left\{  Y_{d}\in
A_{d}\right\}  $ for $d\in\left\{  0,1\right\}  $ and $\pi\left(  A\right)
=\Pr\left\{  \left(  Y_{0},Y_{1}\right)  \in A\right\}  $ for any measurable
set $A$ in $\mathbb{R}^{2}$. In addition, let $F_{0},$ $F_{1}$\ and $F$ denote
marginal distribution functions of $Y_{0}$ and $Y_{1},$ and their joint
distribution function, respectively. That is, $F_{d}\left(  y_{d}\right)
=\mu_{d}\left(  (-\infty,y_{d}]\right)  $ and $F\left(  y_{0},y_{1}\right)
=\pi\left(  (-\infty,y_{0}]\times(-\infty,y_{1}]\right)  $ for any $y_{d}%
\in\mathbb{R}$ and $d\in\left\{  0,1\right\}  $. Let $\mathcal{Y}_{0}$ and
$\mathcal{Y}_{1}$ denote the support of $Y_{0}$ and $Y_{1}$, respectively.

In this paper, the identification region of $F_{\Delta}\left(  \delta\right)
$ is obtained for fixed marginal distributions. When marginal distributions
are only partially identified, DTE bounds are obtained by taking the union of
the bounds over all possible pairs of marginal distributions. Marginal
distributions of potential outcomes are point-identified in randomized
experiments or under selection on observables. Furthermore, previous studies
have shown that even if the selection is endogenous, marginal distributions of
potential outcomes are point or partially identified under relatively weak
conditions. Imbens and Rubin (1997) and Abadie (2002) showed that marginal
distributions for compliers are point-identified under the local average
treatment effects (LATE) assumptions, and Kitagawa (2009) obtained the
identification region of marginal distributions under IV
conditions.\footnote{Note that the conditions considered in these studies do
not restrict dependence between two potential outcomes.}

I\ impose the following assumption on the fixed marginal distribution
functions throughout this paper:

\begin{assumption}
The marginal distribution functions $F_{0}$ and $F_{1}$ are both absolutely
continuous with respect to the Lebesgue measure on $\mathbb{R}$.
\end{assumption}

In this paper, I obtain \emph{sharp} bounds on the DTE. Sharp bounds are
defined as the best possible bounds on the collection of DTE\ values that are
compatible with the observations $(Y,D)$ and given restrictions. Let
$F_{\Delta}^{L}(\delta)$ and $F_{\Delta}^{U}(\delta)$ denote the lower and
upper bounds on the DTE $F_{\Delta}(\delta)$:
\[
F_{\Delta}^{L}(\delta)\leq F_{\Delta}(\delta)\leq F_{\Delta}^{U}(\delta).
\]
If there exists an underlying joint distribution function $F$ that has fixed
marginal distribution functions $F_{0}$ and $F_{1}$ and generates $F_{\Delta
}(\delta)=F_{\Delta}^{L}(\delta)$ for fixed $\delta\in\mathbb{R},$ then
$F_{\Delta}^{L}(\delta)$ is called the \emph{sharp} lower bound. The
\emph{sharp} upper bound can be also defined in the same way. Note that
throughout this study, sharp bounds indicate \emph{pointwise} sharp bounds in
the sense that the underlying joint distribution function $F$ achieving sharp
bounds is allowed to vary with the value of $\delta.\footnote{If the
underlying joint distribution function $F$ does not depend on $\delta$, then
the sharp bounds are called \emph{uniformly} sharp bounds. Uniformly sharp
bounds are outside of the scope of this paper. For more details on uniform
sharpness, see Firpo and Ridder (2008).}$

To identify the DTE, I\ consider support restrictions, which can be written
as
\[
\Pr\left(  \left(  Y_{0},Y_{1}\right)  \in C\right)  =1,
\]
for some closed set $C$ in $\mathbb{R}^{2}$. This class of restrictions
encompasses any restriction that can be written as
\begin{equation}
g\left(  Y_{0},Y_{1}\right)  \leq0\text{ with probability one},\label{h}%
\end{equation}
for any continuous function $g:\mathbb{R}\times\mathbb{R}\rightarrow
\mathbb{R}.$ For example, shape restrictions on the treatment response
function such as MTR, concave response, and convex response can be written in
the form (1). Furthermore, identifying the DTE under support
restrictions opens the way to identify other parameters such as the DTE
conditional on the treated and the untreated in the Roy model, and the DTE
conditional on potential outcomes.%
\begin{figure}
[ptb]
\begin{center}
\includegraphics[
natheight=2.979300in,
natwidth=8.687900in,
height=2.4102in,
width=6.9764in
]%
{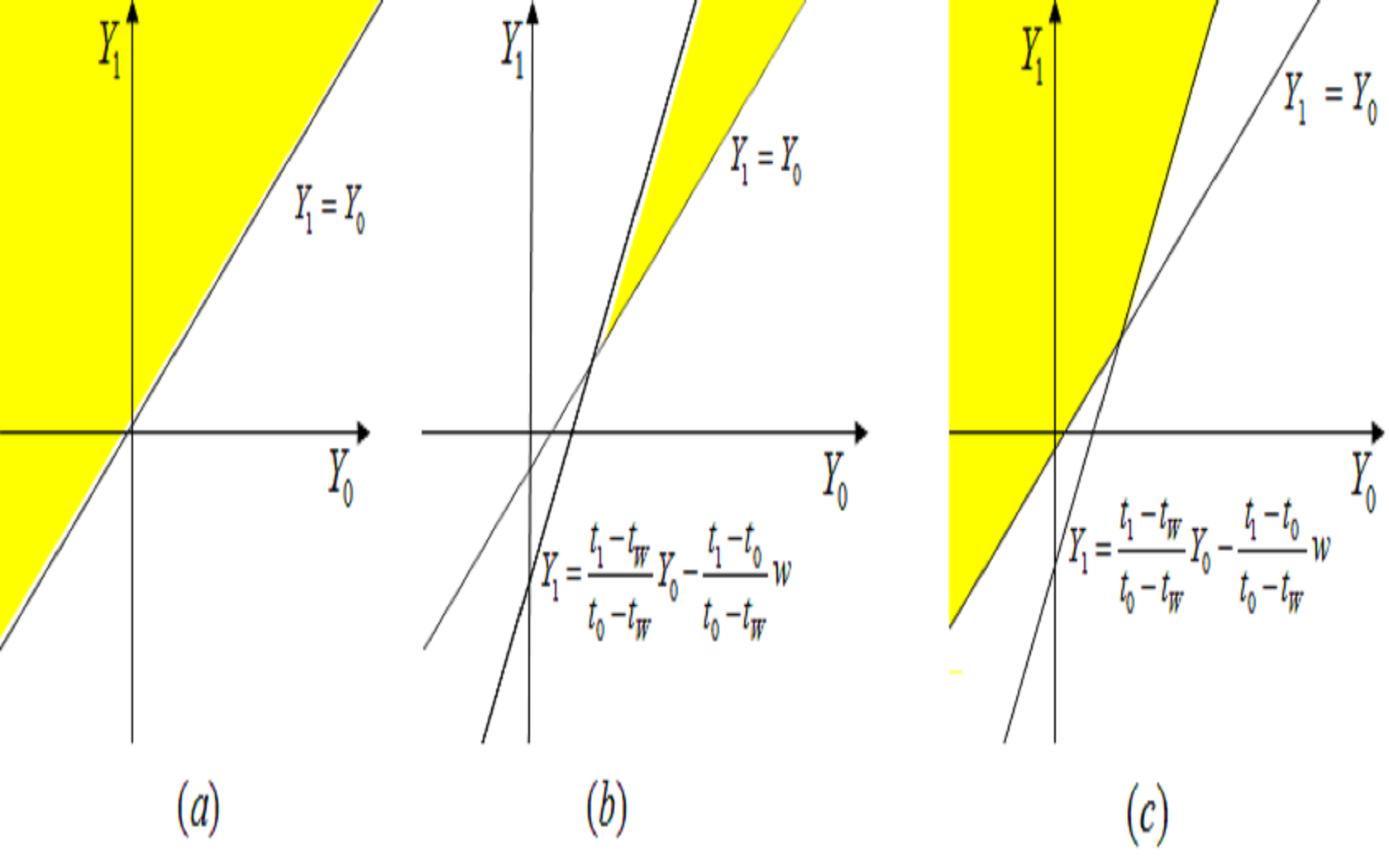}%
\caption{(a) MTR, (b) concave treatment response, (c) convex treatment
response}%
\label{shaperestrictions}%
\end{center}
\end{figure}

\begin{example}
(Monotone Treatment Response) MTR only requires that the potential outcomes be
weakly monotone in treatment with probability one:%
\[
\Pr\left(  Y_{1}\geq Y_{0}\right)  =1.
\]
MTR restricts the support of $(Y_{0},Y_{1})$ to the region above the straight
line $Y_{1}=Y_{0},$ as shown in Figure 1(a).
\end{example}%

\begin{figure}
[ptb]
\begin{center}
\includegraphics[
height=2.7259in,
width=6.6098in
]%
{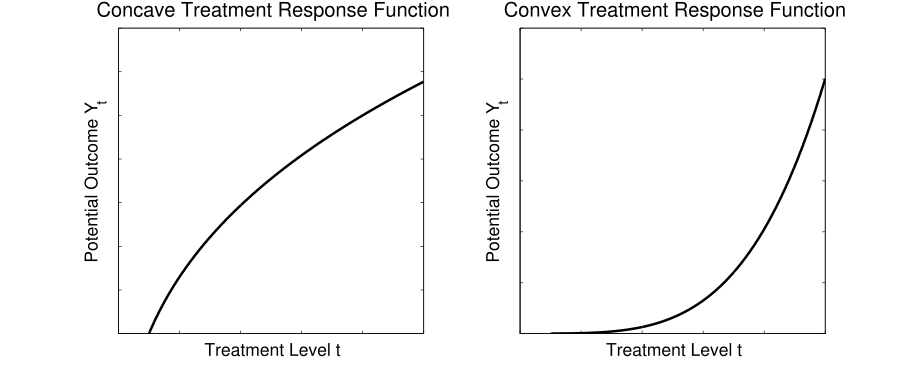}%
\caption{Concave treatment response and convex treatment response}%
\label{concaveconvex}%
\end{center}
\end{figure}

\begin{example}
(Concave/Convex Treatment Response) Consider panel data where the outcome
without treatment and an outcome either with the low-intensity treatment or
with the high-intensity treatment is observed for each
individual.\footnote{Various empirical studies are based on this structure,
e.g. Newhouse et al. (2007), Bandiera et al. (2008), and Suri (2011), among
others.} Let $W$ denote the observed outcome without treatment, while $Y_{0}$
and $Y_{1}$ denote potential outcomes under low-intensity treatment and
high-intensity treatment, respectively. Suppose that the treatment response
function is nondecreasing and that either $\left(  W,Y_{0}\right)  $ or
$\left(  W,Y_{1}\right)  $ is observed for each individual. Concavity and
convexity of the treatment response function imply $\Pr\left(  \frac{Y_{0}%
-W}{t_{0}-t_{W}}\geq\frac{Y_{1}-Y_{0}}{t_{1}-t_{0}},Y_{1}\geq Y_{0}\geq
W\right)  =1$ and $\Pr\left(  \frac{Y_{0}-W}{t_{0}-t_{W}}\geq\frac{Y_{1}%
-Y_{0}}{t_{1}-t_{0}},Y_{1}\geq Y_{0}\geq W\right)  =1$, respectively, where
$t_{d}$ is a level of input for each treatment status $d\in\left\{
0,1\right\}  $ while $t_{W}$\ is a level of input without the treatment and
$t_{W}<t_{0}<t_{1}$. \ Given $W=w,$ concavity and convexity of the treatment
response function restrict the support of $(Y_{0},Y_{1})$ to the region below
the straight line $Y_{1}=\frac{t_{1}-t_{W}}{t_{0}-t_{W}}Y_{0}-\frac
{t_{1}-t_{0}}{t_{0}-t_{W}}w$ and above the straight line $Y_{1}=Y_{0}$, and to
the region above two straight lines $Y_{1}=\frac{t_{1}-t_{W}}{t_{0}-t_{W}%
}Y_{0}-\frac{t_{1}-t_{0}}{t_{0}-t_{W}}w$ and $Y_{1}=Y_{0},$ respectively, as
shown in Figures 1(b) and (c).
\end{example}

\begin{example}
(Roy Model) In the Roy model, individuals self-select into treatment when
their benefits from the treatment are greater than nonpecuniary costs for
treatment participation. The extended Roy model assumes that the nonpecuniary
cost is deterministic with the following selection equation:%
\[
D=\boldsymbol{1}\left\{  Y_{1}-Y_{0}\geq\mu_{C}\left(  Z\right)  \right\}  ,
\]
where $\mu_{C}\left(  Z\right)  $ represents nonpecuniary costs with a vector
of observables $Z$. Then treated $(D=1)$ and untreated people $(D=0)$ are the
observed groups satisfying support restrictions $\left\{  Y_{1}-Y_{0}\geq
\mu_{C}\left(  Z\right)  \right\}  $ and $\left\{  Y_{1}-Y_{0}<\mu_{C}\left(
Z\right)  \right\}  $, respectively.
\end{example}

\begin{example}
(DTE conditional on Potential Outcomes) The conditional DTE for the
unobservable subgroup whose potential outcomes belong to a certain set $C$ is
written as
\[
\Pr\left\{  Y_{1}-Y_{0}\leq\delta|\left(  Y_{0},Y_{1}\right)  \in C\right\}  .
\]
For example, the distribution of the college premium for people whose
potential wage without college degrees is less than or equal to $\theta$ can
be written as
\[
\Pr\left\{  Y_{1}-Y_{0}\leq\delta|Y_{0}\leq\theta\right\}  ,
\]
where $Y_{0}$ and $Y_{1}$ denote the potential wage without and with college
degrees, respectively.
\end{example}

\subsection{DTE\ Bounds without Support Restrictions}

Prior to considering support restrictions, I\ briefly discuss bounds on the
DTE given marginal distributions without those restrictions.

\begin{lemma}
\label{Makarov}\textit{(Makarov (1981)}) Let
\begin{align*}
F_{\Delta}^{L}\left(  \delta\right)   & =\sup_{y}\max\left(  F_{1}\left(
y\right)  -F_{0}\left(  y-\delta\right)  ,0\right)  ,\\
F_{\Delta}^{U}\left(  \delta\right)   & =1+\inf_{y}\min\left(  F_{1}\left(
y\right)  -F_{0}\left(  y-\delta\right)  ,0\right)  .
\end{align*}
Then for any $\delta\in\mathbb{R},$
\[
F_{\Delta}^{L}\left(  \delta\right)  \leq F_{\Delta}\left(  \delta\right)
\leq F_{\Delta}^{U}\left(  \delta\right)  ,
\]
and both $F_{\Delta}^{L}\left(  \delta\right)  $ and $F_{\Delta}^{U}\left(
\delta\right)  $ are sharp.
\end{lemma}

Henceforth, I\ call these bounds Makarov bounds. One way to bound the DTE is
to use joint distribution bounds since the DTE can be obtained from the joint
distribution. When the marginal distributions of $Y_{0}$ and $Y_{1}$ are
given, Fr\'{e}chet inequalities provide some information on their unknown
joint distribution as follows: for any measurable sets $A_{0}$ and $A_{1}$\ in
$\mathbb{R}$,
\[
\max\left\{  \mu_{0}\left(  A_{0}\right)  +\mu_{1}\left(  A_{1}\right)
-1,0\right\}  \leq\pi\left(  A_{0}\times A_{1}\right)  \leq\min\left\{
\mu_{0}\left(  A_{0}\right)  ,\mu_{1}\left(  A_{1}\right)  \right\}  .
\]
Consider the event $\left\{  Y_{0}\in A_{0},Y_{1}\in A_{1}\right\}  $ for any
interval $A_{d}=[a_{d},b_{d}]$ with $a_{d}<b_{d}$ \ and $d\in\left\{
0,1\right\}  .$ In \ Figure 3, $\pi\left(  A_{0}\times A_{1}\right)
$ corresponds to the probability of the shaded \textit{rectangular} region in
the support space of $\left(  Y_{0},Y_{1}\right)  .\footnote{If $A_{0}$ and
$A_{1}$ are given as the unions of multiple intervals, $\left\{  Y_{0}\in
A_{0},Y_{1}\in A_{1}\right\}  $ would correspond to multiple
\textit{rectangular} regions.}$ Note that since marginal distributions are
defined in the one dimensional space, they are informative on the joint
distribution for rectangular regions in the two-dimensional support
space of $\left(  Y_0,Y_1\right)  $, as illustrated in Figure 3.%
\begin{figure}
[ptb]
\begin{center}
\includegraphics[
natheight=3.301800in,
natwidth=4.416600in,
height=2.3385in,
width=3.1185in
]%
{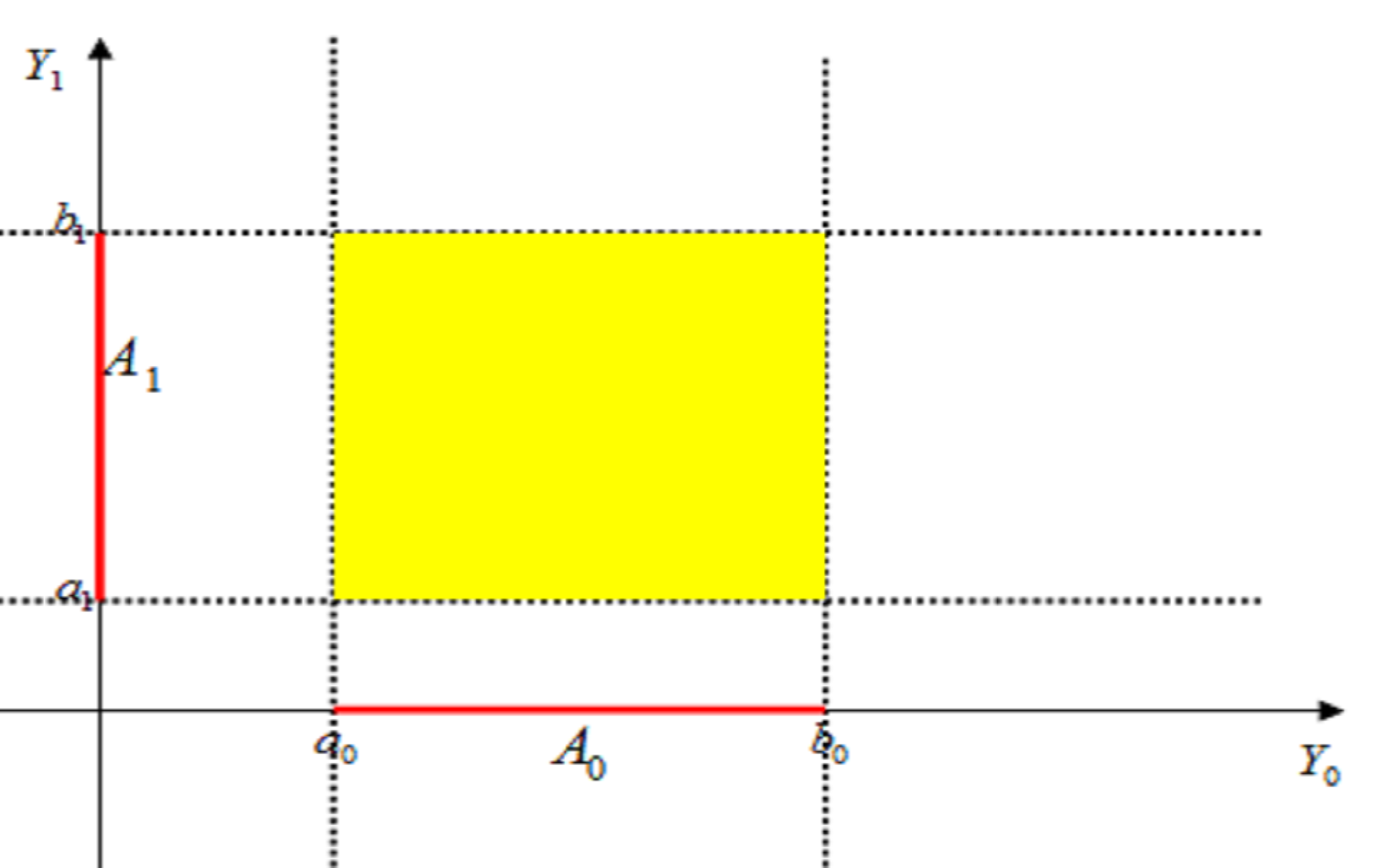}%
\caption{$\{Y_0\in A_0,Y_1\in A_1\}$}%
\label{A0A1}%
\end{center}
\end{figure}

Graphically, the DTE corresponds to the region below the straight line
$Y_{1}=Y_{0}+\delta$ in the support space as shown in Figure
4. Since the given marginal distributions are informative
on the joint distribution for rectangular regions in the support space, one
can bound the DTE by considering two rectangles $\left\{  Y_{0}\geq
y-\delta,Y_{1}\leq y\right\}  $ and $\left\{  Y_{0}<y^{\prime}-\delta
,Y_{1}>y^{\prime}\right\}  $ for any $\left(  y,y^{\prime}\right)
\in\mathbb{R}^{2}.$ Although the probability of each rectangle is not
point-identified, it can be bounded by Fr\'{e}chet inequalities.\footnote{Note
that Fr\'{e}chet lower bounds on $\Pr\left\{  Y_{0}\geq y^{\prime}%
-\delta,Y_{1}\leq y^{\prime}\right\}  $ and $\Pr\left\{  Y_{0}<y^{\prime
}-\delta,Y_{1}>y^{\prime}\right\}  $\ are sharp. They are both achieved when
$Y_{0}$ and $Y_{1}$ are perfectly positively dependent.} Since the DTE is
bounded from below by the Fr\'{e}chet lower bound on $\Pr\left\{  Y_{0}\geq
y-\delta,Y_{1}\leq y\right\}  $ for any $y\in\mathbb{R},$ the lower bound on
the DTE is obtained as follows:
\[
\underset{y}{\sup}\max\left(  F_{1}\left(  y\right)  -F_{0}\left(
y-\delta\right)  ,0\right)  \leq F_{\Delta}\left(  \delta\right)  .
\]
Similarly, the DTE is bounded from above by $1-\Pr\left\{  Y_{0}<y^{\prime
}-\delta,Y_{1}>y^{\prime}\right\}  $ for any $y^{\prime}\in\mathbb{R}$.
Therefore, the upper bound on the DTE is obtained by the Fr\'{e}chet lower
bound on $\Pr\left\{  Y_{0}<y^{\prime}-\delta,Y_{1}>y^{\prime}\right\}  $ as
follows:
\[
F_{\Delta}\left(  \delta\right)  \leq1-\sup_{y}\max\left(  F_{0}\left(
y-\delta\right)  -F_{1}\left(  y\right)  ,0\right)  .
\]
Makarov (1981) proved that those lower and upper bounds are
sharp.\footnote{One may wonder if multiple rectangles below $Y_{1}%
=Y_{0}+\delta$ that overlap one another could yield the more improved lower
bound. However, if the Fr\'{e}chet lower bound on another rectangle $\left\{
Y_{0}\geq y^{\prime\prime}-\delta,Y_{1}\leq y^{\prime\prime}\right\}  $ is
added and the Fr\'{e}chet upper bound on the intersection of the two
rectangles is subtracted, it is smaller than or equal to the lower bound
obtained from the only one rectangle.}%

\begin{figure}
[ptb]
\begin{center}
\includegraphics[
natheight=4.728800in,
natwidth=7.604300in,
height=2.3929in,
width=3.8303in
]%
{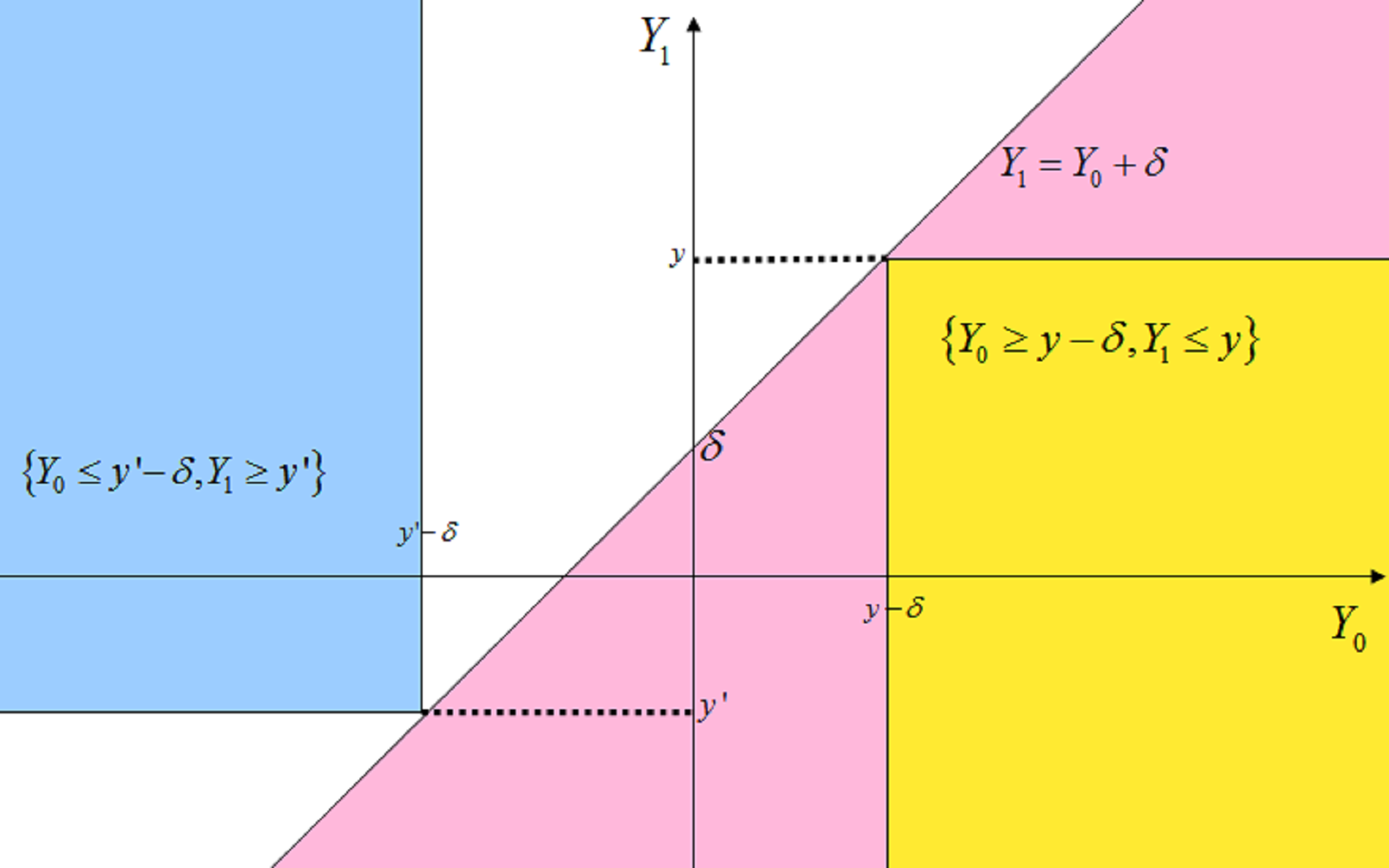}%
\caption{Makarov bounds}%
\label{Makarovbounds}%
\end{center}
\end{figure}

If the marginal distributions of $Y_{0}$ and $Y_{1}$ are both absolutely
continuous with respect to the Lebesgue measure on $\mathbb{R}$, then the
Makarov upper bound and lower bound are achieved when $F\left(  y_{0}%
,y_{1}\right)  =C_{s}^{L}\left(  F_{0}\left(  y_{0}\right)  ,F_{1}\left(
y_{1}\right)  \right)  $ and when $F\left(  y_{0},y_{1}\right)  =C_{t}%
^{U}\left(  F_{0}\left(  y_{0}\right)  ,F_{1}\left(  y_{1}\right)  \right)  $
respectively, where%
\begin{align*}
s  & =F_{\Delta}^{U}\left(  \delta\right)  \text{ and }t=F_{\Delta}^{L}\left(
\delta^{-}\right)  ,\\
C_{s}^{U}\left(  u,v\right)   & =\left\{
\begin{array}
[c]{ll}%
\min\left(  u+s-1,v\right)  , & \text{ }1-s\leq u\leq1,0\leq v\leq s,\\
\max\left(  u+v-1,0\right)  , & \text{ elsewhere,}%
\end{array}
\right. \\
C_{t}^{L}\left(  u,v\right)   & =\left\{
\begin{tabular}
[c]{ll}%
$\min\left(  u,v-t\right)  ,$ & $0\leq u\leq1-t,t\leq v\leq1,$\\
$\max\left(  u+v-1,0\right)  ,$ & elsewhere.
\end{tabular}
\right.
\end{align*}
Note that both $C_{s}^{U}\left(  u,v\right)  $ and $C_{t}^{L}\left(
u,v\right)  $ depend on $\delta$, through $s$ and $t$,
respectively.\footnote{To be precise, when the distribution of $Y_{1}-Y_{0}$
is discontinuous, the Makarov lower bound is attained only for the left limit
of the DTE. That is, $F_{\Delta}\left(  \delta^{-}\right)  =F_{\Delta}%
^{L}\left(  \delta^{-}\right)  =$ $t$ under $C_{t}^{L}$, while under
$C_{s}^{U},$ $F_{\Delta}\left(  \delta\right)  =F_{\Delta}^{U}\left(
\delta\right)  =$ $s$ for the right-continuous distribution function
$F_{\Delta}$. Note that even if both marginal distributions of $Y_{1}$ and
$Y_{0}$ are continuous, the distribution of $Y_{1}-Y_{0}$ may not be
continuous. Hence, typically the lower bound on the DTE is established only
for the left limit of the DTE $\Pr\left[  Y_{1}-Y_{0}<\delta\right]  .$ See
Nelsen (2006)\ for details.} Since the joint distribution achieving Makarov
bounds varies with $\delta,$ Makarov bounds are only \textit{pointwise} sharp,
not \textit{uniformly}.\ To address this issue, Firpo and Ridder (2008)
proposed joint bounds on the DTE for multiple values of $\delta$, which are
tighter than Makarov bounds. However, their improved bounds are not sharp and
sharp bounds on the functional $F_{\Delta}$ are an open question. For details,
see Frank et al. (1997), Nelsen (2006) and Firpo and Ridder (2008).%
\begin{figure}
[ptb]
\begin{center}
\includegraphics[
natheight=4.437300in,
natwidth=7.990000in,
height=2.2468in,
width=4.0222in
]%
{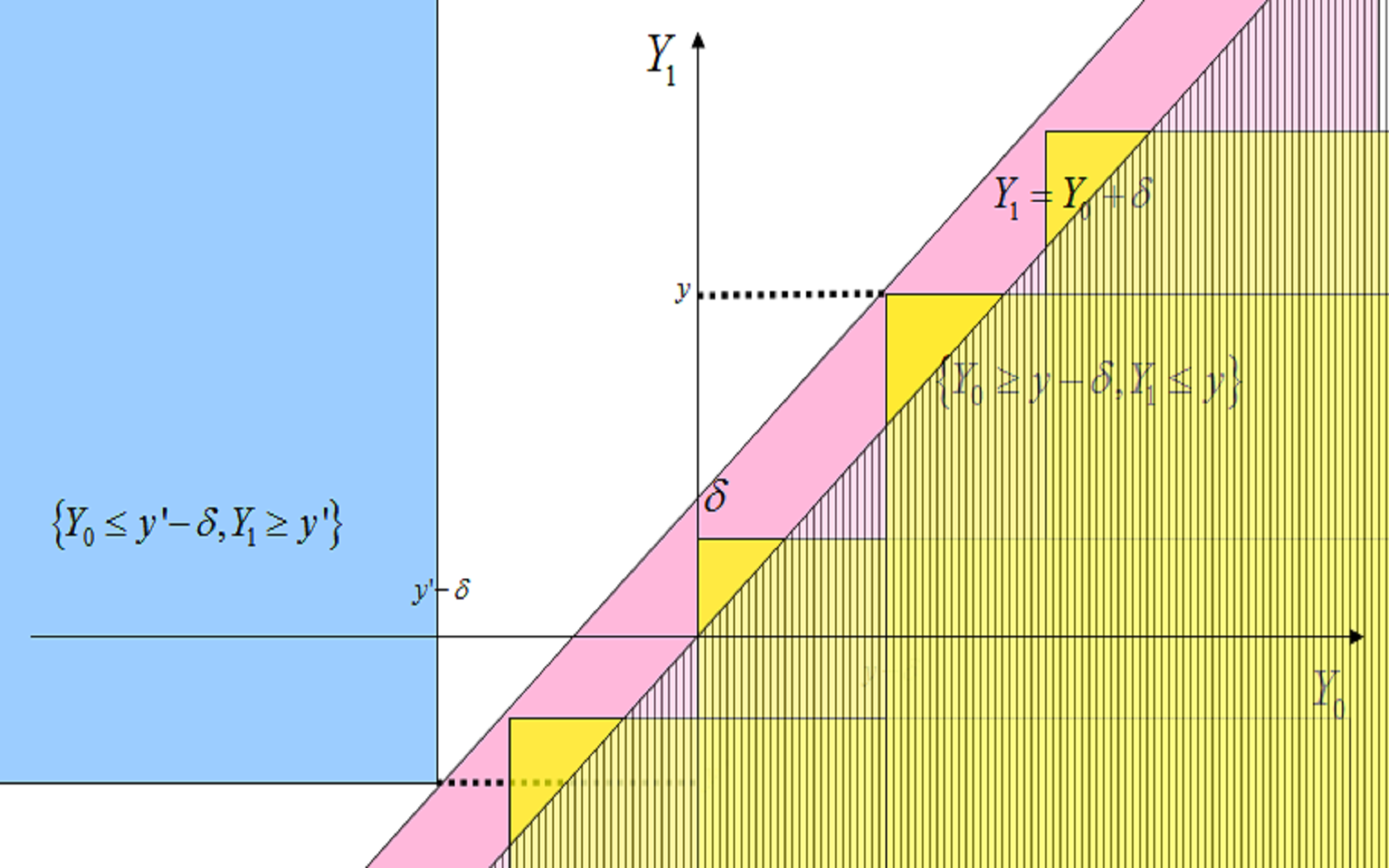}%
\caption{Makarov bounds are not best possible under MTR}%
\label{underMTR}%
\end{center}
\end{figure}

Although Makarov bounds are sharp when no other restrictions are imposed, they
are often too wide to be informative in practice and not sharp in the presence
of additional restrictions on the set of possible pairs of potential outcomes.
Figure 5 illustrates that if the support is restricted to the
region above the straight line $Y_{1}=Y_{0}$ by MTR, the Makarov lower bound
is not the best possible anymore. The lower bound can be improved under MTR
because MTR allows \emph{multiple} mutually exclusive rectangles to be placed
below the straight line $Y_{1}=Y_{0}+\delta$.

Methods of establishing sharp bounds under this class of restrictions and
fixed marginal distributions have remained unanswered in the literature. The
central difficulty lies in finding out the particular joint distributions
achieving sharp bounds among all joint distributions that have the given
marginal distributions and satisfy support restrictions. The next subsection
shows that an optimal transportation approach circumvents this difficulty
through its dual formulation.

\subsection{Optimal Transportation Approach}

An optimal transportation problem was first formulated by Monge (1781) who
studied the most efficient way to move a given distribution of mass to another
distribution in a different location. Much later Monge's problem was
rediscovered and developed by Kantorovich. The optimal transportation problem
of Monge-Kantorovich type is posed as follows. Let $c\left(  y_{0}%
,y_{1}\right)  $ be a nonnegative lower semicontinuous function on
$\mathbb{R}^{2}$ and define $\Pi\left(  \mu_{0},\mu_{1}\right)  $ to be the
set of joint distributions on $\mathbb{R}^{2}$ that have $\mu_{0}$ and
$\mu_{1}$ as marginal distributions. The optimal transportation problem solves%
\begin{equation}
\underset{\pi\in\Pi\left(  \mu_{0},\mu_{1}\right)  }{\inf}\int c\left(
y_{0},y_{1}\right)  d\pi.\label{2.0}%
\end{equation}
The objective function in the minimization problem is linear in the joint
distribution $\pi$ and the constraint is that the joint distribution $\pi$
should have fixed marginal distributions $\mu_{0}$ and $\mu_{1}$. $c\left(
y_{0},y_{1}\right)  $ and $\int c\left(  y_{0},y_{1}\right)  d\pi$ are called
the \emph{cost function} and the \emph{total cost}, respectively. Kantorovich
(1942) developed a dual formulation for the problem (2), which is a
key feature of the optimal transportation approach.

\begin{lemma}
\label{duality} (Kantorovich duality) Let $c:\mathbb{R}\times\mathbb{R}%
\rightarrow\left[  0,\infty\right]  $ be a lower semicontinuous function and
$\Phi_{c}$ the set of all functions $\left(  \varphi,\psi\right)  \in
L^{1}\left(  d\mu_{0}\right)  $ $\times L^{1}\left(  d\mu_{1}\right)  $ with
\begin{equation}
\varphi\left(  y_{0}\right)  +\psi\left(  y_{1}\right)  \leq c\left(
y_{0},y_{1}\right) \label{2.1}%
\end{equation}
Then,%
\begin{equation}
\underset{\pi\in\Pi\left(  \mu_{0},\mu_{1}\right)  }{\inf}\int c\left(
y_{0},y_{1}\right)  d\pi=\underset{\left(  \varphi,\psi\right)  \in\Phi_{c}%
}{\sup}\left(  \int\varphi\left(  y_{0}\right)  d\mu_{0}+\int\psi\left(
y_{1}\right)  d\mu_{1}\right)  .\label{2.2}%
\end{equation}
Also, the infimum in the left-hand side of (4) and the supremum in the
right-hand side of (4) are both attainable, and the value of the
supremum in the right-hand side does not change if one restricts $\left(
\varphi,\psi\right)  $ to be bounded and continuous.
\end{lemma}

\begin{remark}
\label{infinitec} Note that the cost function $c\left(  y_{0},y_{1}\right)  $
may be infinite for some $\left(  y_{0},y_{1}\right)  \in\mathbb{R}^{2}.$
Since $c$ is a nonnegative function, the integral $\int c\left(  y_{0}%
,y_{1}\right)  d\pi\in\left[  0,\infty\right]  $ is well-defined.
\end{remark}

This dual formulation provides a key to solve the optimization problem
(2); I can overcome the difficulty associated with picking the
maximizer joint distribution in the set $\Pi\left(  \mu_{0},\mu_{1}\right)  $
by solving optimization with respect to given marginal distributions. The dual
functions $\varphi\left(  y_{0}\right)  $ and $\psi\left(  y_{1}\right)  $ are
Lagrange multipliers corresponding to the constraints $\pi\left(  y_{0}%
\times\mathbb{R}\right)  =\mu_{0}\left(  y_{0}\right)  $ and $\pi\left(
\mathbb{R}\times y_{1}\right)  =\mu_{1}\left(  y_{1}\right)  ,$ respectively,
for each $y_{0}$ and $y_{1}$ in $\mathcal{Y}_{0}$ and $\mathcal{Y}_{1}$.
Henceforth they are both assumed to be bounded and continuous without loss of
generality. By the condition (3), each pair $\left(  \varphi
,\psi\right)  $ in $\Phi_{c}$ satisfies
\begin{align}
\varphi\left(  y_{0}\right)   & \leq\underset{y_{1}\in\mathbb{R}}{\inf
}\left\{  c\left(  y_{0},y_{1}\right)  -\psi\left(  y_{1}\right)  \right\}
,\label{2.3}\\
\psi\left(  y_{1}\right)   & \leq\underset{y_{0}\in\mathbb{R}}{\inf}\left\{
c\left(  y_{0},y_{1}\right)  -\varphi\left(  y_{0}\right)  \right\}
.\nonumber
\end{align}
At the optimum for $\left(  y_{0},y_{1}\right)  $\ in the support of the
optimal joint distribution, the inequality in (3) holds with equality
and there exists a pair of dual functions $\left(  \varphi,\psi\right)  $ that
satisfies both inequalities in (5) with equalities.

In recent years, this dual formulation has turned out to be powerful and
useful for various problems related to the equilibrium and decentralization in
economics. See Ekeland (2005, 2010), Carlier (2010), Chiappori et al. (2010),
Chernozhukov et al. (2010), and Galichon and Salani\'{e} (2012). In
econometrics, Galichon and Henry (2009) and Ekeland et al. (2010) showed that
the dual formulation yields a test statistic for a set of theoretical
restrictions in partially identified economic models. They set the cost
function as an indicator for incompatibility of the structure with the data
and derived a \ Kolmogorov Smirnov type test statistic from a well known dual
representation theorem; see Lemma 3 below. Similarly, Galichon and
Henry (2011) showed that the identified set of structural parameters in game
theoretic models with pure strategy equilibria can be formulated as an optimal
transportation problem using the $\{0,1\}$-valued cost function.

Establishing sharp bounds on the DTE is also an optimal transportation problem
with an indicator function as the cost function. The DTE can be written as the
integration of an indicator function with respect to the joint distribution
$\pi$ as follows:%
\[
F_{\Delta}\left(  \delta\right)  =\Pr\left(  Y_{1}-Y_{0}<\delta\right)
=\int\boldsymbol{1}\left\{  y_{1}-y_{0}<\delta\right\}  d\pi.
\]
Since marginal distributions of potential outcomes are given as $\mu_{0}$ and
$\mu_{1},$ establishing sharp bounds reduces to picking a particular joint
distribution maximizing or minimizing the DTE from all possible joint
distributions having $\mu_{0}$ and $\mu_{1}$ as their marginal distributions.
Then the DTE is bounded as follows:
\[
\underset{\pi\in\Pi\left(  \mu_{0},\mu_{1}\right)  }{\inf}\int\boldsymbol{1}%
\left\{  y_{1}-y_{0}<\delta\right\}  d\pi\leq F_{\Delta}\left(  \delta\right)
\leq\underset{\pi\in\Pi\left(  \mu_{0},\mu_{1}\right)  }{\sup}\int
\boldsymbol{1}\left\{  y_{1}-y_{0}\leq\delta\right\}  d\pi,
\]
where $\Pi\left(  \mu_{0},\mu_{1}\right)  $ is the set of joint distributions
that have $\mu_{0}$ and $\mu_{1}$ as marginal distributions. For the indicator
function, the Kantorovich duality lemma for $\{0,1\}-$valued costs in Villani
(2003) can be applied as follows:

\begin{lemma}
\label{01}(Kantorovich duality for $\{0,1\}$-valued costs) The sharp lower
bound on the DTE has the following dual representation:%
\begin{align}
& \underset{\pi\in\Pi\left(  \mu_{0},\mu_{1}\right)  }{\inf}\int
\boldsymbol{1}\left\{  y_{1}-y_{0}<\delta\right\}  d\pi\label{2.4}\\
& =\sup_{A\subset\mathbb{R}}\left\{  \mu_{0}\left(  A\right)  -\mu_{1}\left(
A^{D}\right)  ;\text{ }A\text{ is closed}\right\} \nonumber
\end{align}
where%
\[
A^{D}=\left\{  y_{1}\in\mathbb{R}|\exists y_{0}\in A\text{ s.t. }y_{1}%
-y_{0}\geq\delta\right\}  .
\]
Similarly, the sharp upper bound on the DTE can be written as follows:%
\begin{align*}
& \underset{\pi\in\Pi\left(  \mu_{0},\mu_{1}\right)  }{\sup}\int
\boldsymbol{1}\left\{  y_{1}-y_{0}\leq\delta\right\}  d\pi\\
& =1-\underset{F\in\Pi\left(  F_{0},F_{1}\right)  }{\inf}\int\boldsymbol{1}%
\left\{  y_{1}-y_{0}>\delta\right\}  d\pi\\
& =1-\sup_{A\subset\mathbb{R}}\left\{  \mu_{0}\left(  A\right)  -\mu
_{1}\left(  A^{E}\right)  ;\text{ }A\text{ is closed}\right\}
\end{align*}
where
\[
A^{E}=\left\{  y_{1}\in\mathbb{R}|\exists y_{0}\in A\text{ s.t. }y_{1}%
-y_{0}\leq\delta\right\}  .
\]

\end{lemma}

\begin{proof}
See pp. $44-46$ of Villani (2003).
\end{proof}

In the following discussion, I\ focus on the lower bound on the DTE since the
procedure to obtain the upper bound is similar.

\begin{remark}
\label{superlevel} In the proof of Lemma 3, Villani (2003) showed that
at the optimum, $A=\left\{  x\in\mathbb{R}|\varphi\left(  x\right)  \geq
s\right\}  $ for some $s\in\left[  0,1\right]  $. Since the function $\varphi$
is continuous, if $\varphi$ is nondecreasing then $A=[a,\infty)$ for some
$a\in\left[  \mathbb{-\infty},\mathbb{\infty}\right]  $ where $A=\phi$ if
$a=\mathbb{\infty}.$ In contrast, if $\varphi$ is nonincreasing, then
$A=(-\infty,a]$ where $A=\phi$ if $a=-\mathbb{\infty}$
\end{remark}

Remember that for any $\left(  y_{0},y_{1}\right)  $ in the support of the
optimal joint distribution, $\varphi$ and $\psi$ satisfy%
\begin{equation}
\varphi\left(  y_{0}\right)  =\underset{y_{1}\in\mathbb{R}}{\inf}\left\{
\boldsymbol{1}\left\{  y_{1}-y_{0}<\delta\right\}  -\psi\left(  y_{1}\right)
\right\}  .\label{2.4.5}%
\end{equation}
Pick $\left(  y_{0}^{\prime},y_{1}^{\prime}\right)  $ and $\left(
y_{0}^{\prime\prime},y_{1}^{\prime\prime}\right)  $\ with $y_{0}^{\prime
\prime}>y_{0}^{\prime}$\ in the support of the optimal joint distribution.
Then,%
\begin{align}
\varphi\left(  y_{0}^{\prime}\right)   & =\boldsymbol{1}\left\{  y_{1}%
^{\prime}-y_{0}^{\prime}<\delta\right\}  -\psi\left(  y_{1}^{\prime}\right)
\label{2.4.6}\\
& \leq\boldsymbol{1}\left\{  y_{1}^{\prime\prime}-y_{0}^{\prime}%
<\delta\right\}  -\psi\left(  y_{1}^{\prime\prime}\right) \nonumber\\
& \leq\boldsymbol{1}\left\{  y_{1}^{\prime\prime}-y_{0}^{\prime\prime}%
<\delta\right\}  -\psi\left(  y_{1}^{\prime\prime}\right) \nonumber\\
& =\varphi\left(  y_{0}^{\prime\prime}\right)  .\nonumber
\end{align}

The inequality in the second line of (8) is obvious from
(7) and the inequality in the third line of (8) holds
because $\boldsymbol{1}\left\{  y_{1}-y_{0}<\delta\right\}  $ is nondecreasing
in $y_{0}$. Since $\varphi$ is nondecreasing on the set $\left\{  y_{0}%
\in\mathcal{Y}_{0}|\exists y_{1}\in\mathcal{Y}_{1}\text{ s.t. }\left(
y_{0},y_{1}\right)  \in Supp\left(  \pi\right)  \right\}  $, by Remark
2 $A$ can be written as $[a,\infty)$ for some $a\in\left[
\mathbb{-\infty},\mathbb{\infty}\right]  .$%

\begin{figure}
[ptb]
\begin{center}
\includegraphics[
natheight=3.135800in,
natwidth=4.750400in,
height=2.2234in,
width=3.352in
]%
{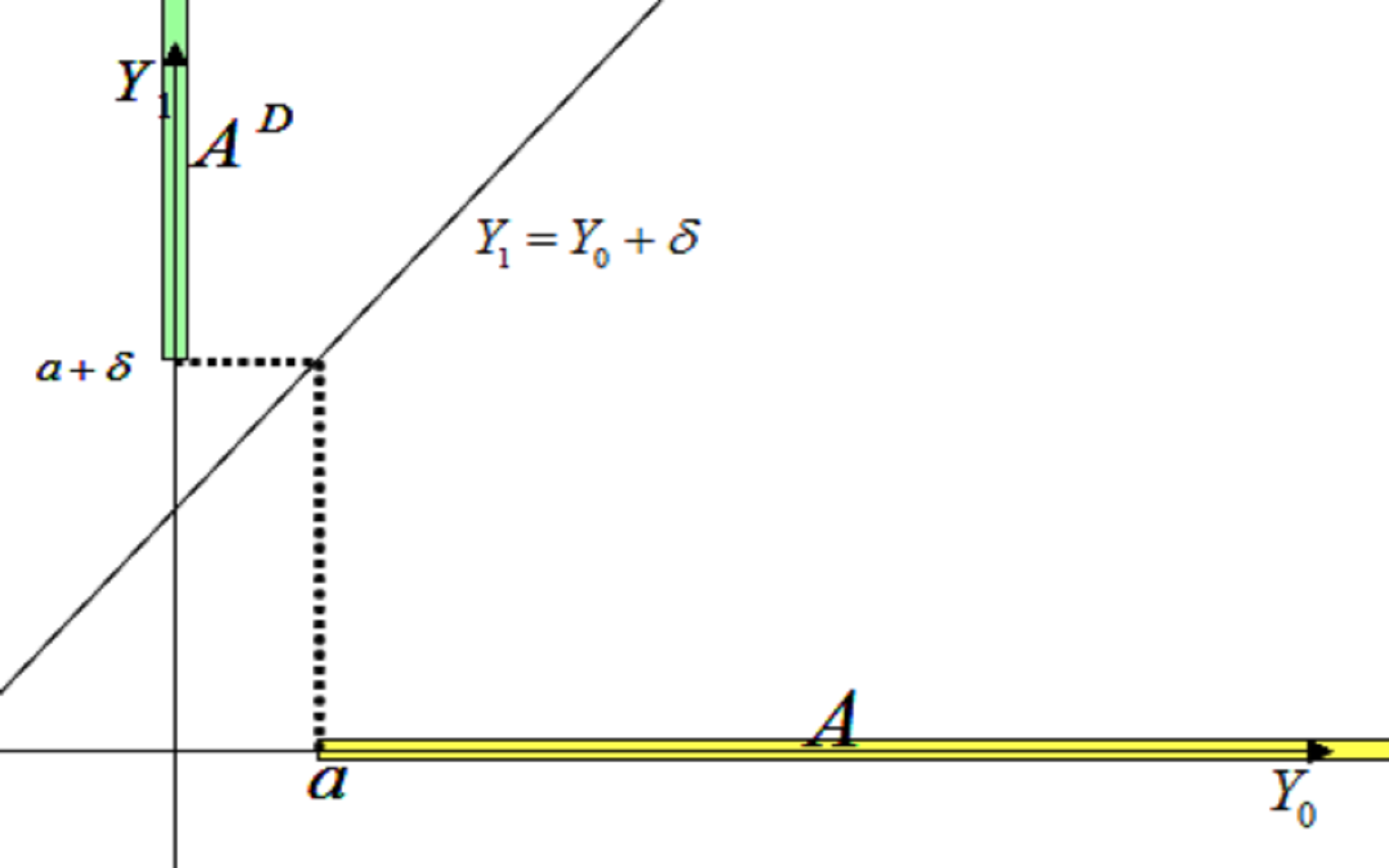}%
\caption{$A^{D}$ for $A=[a,\infty)$}%
\label{AAD}%
\end{center}
\end{figure}

As shown in Figure 6, $A^{D}=\phi$ for $A=\phi,$ and $A^{D}%
=[a+\delta,\infty)$ for $A=[a,\infty)$ with $a\in\left(  -\infty
,\infty\right)  $. Then, $\mu_{0}\left(  A\right)  -\mu_{1}\left(
A^{D}\right)  =0$ for $A=\phi$, while $\mu_{0}\left(  A\right)  -\mu
_{1}\left(  A^{D}\right)  =F_{1}\left(  a+\delta\right)  -F_{0}\left(
a\right)  $ for $A=[a,\infty).$ Therefore, the RHS in (6) reduces to%
\[
\underset{a\in\mathbb{R}}{\sup}\max\left[  F_{1}\left(  a+\delta\right)
-F_{0}\left(  a\right)  ,0\right]  ,
\]
which is equal to the Makarov lower bound. One can derive the Makarov upper
bound in the same way.

Now consider the support restriction $\Pr\left(  \left(  Y_{0},Y_{1}\right)
\in C\right)  =1$. Note that this restriction is linear in the entire joint
distribution $\pi,$ since it can be rewritten as $\int\boldsymbol{1}%
_{C}\left(  y_{0},y_{1}\right)  d\pi=1$. The linearity makes it possible to
handle this restriction with penalty. In particular, since support
restrictions hold with probability one, the corresponding penalty is infinite.
Therefore, one can embed $1-\boldsymbol{1}_{C}\left(  y_{0},y_{1}\right)  $
into the cost function with an infinite multiplier $\lambda=\infty$ as
follows:%
\begin{equation}
\underset{\pi\in\Pi\left(  \mu_{0},\mu_{1}\right)  }{\inf}\int\left\{
\boldsymbol{1}\left\{  y_{1}-y_{0}<\delta\right\}  +\lambda\left(
1-\boldsymbol{1}_{C}\left(  y_{0},y_{1}\right)  \right)  \right\}
d\pi\label{3.1}%
\end{equation}
The minimization problem (9) is well defined with $\lambda=\infty$ as
noted in Remark 1. Note that for $\lambda=\infty,$ any joint
distribution which violates the restriction $\Pr\left(  \left(  Y_{0}%
,Y_{1}\right)  \in C\right)  =1$ would cause infinite total costs in
(9) and it is obviously excluded from the potential optimal joint
distribution candidates. The optimal joint distribution should thus satisfy
the restriction $\Pr\left(  \left(  Y_{0},Y_{1}\right)  \in C\right)  =1$ to
avoid infinite costs by not permitting any positive probability density for
the region\ outside of the set $C$. Similarly, the upper bound on the DTE is
written as%
\begin{align}
& \underset{\pi\in\Pi\left(  \mu_{0},\mu_{1}\right)  }{\sup}\int\left\{
\boldsymbol{1}\left\{  y_{1}-y_{0}\leq\delta\right\}  -\lambda\left(
1-\boldsymbol{1}_{C}\left(  y_{0},y_{1}\right)  \right)  \right\}
d\pi\label{3.3.3.4}\\
& =1-\underset{\pi\in\Pi\left(  \mu_{0},\mu_{1}\right)  }{\inf}\int\left\{
\boldsymbol{1}\left\{  y_{1}-y_{0}>\delta\right\}  +\lambda\left(
1-\boldsymbol{1}_{C}\left(  y_{0},y_{1}\right)  \right)  \right\}
d\pi.\nonumber
\end{align}

To the best of my knowledge, this is the first paper that allows for $\left\{
0,1,\infty\right\}  $-valued costs. Although the econometrics literature based
on the optimal transportation approach has used Lemma 3 for
$\{0,1\}-$valued costs, the problem (9) cannot be solved using Lemma
3. In the next section, I\ develop a dual representation for
(9) in order to characterize sharp bounds on\ the DTE.

\section{Main Results}

This section characterizes sharp DTE bounds under general support restrictions
by developing a dual representation for problems (9) and
(10). I\ use this characterization to derive sharp DTE bounds for
various economic examples. Also, I provide intuition regarding improvement of
the identification region via graphical illustrations.

\subsection{Characterization}

The following theorem is the main result of the paper.

\begin{theorem}
\label{01infinity}\bigskip\ The sharp lower and upper bounds on the DTE under
$\Pr\left(  \left(  Y_{0},Y_{1}\right)  \in C\right)  =1$ are characterized as
follows: for any $\delta\in\mathbb{R},$%
\[
F_{\Delta}^{L}\left(  \delta\right)  \leq F_{\Delta}\left(  \delta\right)
\leq F_{\Delta}^{U}\left(  \delta\right)  ,
\]
where%
\begin{align}
F_{\Delta}^{L}\left(  \delta\right)   & =\underset{\left\{  A_{k}\right\}
_{k=-\infty}^{\infty}}{\sup}\sum\limits_{k=-\infty}^{\infty}\max\left\{
\mu_{0}\left(  A_{k}\right)  -\mu_{1}\left(  A_{k}^{C}\right)  ,0\right\}
,\label{dual01infinity}\\
F_{\Delta}^{U}\left(  \delta\right)   & =1-\underset{\left\{  B_{k}\right\}
_{k=-\infty}^{\infty}}{\sup}\sum\limits_{k=-\infty}^{\infty}\max\left\{
\mu_{0}\left(  B_{k}\right)  -\mu_{1}\left(  B_{k}^{C}\right)  ,0\right\}
,\nonumber
\end{align}
where%
\begin{align*}
& \left\{  A_{k}\right\}  _{k=-\infty}^{\infty}\text{ and }\left\{
B_{k}\right\}  _{k=-\infty}^{\infty}\text{ are both monotonically decreasing
sequences of open sets,}\\
A_{k}^{C}  & =%
\begin{tabular}
[c]{l}%
$\left\{  y_{1}\in\mathbb{R}|\exists y_{0}\in A_{k}\text{ s.t. }y_{1}%
-y_{0}\geq\delta\text{ and }\left(  y_{0},y_{1}\right)  \in C\right\}  $\\
$\cup\left\{  y_{1}\in\mathbb{R}|\exists y_{0}\in A_{k+1}\text{ s.t. }%
y_{1}-y_{0}<\delta\text{ and }\left(  y_{0},y_{1}\right)  \in C\right\}  $,
\end{tabular}
\\
B_{k}^{C}  & =%
\begin{tabular}
[c]{l}%
$\left\{  y_{1}\in\mathbb{R}|\exists y_{0}\in B_{k}\text{ s.t. }y_{1}%
-y_{0}\leq\delta\text{ and }\left(  y_{0},y_{1}\right)  \in C\right\}  $\\
$\cup\left\{  y_{1}\in\mathbb{R}|\exists y_{0}\in B_{k+1}\text{ s.t. }%
y_{1}-y_{0}>\delta\text{ and }\left(  y_{0},y_{1}\right)  \in C\right\}
\text{ for any integer }k.$%
\end{tabular}
\end{align*}

\end{theorem}

\begin{proof}
See Appendix A.
\end{proof}

Theorem 1 is obtained by applying Kantorovich duality in Lemma
2 to the optimal transportation problems (9) and
(10). Note that the sharpness of the bounds is also confirmed by
Lemma 2. Since characterization of the upper bound is similar to
that of the lower bound, I\ maintain the focus of the discussion on the lower
bound. The minimization problem (9) can be written in the dual
formulation as follows: for $\lambda=\infty,$
\begin{align*}
& \underset{\pi\in\Pi\left(  \mu_{0},\mu_{1}\right)  }{\inf}\int\left\{
\boldsymbol{1}\left\{  y_{1}-y_{0}<\delta\right\}  +\lambda\left(
1-\boldsymbol{1}_{C}\left(  y_{0},y_{1}\right)  \right)  \right\}  d\pi\\
& =\underset{\left(  \varphi,\psi\right)  \in\Phi_{c}}{\sup}\left(
\int\varphi\left(  y_{0}\right)  d\mu_{0}+\int\psi\left(  y_{1}\right)
d\mu_{1}\right)  ,
\end{align*}
where
\[
\Phi_{c}=\left\{  \left(  \varphi,\psi\right)  ;\text{ }\varphi\left(
y_{0}\right)  +\psi\left(  y_{1}\right)  \leq\boldsymbol{1}\left\{
y_{1}-y_{0}<\delta\right\}  +\lambda\left(  1-\boldsymbol{1}_{C}\left(
y_{0},y_{1}\right)  \right)  \text{ with }\lambda=\infty\right\}  .
\]
Note that at the optimum $\varphi\left(  y_{0}\right)  +\psi\left(
y_{1}\right)  =\boldsymbol{1}\left\{  y_{1}-y_{0}<\delta\right\}  $ for any
$\left(  y_{0},y_{1}\right)  $ in the support of the optimal joint
distribution. Therefore, dual functions $\varphi$ and $\psi$ can be written as
follows: for any $\left(  y_{0},y_{1}\right)  $ in the support of the optimal
joint distribution,%
\[
\varphi\left(  y_{0}\right)  =\underset{y_{1}:\left(  y_{0},y_{1}\right)  \in
C}{\inf}\left\{  \boldsymbol{1}\left\{  y_{1}-y_{0}<\delta\right\}
-\psi\left(  y_{1}\right)  \right\}  .
\]

In my proof of Theorem 1, $A_{k}$ is defined as $A_{k}=\left\{
x\in\mathbb{R}:\varphi(x)>s+k\right\}  $ for the function $\varphi$, some
$s\in\lbrack0,1],$ and each integer $k.$ Since the dual function $\varphi$ is
continuous, if $\varphi$ is nondecreasing then $A_{k} $ $=(a_{k},\infty)$ for
some $a_{k}\in\left[  -\infty,\infty\right]  .$ Note that $A_{k}=\phi$ for
$a_{k}=\infty.$ Also, since $\left\{  A_{k}\right\}  _{k=-\infty}^{\infty}$ is
a monotonically decreasing sequence of open sets, $a_{k}\leq a_{k+1}$ for
every integer $k.$ In contrast, if $\varphi$ is nonincreasing at the optimum
then $A_{k}=(-\infty,a_{k})$ for $a_{k}\in\left[  -\infty,\infty\right]  $ and
$a_{k+1}\leq a_{k}$ for each integer $k$. Note that $A_{k}=\phi$ for
$a_{k}=-\infty$. In the next subsection, I\ will show that the function
$\varphi$ is monotone for economic examples considered in this paper and that
sharp DTE bounds in each example are readily derived from monotonicity of
$\varphi$.

\begin{remark}
\label{continuity} (Robustness of the sharp bounds) My sharp DTE\ bounds are
robust for support restrictions in the sense that they do not rely too heavily
on the small deviation of the restriction. I\ can verify this by showing that
sharp bounds under $\Pr\left(  \left(  Y_{0},Y_{1}\right)  \in C\right)  \geq
p$ converge to those under $\Pr\left(  \left(  Y_{0},Y_{1}\right)  \in
C\right)  =1,$ as $p$ goes to one. The sharp lower bound under $\Pr\left(
\left(  Y_{0},Y_{1}\right)  \in C\right)  \geq p$ can be obtained with a
multiplier $\widetilde{\lambda}_{p}\geq0$ as follows:%
\begin{equation}
F_{\Delta}^{L}\left(  \delta\right)  =\underset{\pi\in\Pi\left(  \mu_{0}%
,\mu_{1}\right)  }{\inf}\int\left\{  \boldsymbol{1}\left\{  y_{1}-y_{0}%
<\delta\right\}  +\widetilde{\lambda}_{p}\left(  1-\boldsymbol{1}_{C}\left(
y_{0},y_{1}\right)  \right)  \right\}  d\pi.\label{con1}%
\end{equation}
Obviously, $\widetilde{\lambda}_{0}=0$. Furthermore, $\widetilde{\lambda}%
_{p}\leq\widetilde{\lambda}_{q}$ for $0\leq p<q\leq1$ since $F_{\Delta}%
^{L}\left(  \delta\right)  $ is nondecreasing in $p.$ The proof of Theorem
1 can be easily adapted to the more general case in which the
multiplier is given as a positive integer. If $\widetilde{\lambda}_{p}=2K$ in
(12)\ for some positive integer $K$, then the dual representation
reduces to%
\[
\underset{\left\{  A_{k}\right\}  _{k=-\infty}^{\infty}}{\sup}\sum
\limits_{-\left(  K-1\right)  }^{K}\max\left\{  \mu_{0}\left(  A_{k}\right)
-\mu_{1}\left(  A_{k}^{C}\right)  ,0\right\}  ,
\]
where $\left\{  A_{k}\right\}  _{k=-\left(  K-1\right)  }^{K}$ is
monotonically decreasing. As $K$ goes to infinity, this obviously converges to
the dual representation for the infinite Lagrange multiplier, which is given
in (11).
\end{remark}

\subsection{Economic Examples}

In this subsection, I\ derive sharp bounds on the DTE for concrete economic
examples from the general characterization in Theorem 1. \ As
economic examples, MTR, concave treatment response, convex treatment response,
and the Roy model of self-selection are discussed.

\subsubsection{Monotone Treatment Response}

Since the seminal work of Manski (1997), it has been widely recognized that
MTR has an interesting identifying power for treatment effects parameters.
MTR only requires that the potential outcomes be weakly monotone in treatment
with probability one:%
\[
\Pr\left(  Y_{1}\geq Y_{0}\right)  =1.
\]

His\ bounds on the DTE under MTR are obtained as follows: for $\delta<0,$
$F_{\Delta}\left(  \delta\right)  =0,$ and for $\delta\geq0,$%
\[
\Pr\left(  Y-y_{0}^{L}\leq\delta|D=1\right)  p+\Pr\left(  y_{1}^{U}%
-Y\leq\delta|D=0\right)  \left(  1-p\right)  \leq F_{\Delta}\left(
\delta\right)  \leq1,
\]
where $p=\Pr\left(  D=1\right)  ,$ and $y_{0}^{L}$ \ is the support infimum of
$Y_{0}$ while $y_{1}^{U}$ \ is the support supremum of $Y_{1}.$ He did not
impose any other condition such as given marginal distributions of $Y_{0}
$\ and $Y_{1}$. Note that MTR has no identifying power on the DTE\ in the
binary treatment setting without additional information. Since MTR restricts
only the lowest possible value of $Y_{1}-Y_{0}$ as zero, the upper bound is
trivially obtained as one for any $\delta\geq0$. Similarly, MTR is
uninformative for the\ lower bound, since MTR does not restrict the highest
possible value of $Y_{1}-Y_{0}$.\footnote{Note that $Y_{1}$ is observed for
the treated and $Y_{0}$ is observed for the untreated groups. For the treated,
the highest possible value is $Y-Y_{0}^{L}$, while it is $Y_{1}^{U}-Y$ for the
untreated. The lower bound is achieved when $\Pr(Y_{0}=y_{0}^{L}|D=1)=1$ and
$(Y_{1}=y_{1}^{U}|D=0)=1.$} Furthermore, when the support of each potential
outcome is given as $\mathbb{R}$, they yield completely uninformative upper
and lower bounds $\left[  0,1\right]  .$

However, I\ show that given marginal distribution functions $F_{0}$ and
$F_{1}$, MTR has substantial identifying power for the lower bound on the DTE.

\begin{corollary}
\label{MTR} Suppose that $\Pr\left(  Y_{1}=Y_{0}\right)  =0$. Under MTR, sharp
bounds on the DTE are given as follows: for any $\delta\in\mathbb{R}, $%
\[
F_{\Delta}^{L}\left(  \delta\right)  \leq F_{\Delta}\left(  \delta\right)
\leq F_{\Delta}^{U}\left(  \delta\right)  ,
\]
where%
\begin{align*}
F_{\Delta}^{U}\left(  \delta\right)   & =\left\{
\begin{array}
[c]{cc}%
1+\underset{y\in\mathbb{R}}{\inf}\left\{  \min\left(  F_{1}\left(  y\right)
-F_{0}\left(  y-\delta\right)  \right)  ,0\right\}  , & \text{for }\delta
\geq0,\\
0, & \text{for }\delta<0.
\end{array}
\right.  ,\\
F_{\Delta}^{L}\left(  \delta\right)   & =\left\{
\begin{array}
[c]{cc}%
\underset{\left\{  a_{k}\right\}  _{k=-\infty}^{\infty}\in\mathcal{A}_{\delta
}}{\sup}\sum\limits_{k=-\infty}^{\infty}\max\left\{  F_{1}\left(
a_{k+1}\right)  -F_{0}\left(  a_{k}\right)  ,0\right\}  , & \text{for }%
\delta\geq0,\\
0, & \text{for }\delta<0,
\end{array}
\right.  ,\\
\text{where }\mathcal{A}_{\delta}  & =\left\{  \left\{  a_{k}\right\}
_{k=-\infty}^{\infty};0\leq a_{k+1}-a_{k}\leq\delta\text{ for each integer
}k\right\}  .
\end{align*}

\end{corollary}

\begin{proof}
See Appendix A.
\end{proof}

The identifying power of MTR on the lower bound has an interesting
graphical interpretation. As shown in Figure 7(a), the DTE under MTR corresponds to the probability of the region between
two straight lines $Y_{1}=Y_{0}$ and $Y_{1}=Y_{0}+\delta$. Given marginal
distributions, the Makarov lower bound is obtained by picking $y^{\ast}%
\in\mathbb{R}$ such that a rectangle $[y^{\ast}-\delta,\infty)\times
(-\infty,y^{\ast}]$ yields the maximum Fr\'{e}chet lower bound among all
rectangles below the straight line $Y_{1}=Y_{0}+\delta.$ As shown in Figure 7(b), under MTR the probability of any
rectangle $[y-\delta,\infty)\times(-\infty,y]$ below the straight line
$Y_{1}=Y_{0}+\delta$ is equal to that of the triangle between two straight
lines $Y_{1}=Y_{0}+\delta$ and $Y_{1}=Y_{0}.$ Now one can draw \emph{multiple}
mutually disjoint triangles between two straight lines $Y_{1}=Y_{0} $ and
$Y_{1}=Y_{0}+\delta$ as in Figure 7(c). Since
the probability of each triangle is equal to the probability of the rectangle
extended to the right and bottom sides, the lower bound on each triangle is
obtained by applying the Fr\'{e}chet lower bound to the extended rectangle.
Then the improved lower bound is obtained by summing the Fr\'{e}chet \ lower
bounds on the triangles.%
\begin{figure}
[ptb]
\begin{center}
\includegraphics[
natheight=4.937200in,
natwidth=5.521000in,
height=4.9926in,
width=5.578in
]%
{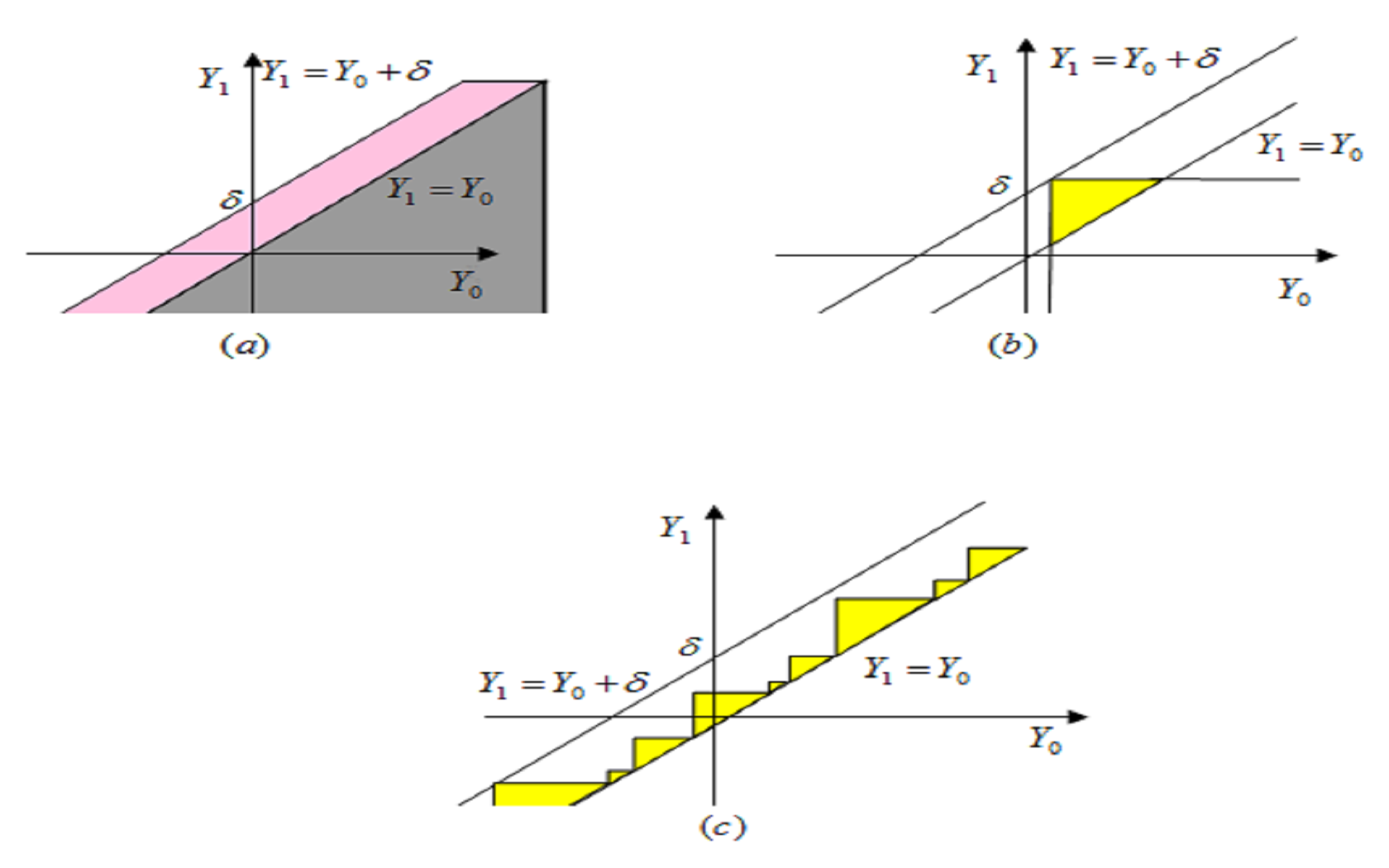}%
\caption{Improved lower bound under MTR}%
\label{improvedlowerboundunderMTR}%
\end{center}
\end{figure}

One of the key benefits of my characterization based on the optimal
transportation approach is that it guarantees sharpness of the bounds. To show
sharpness of given bounds in a copula approach, one should show what
dependence structures achieve the bounds under fixed marginal distributions.
This is technically difficult under MTR. However, the optimal transportation
approach gets around this challenge by focusing on a dual representation
involving given marginal distributions only.

Now I\ provide a sketch of the procedure to derive the lower bound under MTR
from Theorem 1. The proof of deriving the lower bound from
Theorem 1 proceeds in two stpng.

The first step is to show that the dual function $\varphi$ is nondecreasing so
that one can put $A_{k}=(a_{k},\infty)$ for $a_{k}\in\left[  -\infty
,\infty\right]  $ at the optimum. For any $\left(  y_{0},y_{1}\right)  $ in
the support of the optimal joint distribution, the dual function $\varphi$ for
the lower bound is written as%
\[
\varphi\left(  y_{0}\right)  =\underset{y_{1}\geq y_{0}}{\inf}\left\{
\boldsymbol{1}\left\{  y_{1}-y_{0}<\delta\right\}  -\psi\left(  y_{1}\right)
\right\}  .
\]
For any $\left(  y_{0}^{\prime},y_{1}^{\prime}\right)  $ and $\left(
y_{0}^{\prime\prime},y_{1}^{\prime\prime}\right)  $ with $y_{0}^{\prime\prime
}>y_{0}^{\prime}$ in the support of the optimal joint distribution,
\begin{align*}
\varphi\left(  y_{0}^{\prime}\right)   & =\boldsymbol{1}\left\{  y_{1}%
^{\prime}-y_{0}^{\prime}<\delta\right\}  -\psi\left(  y_{1}^{\prime}\right) \\
& \leq\boldsymbol{1}\left\{  y_{1}^{\prime\prime}-y_{0}^{\prime}%
<\delta\right\}  -\psi\left(  y_{1}^{\prime\prime}\right) \\
& \leq\boldsymbol{1}\left\{  y_{1}^{\prime\prime}-y_{0}^{\prime\prime}%
<\delta\right\}  -\psi\left(  y_{1}^{\prime\prime}\right) \\
& =\varphi\left(  y_{0}^{\prime\prime}\right)  .
\end{align*}
The first inequality in the second line follows from $y_{1}^{\prime\prime}\geq
y_{0}^{\prime\prime}>y_{0}^{\prime}$ The second inequality in the third line
is satisfied because $\boldsymbol{1}\left\{  y_{1}-y_{0}<\delta\right\}  $ is
nondecreasing in $y_{0}.$ Consequently, $\varphi$ is nondecreasing and thus
$A_{k}=(a_{k},\infty)$ for $a_{k}\in\lbrack-\infty,\infty]$ at the optimum.

$A_{k}^{D}$ is obtained from $A_{k}$ as follows: for $\delta>0$ and
$A_{k}=(a_{k},\infty)$ and $A_{k+1}=(a_{k+1},\infty),$%
\begin{align*}
A_{k}^{D}  & =\left\{  y_{1}\in\mathbb{R}|\exists y_{0}>a_{k}\text{ s.t.
}\delta\leq y_{1}-y_{0}\right\}  \cup\left\{  y_{1}\in\mathbb{R}|\exists
y_{0}>a_{k+1}\text{ s.t. }0\leq y_{1}-y_{0}<\delta\right\} \\
& =\left(  a_{k}+\delta,\infty\right)  \cup\left(  a_{k+1},\infty\right) \\
& =\left(  \min\left\{  a_{k}+\delta,a_{k+1}\right\}  ,\infty\right)  .
\end{align*}%
\begin{figure}
[ptb]
\begin{center}
\includegraphics[
natheight=3.416900in,
natwidth=4.854200in,
height=2.591in,
width=3.6685in
]%
{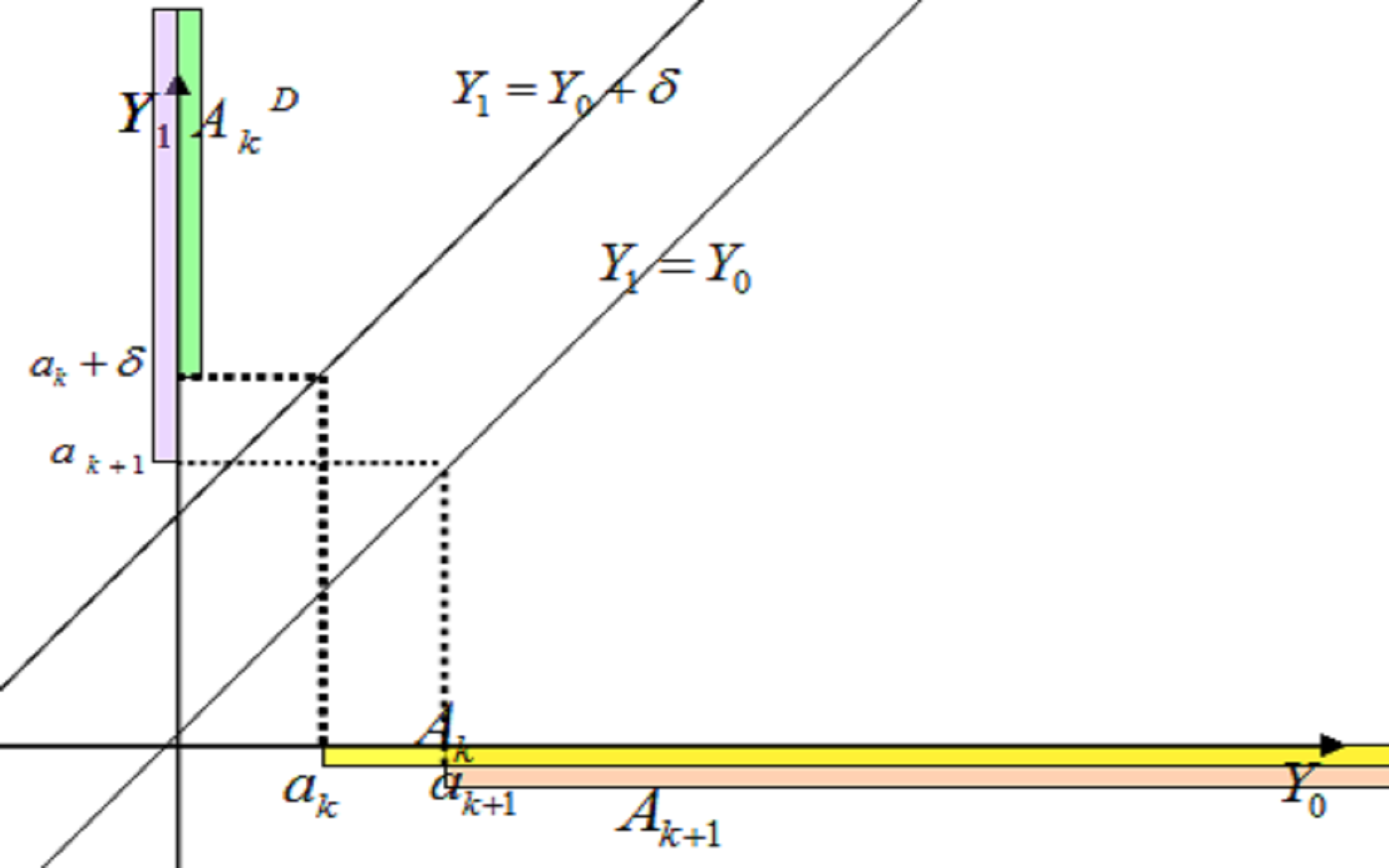}%
\caption{$A_{k}^{D}$ for $A_{k}=\left(  a_{k},\infty\right)  $ and
$A_{k+1}=\left(  a_{k+1},\infty\right)  .$}%
\label{AkDforAk}%
\end{center}
\end{figure}
\begin{figure}
[ptb]
\begin{center}
\includegraphics[
natheight=2.874600in,
natwidth=5.979300in,
height=2.3281in,
width=4.811in
]%
{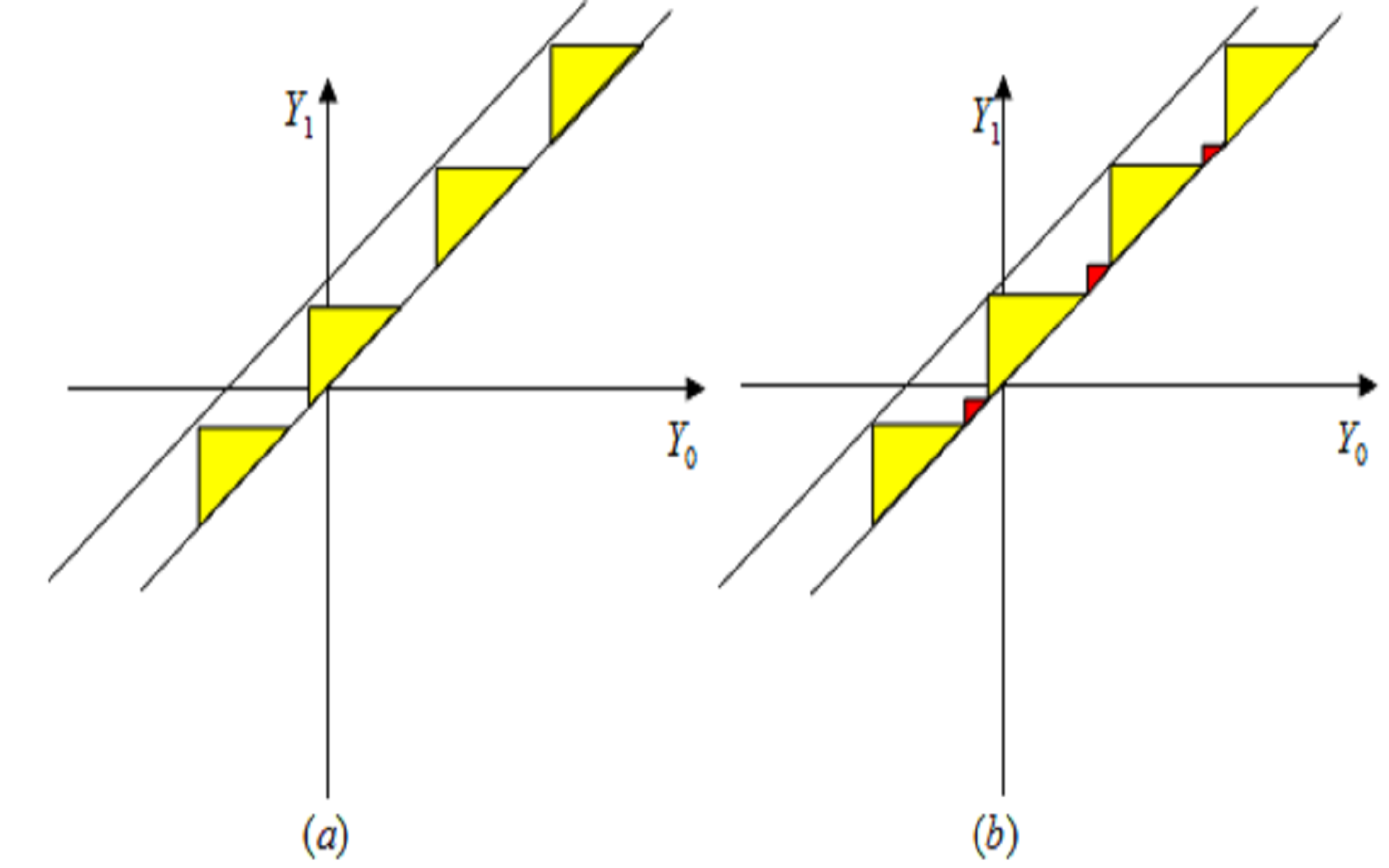}%
\caption{$a_{k+1}\leq a_{k}+\delta$ at the optimum}%
\label{ab1}%
\end{center}
\end{figure}

At the optimum, $\left\{  a_{k}\right\}  _{k=-\infty}^{\infty}$ should
satisfy\ $a_{k+1}\leq a_{k}+\delta$ for each integer $k$. The rigorous proof
is provided in Appendix A. I\ demonstrate this graphically here. As shown in
Figure 7(c), my improved lower bound represents
the sum of Fr\'{e}chet lower bounds on the probability of a sequence of
disjoint triangles. Suppose that $a_{k+1}>a_{k}+\delta$ for some integer $k$.
This implies that triangles in the region between two straight lines
$Y_{1}=Y_{0}+\delta$ and $Y_{1}=Y_{0}$ lie sparsely as shown in Figure
9(a). Then by adding extra triangles that fill the empty region
between two sparse triangles as shown in Figure 9(b), one can always
construct a sequence of mutually exclusive triangles that yield the identical
or improved lower bound. Therefore, without loss of generality, one can assume
$a_{k+1}\leq a_{k}+\delta$ for every integer $k$.%
\begin{figure}
[ptb]
\begin{center}
\includegraphics[
natheight=2.896300in,
natwidth=5.812400in,
height=2.9395in,
width=5.8712in
]%
{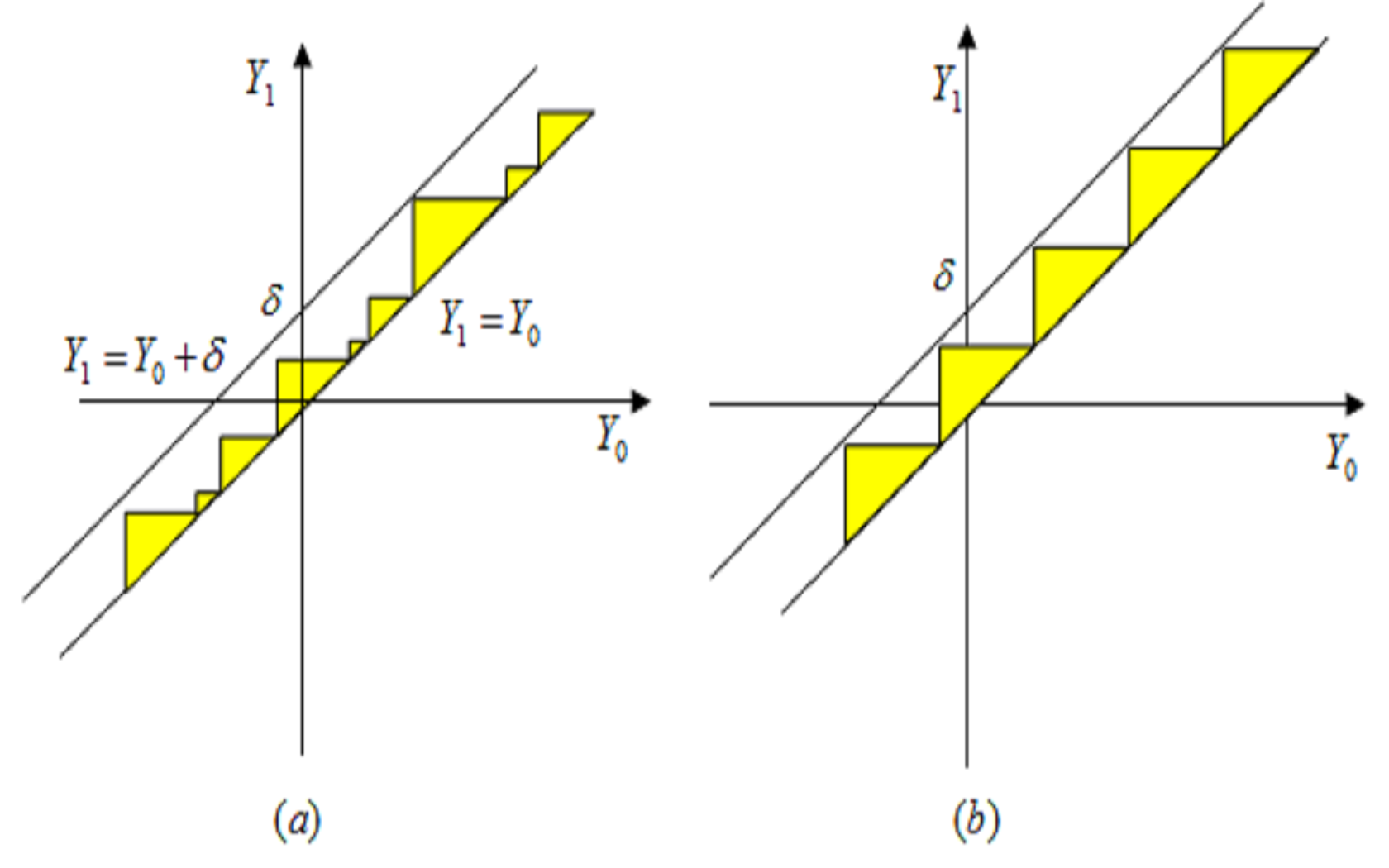}%
\caption{$a_{k+1}\leq a_{k}+\delta$ v.s. $a_{k+1}=a_{k}+\delta$}%
\label{ab2}%
\end{center}
\end{figure}

On the other hand, ones\ cannot exclude the case where $a_{k+1}<a_{k}+\delta$
for some integer $k$ at the optimum$.$ This implies that for some $k,$ the
triangle is not large enough to fit in the region corresponding to the DTE
under MTR as shown in Figure 10(b). It depends on the underlying joint
distribution which sequence of triangles would yield the tighter lower bound,
and it is possible that $a_{k+1}<a_{k}+\delta$ for some integer $k$ at the
optimum$.$ Therefore,
\begin{align*}
A_{k}^{D}  & =\left(  a_{k}+\delta,\infty\right)  \cup\left(  a_{k+1}%
,\infty\right) \\
& =\left(  \min\left\{  a_{k}+\delta,a_{k+1}\right\}  ,\infty\right) \\
& =\left(  a_{k+1},\infty\right)  .
\end{align*}
Consequently, for $\delta\geq0,$
\begin{align*}
F_{\Delta}^{L}\left(  \delta\right)   & =\underset{\left\{  A_{k}\right\}
_{k=-\infty}^{\infty}}{\sup}\sum\limits_{k=-\infty}^{\infty}\max\left\{
\mu_{0}\left(  A_{k}\right)  -\mu_{1}\left(  A_{k}^{D}\right)  ,0\right\} \\
& =\underset{\left\{  a_{k}\right\}  _{k=-\infty}^{\infty}}{\sup}%
\sum\limits_{k=-\infty}^{\infty}\max\left\{  F_{1}\left(  a_{k+1}\right)
-F_{0}\left(  a_{k}\right)  ,0\right\}
\end{align*}
where $0\leq a_{k+1}-a_{k}\leq\delta.$

\subsubsection{Concave/Convex Treatment Response}

Recall the setting of Example 2 in Subsection 2.1. Let $W$ denote the outcome
without treatment and let $Y_{0}$ and $Y_{1}$ denote the potential outcomes
with treatment at low-intensity, and with treatment at high-intensity,
respectively. Let $t_{d}$ denote the level of input for each treatment status
for $d=0,1,$ while $t_{W}$\ is a level of input without the treatment with
$t_{W}<t_{0}<t_{1}$. Either $\left(  W,Y_{0}\right)  $ or $\left(
X,Y_{1}\right)  $ is observed for each individual, but not $(W,Y_{0},Y_{1})$.
Given $W=w,$ the distribution of $Y_{1}-Y_{0}$\ under concave treatment
response corresponds to\ the probability of the intersection of $\left\{
Y_{1}-Y_{0}\leq\delta\right\}  $, $\left\{  \frac{Y_{0}-w}{t_{0}-t_{W}}%
\geq\frac{Y_{1}-Y_{0}}{t_{1}-t_{0}}\right\}  ,$ and $\left\{  Y_{1}\geq
Y_{0}\geq w\right\}  $\ in the support space of $(Y_{0},Y_{1})$. Similarly,
given $W=w,$ the distribution of $Y_{1}-Y_{0}$\ under convex treatment
response corresponds to the\ probability of the intersection of $\left\{
Y_{1}-Y_{0}\leq\delta\right\}  $, $\left\{  \frac{Y_{1}-Y_{0}}{t_{1}-t_{0}%
}\geq\frac{Y_{0}-w}{t_{0}-t_{W}}\right\}  ,$ and $\left\{  Y_{1}\geq Y_{0}\geq
w\right\}  $\ in the support space of $(Y_{0},Y_{1})$. Note that $\left\{
\frac{Y_{0}-w}{t_{0}-t_{W}}\geq\frac{Y_{1}-Y_{0}}{t_{1}-t_{0}}\right\}  $ and
$\left\{  \frac{Y_{1}-Y_{0}}{t_{1}-t_{0}}\geq\frac{Y_{0}-w}{t_{0}-t_{W}%
}\right\}  $ correspond to the regions below and above the straight line
$Y_{1}=$ $\frac{t_{1}-t_{W}}{t_{0}-t_{W}}Y_{0}-\frac{t_{1}-t_{0}}{t_{0}-t_{W}%
}w$, respectively.%
\begin{figure}
[ptb]
\begin{center}
\includegraphics[
natheight=3.739500in,
natwidth=4.499600in,
height=3.2076in,
width=3.8536in
]%
{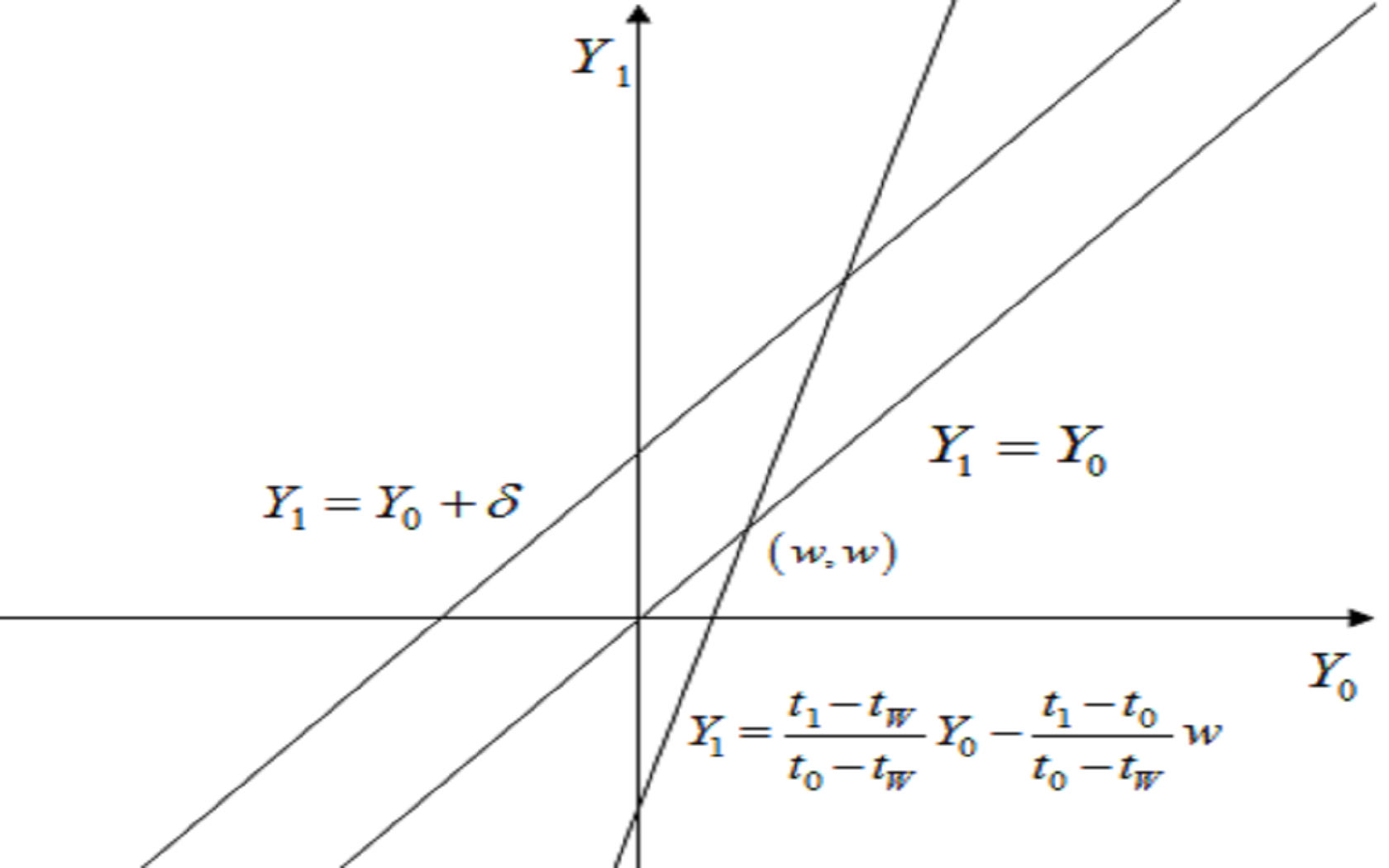}%
\caption{The DTE under concave/convex treatment response}%
\label{concavetreatmentresponse}%
\end{center}
\end{figure}

Corollary 2 derives sharp bounds under concave treatment
response and convex treatment response from Theorem 1.

\begin{corollary}
\label{concaveTR}Take any $w$ in the support of $W$ such that the conditional
marginal distributions of $Y_{1}$ and $Y_{0}$ given $W=w$ are both absolutely
continuous with respect to the Lebesgue measure on $\mathbb{R}$.$\ $Let
$F_{0,W}\left(  \cdot|w\right)  $ and $F_{1,W}\left(  \cdot|w\right)  $\ be
conditional distribution functions of $Y_{0}$ and $Y_{1}$ given $W=w$,
respectively.\newline(i) Under concave treatment response, sharp bounds on the
DTE\ are given as follows: for any $\delta\in\mathbb{R},$%
\[
F_{\Delta}^{L}\left(  \delta\right)  \leq F_{\Delta}\left(  \delta\right)
\leq F_{\Delta}^{U}\left(  \delta\right)
\]
where
\begin{align*}
F_{\Delta}^{L}\left(  \delta\right)   & =\underset{\left\{  a_{k}\right\}
_{k=-\infty}^{\infty}}{\sup}\sum\limits_{k=-\infty}^{\infty}\int\max\left\{
F_{1,W}\left(  a_{k+1}|w\right)  -F_{0,W}\left(  a_{k}|w\right)  ,0\right\}
dF_{W},\\
F_{\Delta}^{U}\left(  \delta\right)   & =1+\int\underset{\left\{
b_{k}\right\}  _{k=-\infty}^{\infty}}{\inf}\sum\limits_{k=-\infty}^{\infty
}\left\{  \min\left(  F_{1,W}\left(  \frac{1}{T_{0}}b_{k+1}-\frac{T_{1}}%
{T_{0}}w\ |w\right)  -F_{0,W}\left(  b_{k}\ |w\right)  \right)  ,0\right\}
dF_{W},
\end{align*}
with
\begin{align*}
0  & \leq a_{k+1}-a_{k}\leq\delta,\\
T_{0}\left(  b_{k}+\delta\right)  +T_{1}  & \leq b_{k+1}\leq b_{k},\\
\text{where \ }T_{1}  & =\frac{t_{1}-t_{0}}{t_{1}-t_{W}},\\
T_{0}  & =1-T_{1}.\newline%
\end{align*}
(ii) Under convex treatment response,%
\begin{align*}
F_{\Delta}^{L}\left(  \delta\right)   & =\int\underset{\left\{  a_{k}\right\}
_{k=-\infty}^{\infty}}{\sup}\sum\limits_{k=-\infty}^{\infty}\max\left\{
F_{1,W}\left(  S_{1}a_{k+1}+\left(  1-S_{1}\right)  w|w\right)  -F_{0,W}%
\left(  a_{k}|w\right)  ,0\right\}  dF_{W},\\
F_{\Delta}^{U}\left(  \delta\right)   & =1+\int\underset{y\in\mathbb{R}}{\inf
}\left\{  \min\left(  F_{1,W}\left(  y|w\right)  -F_{0,W}\left(
y-\delta|w\right)  \right)  ,0\right\}  dF_{W}.
\end{align*}
with%
\begin{align*}
a_{k}  & \leq a_{k+1}\leq\frac{1}{S_{1}}\left\{  \left(  a_{k}+\delta\right)
+\frac{1}{S_{0}}w\right\}  ,\\
S_{1}  & =\frac{t_{1}-t_{W}}{t_{0}-t_{W}},\\
S_{0}  & =\frac{t_{0}-t_{W}}{t_{1}-t_{0}}.
\end{align*}

\end{corollary}

\begin{proof}
See Appendix A.
\end{proof}

\subsubsection{Roy Model}

Establishing sharp DTE bounds under support restrictions allows us to derive
sharp DTE bounds in the Roy model. In the Roy model, each agent selects into
treatment when the net benefit from doing so is positive. The Roy model is
often divided into three versions according to the form of its selection
equation: the original Roy model, the extended Roy model, and the generalized
Roy model. Most of the recent literature considers the extended or generalized
Roy model that accounts for nonpecuniary costs of selection.

Consider the generalized Roy model in Heckman et al. (2011)\ and French and
Taber (2011):
\begin{align*}
Y  & =\mu\left(  D,X\right)  +U_{D},\\
D  & =\boldsymbol{1}\left\{  Y_{1}-Y_{0}\geq m_{C}\left(  Z\right)
+U_{C}\right\}  ,
\end{align*}
where $X$ is a vector of observed covariates while $\left(  U_{1}%
,U_{0}\right)  $ are unobserved gains in the equation of potential outcomes.
In the selection equation, $Z$ is a vector of observed cost shifters while
$U_{C}$ is an unobserved scalar cost. The main assumption in this model is
\[
\left(  U_{1},U_{0},U_{c}\right)  \perp\!\!\!\perp(X,Z).
\]
As two special cases of the generalized Roy model, the original Roy model
assumes that $\mu_{C}\left(  Z\right)  =U_{C}=0$ and the extended Roy model
assumes that each agent's cost is deterministic with $U_{C}=0$. My result
provides DTE bounds in the extended Roy model:%
\begin{align*}
Y  & =m\left(  D,X\right)  +U_{D},\\
D  & =\boldsymbol{1}\left\{  Y_{1}-Y_{0}\geq m_{C}\left(  Z\right)  \right\}
.
\end{align*}
The DTE in the extended Roy model is written as follows:%
\begin{align*}
F_{\Delta}\left(  \delta\right)   & =E\left[  \Pr\left(  Y_{1}-Y_{0}\leq
\delta|X\right)  \right] \\
& =E\left[  \Pr\left(  Y_{1}-Y_{0}\leq\delta|X,z\right)  \right] \\
& =E\left[  F_{\Delta}\left(  \delta|1,X,z\right)  \right]  p\left(  z\right)
+E\left[  F_{\Delta}\left(  \delta|0,X,z\right)  \right]  \left(  1-p\left(
z\right)  \right)  ,
\end{align*}
where $p\left(  z\right)  =\Pr\left(  D=1|Z=z\right)  $, $F_{\Delta}\left(
\delta|d,,X,z\right)  =\Pr\left(  Y_{1}-Y_{0}\leq\delta|D=d,X,Z=z\right)  $
for $d\in\left\{  0,1\right\}  .$ French and Taber (2011) listed sufficient
conditions under which the marginal distributions of potential outcomes are
point-identified in the generalized Roy model.\footnote{See Assumption 4.1-4.6
in French and Taber (2011). These assumptions include some high level
conditions such as the full support of both instruments and of exclusive
covariates for each sector. If those conditions are not satisfied, the
marginal distributions may only be partially identified.} Those assumptions
also apply to the extended Roy model since it is a special case of the
generalized Roy model. Under their conditions, conditional marginal
distributions of $Y_{0}$ and $Y_{1}$ on the treated $(D=1)$ and untreated
$(D=0)$ are also all point-identified. Note that given $Z=z$, the treated and
untreated groups correspond to the regions $\left\{  Y_{1}-Y_{0}\geq
m_{C}\left(  z\right)  \right\}  $ and $\left\{  Y_{1}-Y_{0}<m_{C}\left(
z\right)  \right\}  $\ respectively. Let $F_{d_{1}}\left(  y|d_{2},z\right)
=\Pr\left(  Y_{d_{1}}\leq y|D=d_{2},Z=z\right)  .$ Bounds on the DTE are
obtained based on the identified marginal distributions on the treated and
untreated as follows: for $d\in\left\{  0,1\right\}  ,$
\[
F_{\Delta}^{L}\left(  \delta|d,z\right)  \leq F_{\Delta}\left(  \delta
|d,z\right)  \leq F_{\Delta}^{U}\left(  \delta|d,z\right)  ,
\]
where
\[
F_{\Delta}^{L}\left(  \delta|1,z\right)  =\left\{
\begin{array}
[c]{cc}%
\underset{\left\{  a_{k}\right\}  _{k=-\infty}^{\infty}}{\sup}\sum
\limits_{k=-\infty}^{\infty}\max\left\{
\begin{array}
[c]{c}%
F_{1}\left(  a_{k+1}+m_{C}\left(  z\right)  |1,z\right)  -F_{0}\left(
a_{k}|1,z\right)  ,\\
0
\end{array}
\right\}  , & \text{for }\delta\geq m_{C}\left(  z\right)  ,\\
0, & \text{for }\delta<m_{C}\left(  z\right)  ,
\end{array}
\right.
\]
with
\[
a_{k}\leq a_{k+1}\leq a_{k}+\delta-m_{C}\left(  z\right)  ,
\]
and%
\begin{align*}
F_{\Delta}^{U}\left(  \delta|1,z\right)   & =\left\{
\begin{array}
[c]{cc}%
1+\underset{y\in\mathbb{R}}{\inf}\left\{  \min\left(  F_{1}\left(
y|1,z\right)  -F_{0}\left(  y-\delta|1,z\right)  \right)  ,0\right\}  , &
\text{for }\delta\geq m_{C}\left(  z\right)  ,\\
0, & \text{for }\delta<m_{C}\left(  z\right)  ,
\end{array}
\right. \\
F_{\Delta}^{L}\left(  \delta|0,z\right)   & =\left\{
\begin{array}
[c]{cc}%
1, & \text{for }\delta\geq m_{C}\left(  z\right)  ,\\
\underset{y\in\mathbb{R}}{\sup}\max\left\{  F_{1}\left(  y\right)
-F_{0}\left(  y-\delta\right)  ,0\right\}  , & \text{for }\delta<m_{C}\left(
z\right)  ,
\end{array}
\right. \\
F_{\Delta}^{U}\left(  \delta|0,z\right)   & =\left\{
\begin{array}
[c]{cc}%
1, & \text{for }\delta\geq m_{C}\left(  z\right)  ,\\
1+\underset{\left\{  b_{k}\right\}  _{k=-\infty}^{\infty}}{\inf}\left\{
\min\left(  F_{1}\left(  b_{k+1}+m_{C}\left(  z\right)  \right)  -F_{0}\left(
b_{k}\right)  \right)  ,0\right\}  , & \text{for }\delta<m_{C}\left(
z\right)  ,
\end{array}
\right.
\end{align*}
with%
\[
b_{k}+\delta-m_{C}\left(  z\right)  \leq b_{k+1}\leq b_{k}.
\]
Based on the bounds on $F_{\Delta}\left(  \delta|d,z\right)  $, the
identification region of the DTE can be obtained by intersection bounds as
presented in Chernozhukov et al. (2013).\footnote{The bounds on the DTE are
sharp without any other additional assumption. Park (2013) showed that the DTE
can be point-identified in the extended Roy model under continuous IV with the
large support and a restriction on the function $m_{c}.$}

\begin{corollary}
\label{Roy}The DTE in the extended Roy model is bounded as follows:
\[
F_{\Delta}^{L}\left(  \delta\right)  \leq F_{\Delta}\left(  \delta\right)
\leq F_{\Delta}^{U}\left(  \delta\right)  ,
\]
where%
\begin{align*}
F_{\Delta}^{L}\left(  \delta\right)   & =\sup_{z}\left[  F_{\Delta}^{L}\left(
\delta|1,z\right)  p\left(  z\right)  +F_{\Delta}^{L}\left(  \delta
|0,z\right)  \left(  1-p\left(  z\right)  \right)  \right]  ,\\
F_{\Delta}^{U}\left(  \delta\right)   & =\inf_{z}\left[  F_{\Delta}^{U}\left(
\delta|1,z\right)  p\left(  z\right)  +F_{\Delta}^{U}\left(  \delta
|0,z\right)  \left(  1-p\left(  z\right)  \right)  \right]  .
\end{align*}

\end{corollary}

\section{Numerical Illustration}

This section provides numerical illustration to assess the informativeness of
my new bounds. Since my sharp bounds on the DTE\ under support restrictions
are written with respect to given marginal distribution functions $F_{0}$ and
$F_{1}$, the tightness of the bounds is affected by the properties of these
marginal distributions. I\ report the results of numerical examples to clarify
the association between the identifying power of my bounds and the marginal
distribution functions $F_{0}$ and $F_{1}.$ I\ focus on MTR, which is one of
the most widely applicable support restrictions in economics.

My numerical examples use the following data generating process for the
potential outcomes equation: for $d\in\left\{  0,1\right\}  ,$%
\[
Y_{d}=\beta d+\varepsilon,
\]
where $\beta\sim\chi^{2}\left(  k_{1}\right)  ,$ $\varepsilon\sim N\left(
0,k_{2}\right)  $, and $\beta\perp\!\!\!\perp\varepsilon.$ Obviously,
treatment effects $\Delta=$ $\beta\sim\chi^{2}\left(  k_{1}\right)  $ satisfy
MTR and marginal distribution functions $F_{0}$\ and $F_{1}$ are given as%
\begin{align*}
F_{1}\left(  y\right)   & =\int\limits_{-\infty}^{\infty}G\left(
y-x;k_{1}\right)  \phi\left(  \frac{x}{\sqrt{k_{2}}}\right)  dx,\\
F_{0}\left(  y\right)   & =\Phi\left(  \frac{y}{\sqrt{k_{2}}}\right)  ,
\end{align*}
where $G\left(  \cdot;k_{1}\right)  $ is the distribution function of a
$\chi^{2}\left(  k_{1}\right)  $ and $\Phi\left(  \cdot\right)  $ are the
standard normal probability density function and its distribution function, respectively.

Recall that the sharp upper bound under MTR is identical to the Makarov upper
boun, and the sharp lower bound on the DTE under MTR is given as follows: for
$\delta\geq0$,%
\begin{equation}
\underset{\left\{  a_{k}\right\}  _{k=-\infty}^{\infty}\in\mathcal{A}_{\delta
}}{\sup}\sum\limits_{k=-\infty}^{\infty}\max\left\{  F_{1}\left(
a_{k+1}\right)  -F_{0}\left(  a_{k}\right)  ,0\right\}  ,\label{opt1}%
\end{equation}
where $\mathcal{A}_{\delta}=\left\{  \left\{  a_{k}\right\}  _{k=-\infty
}^{\infty};0\leq a_{k+1}-a_{k}\leq\delta\text{ for each integer }k\right\}  .$
The lower bound requires computing the optimal sequence of $a_{k}$. The
specific computation procedure is described in Appendix B.
\begin{figure}
[ptb]
\begin{center}
\includegraphics[
height=3.7213in,
width=6.6746in
]%
{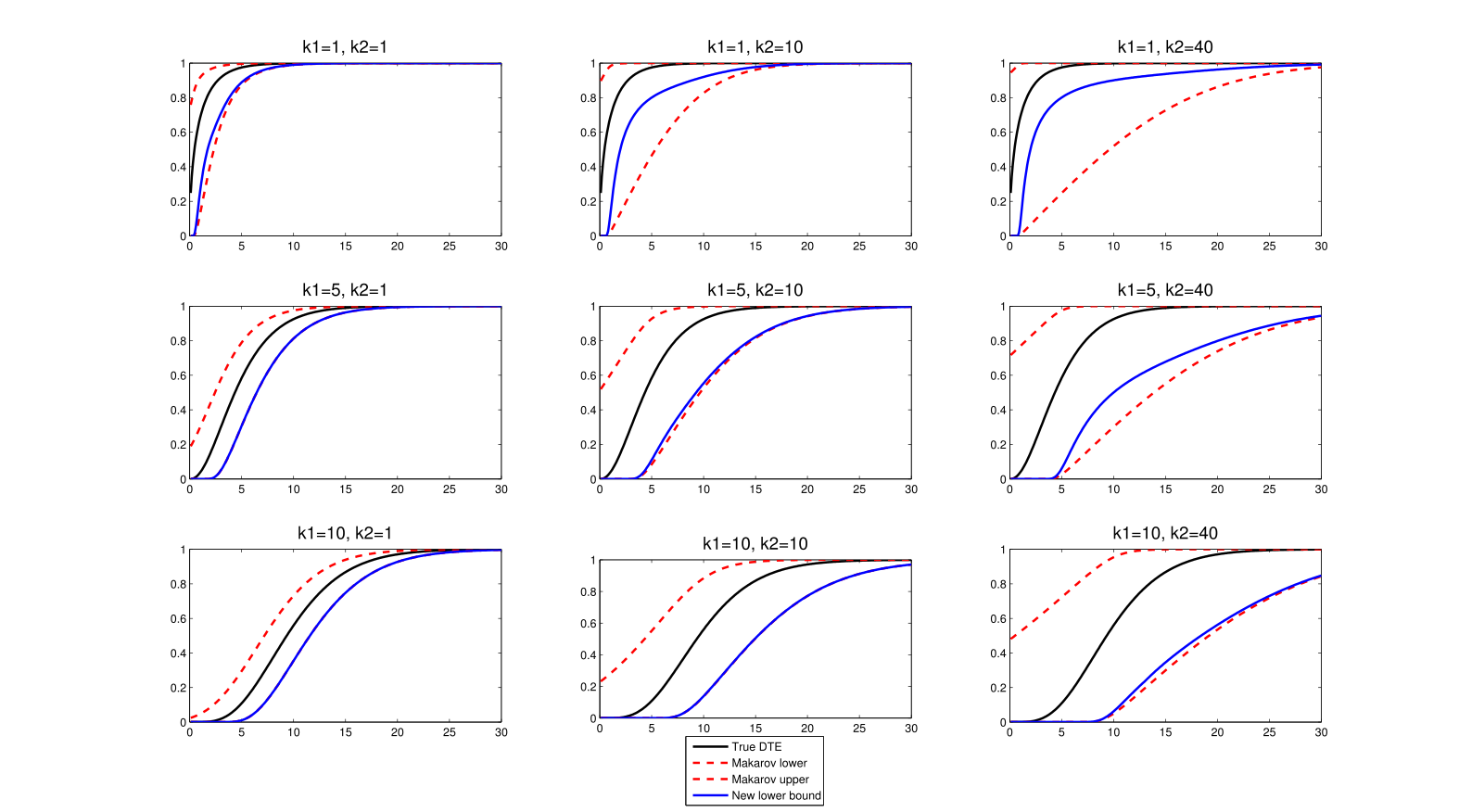}%
\caption{New bounds v.s. Makarov bounds}%
\label{bounds}%
\end{center}
\end{figure}

Figure 12 shows the true DTE as well as Makarov bounds and the
improved lower bound under MTR for $k_{1}=1,5,10$ and $k_{2}=1,10,40.$ To see
the effect of marginal distributions for the fixed true DTE $\Delta\sim
\chi^{2}\left(  k_{1}\right)  ,$ I\ focus on how the DTE bounds change for
different values of $k_{2}$\ and fixed $k_{1}$.

Figure 12 shows that Makarov bounds and my new lower bound become
less informative as $k_{2}$ increases. My data generating process assumes
$Y_{1}-Y_{0}\sim\chi^{2}\left(  k_{1}\right)  $, $Y_{0}\sim N\left(
0,k_{2}\right)  $ and $Y_{1}-Y_{0}\perp\!\!\!\perp Y_{0}.$\ \ When the true
DTE is fixed with a given value of $k_{1},$ both Makarov bounds and my new
bounds move further away from the true DTE as the randomness in the potential
outcomes $Y_{0}$ and $Y_{1}$\ increases with higher $k_{2}$. If $k_{2}=0$ as
an extreme case, in which $Y_{0}$ has a degenerate distribution, obviously
Makarov bounds as well as my new bounds point-identify the DTE.

Interestingly, as $k_{2}$ increases, my new lower bound moves further away
from the true DTE much more slowly than the Makarov lower bound. Therefore,
the information gain from MTR, which is represented by the distance between my
new lower bound and the Makarov lower bound, increases as $k_{2}$ increases.
This shows that under MTR, my new lower bound gets additional information from
the larger variation of marginal distributions.

To develop intuition, recall Figure 7(c). Under
MTR, the larger variation in marginal distributions $F_{0}$ and $F_{1}$ over
the support causes more triangles having positive probability lower bounds,
which leads the improvement of my new lower bound. On the other hand, the
Makarov lower bound gets no such informational gain because it uses only one
triangle while my new lower bound takes advantage of \emph{multiple} triangles.%

\begin{figure}
[ptb]
\begin{center}
\includegraphics[
height=3.7213in,
width=6.6746in
]%
{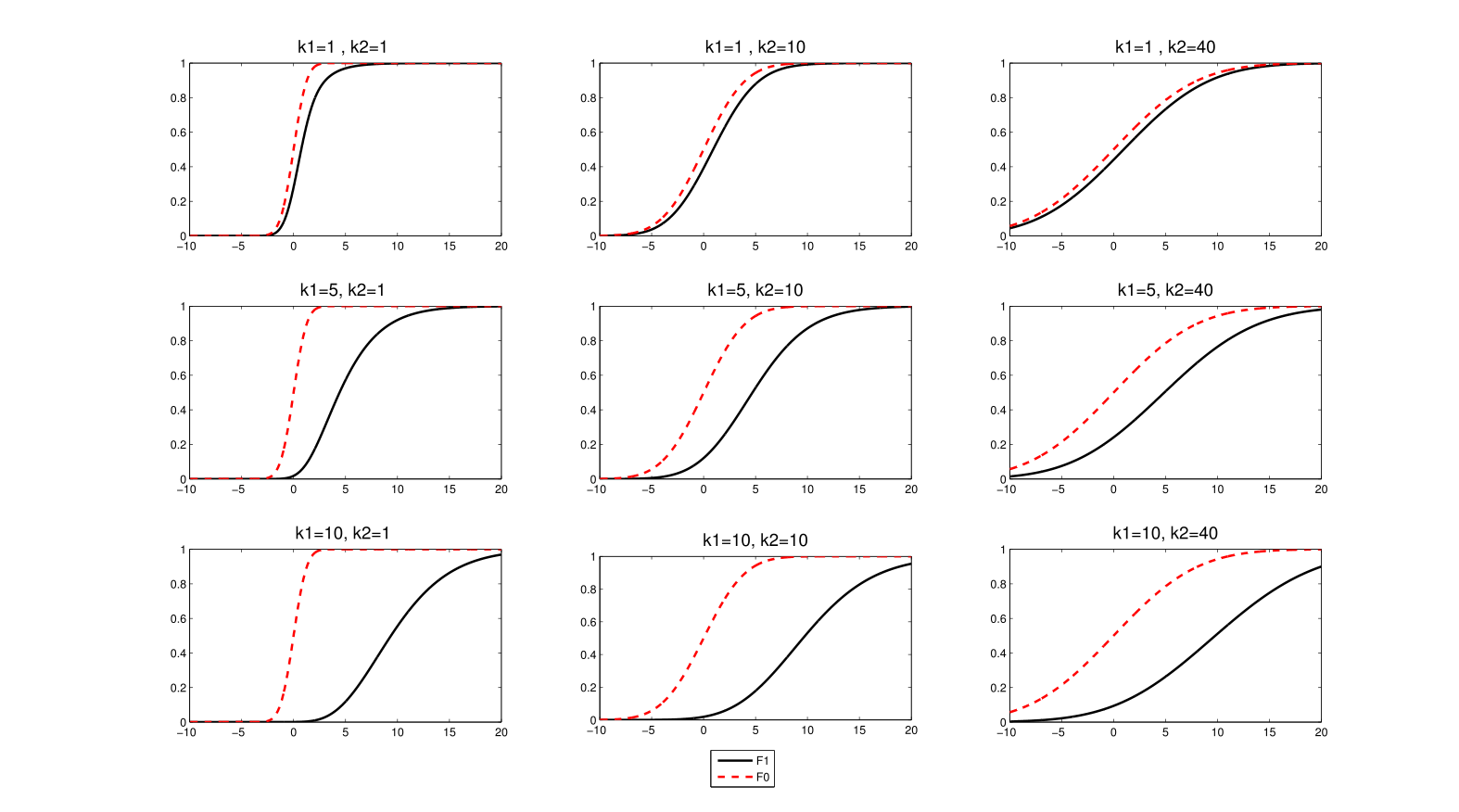}%
\caption{Marginal distributions of potential outcomes}%
\label{distributionfunctionsofpotentialoutcomes}%
\end{center}
\end{figure}

\section{Application to the Distribution of Effects\ of Smoking on Birth
Weight}

In this section, I\ apply the results presented in Section 3 to an empirical
analysis of the distribution of smoking effects on infant birth
weight.\ Smoking not only has a direct impact on infant birth weight, but is
also associated with unobservable factors that affect infant birth weight.
I\ identify marginal distributions of potential infant birth weight with and
without smoking by making use of a state cigarette tax hike in Massachusetts
(MA) in January 1993 as a source of exogenous variation. I\ focus on pregnant
women who change their smoking behavior from smoking to nonsmoking in response
to the tax increase.\ To identify the distribution of smoking effects,
I\ impose a MTR restriction that smoking has nonpositive effects on infant
birth weight with probability one. I\ propose an estimation procedure and
report estimates of the DTE\ bounds. I\ compare my new bounds to Makarov
bounds to demonstrate the informativeness and usefulness of my methodology.

\subsection{Background}

Birth weight has been widely used as an indicator of infant health and welfare
in economic research. Researchers have investigated social costs associated
with low birth weight (LBW), which is defined as birth weight less than 2500
grams, to understand the short term and long term effects of children's
endowments. For example, Almond et al. (2005) estimated the effects of birth
weight on medical costs, other health outcomes, and mortality rate, and Currie
and Hyson (1999) and Currie and Moretti (2007) evaluated the effects of low
birth weight on educational attainment and long term labor market outcomes.
Almond and Currie (2011) provide a survey of this literature.

Smoking has been acknowledged as the most significant and preventable cause of
LBW, and thus various efforts have been made to reduce the number of women
smoking during pregnancy. As one of these efforts, increases in cigarette
taxes have been widely used as a policy instrument between 1980 and 2009 in
the U. S. Tax rates on cigarettes have increased by approximately $\$0.80$
each year on average across all states, and more than $80$ tax increases of
$\$0.25$ \ have been implemented in the past $15$ years (Simon (2012) and
Orzechowski and Walker (2011)).

In the literature, there have been various attempts to clarify the causal
effects of smoking on infant birth weight. Most previous empirical studies
have evaluated the average\ effects of smoking or effects on the marginal
distribution of potential infant birth weight focusing on the methods to
overcome the endogeneity of smoking behavior.

My analysis pays particular attention to the distribution of smoking effects
on infant birth weight. The DTE\ conveys the information on the targets of
anti-smoking policy, which is particularly important for this study, because
the DTE can answer the following questions: "how many births are significantly
vulnerable to smoking ?" and "who should the interventions intensively target?".

I make use of the cigarette tax increase in MA in January of 1993, which
increased the state excise tax from $\$0.26$ to $\$0.51$ per pack, as an
instrument to identify marginal distributions of potential birth weight
acknowledging the presence of endogeneity in smoking behavior. In November
1992, MA voters passed a ballot referendum to raise the tax on tobacco
products, and in 1993 the Massachusetts Tobacco Control Program was
established with a portion of the funds raised through this referendum. The
Massachusetts Tobacco Control Program initiated activities to promote smoking
cessation such as media campaigns, smoking cessation counselling, enforcement
of local antismoking laws, and educational programs targeted primarily at
teenagers and pregnant women.

The IV framework developed by Abadie, Angrist and Imbens (2002) is used to
identify and estimate marginal distributions of potential infant birth weight
for pregnant women who change their smoking status from smoking to nonsmoking
in response to the tax increase. Henceforth, I\ call this group of people
compliers. Based on the estimated marginal distributions, I\ establish sharp
bounds on the smoking effects under the MTR assumption that smoking has
adverse effects on infant birth weight.

\subsection{Related Literature}

The related literature can be divided into three strands by their empirical
strategy to overcome the endogenous selection problem. The first strand of the
literature, including Almond et al. (2005), assumes that smoking behavior is
exogenous conditional on observables such as mother's and father's
characteristics, prenatal care information, and maternal medical risk factors.
However, Caetano (2012) found strong evidence that smoking behavior is still
endogenous after controlling for the most complete covariate specification in
the literature. The second strand of the literature, including Permutt and
Hebel (1989), Evans and Ringel (1999), Lien and Evans (2005), and Hoderlein
and Sasaki (2011) takes an IV strategy. Permutt and Hebel (1989) made use of
randomized counselling as an exogenous variation, while Evans and Ringel
(1999), Hoderlein and Sasaki (2013) took advantage of cigarette tax rates or
tax increases.\footnote{Permutt and Hebel (1989), Evans and Ringel (1999) and
Lien and Evans (2005) two-stage linear regression to estimate the average
effect of smoking using an instrument. Hoderlein and Sasaki (2011) adopted the
number of cigarettes as a continuous treatment, and identified and estimated
the average marginal effect of a cigarette based on the nonseparable model
with a triangular structure.} The last strand takes a panel data approach.
This approach isolates the effects of unobservables using data on mothers with
multiple births and identifies the effect of smoking from the change in their
smoking status from one pregnancy to another. To do this, Abrevaya (2006)
constructed the panel data set with novel matching algorithms\ between women
having multiple births and children on federal natality data. The panel data
set constructed by Abrevaya (2006) has been used in other recent studies such
as Arellano and Bonhomme (2011) and Jun et al. (2013). Jun et al. (2013)
tested stochastic dominance between two marginal distributions of potential
birth weight with and without smoking. Arellano and Bonhomme (2011) identified
the distribution of smoking effects using the random coefficient panel data model.

To the best of my knowledge, the only existing study that examines the
distribution of smoking effects is Arellano and Bonhomme (2012). While they
point-identify the distribution of smoking effects, their approach presumes
access to the panel data with individuals who changed their smoking status
within their multiple births. Specifically, they use the following panel data
model with random coefficients:%
\[
Y_{it}=\alpha_{i}+\beta_{i}D_{it}+X_{it}^{\prime}\gamma+\varepsilon_{it}%
\]
where $Y_{it}$ is infant birth weight and $D_{it}$ \ is an indicator for woman
$i$ smoking before she had her $t$-th baby. Extending Kotlarski's
deconvolution idea, they identify the \emph{distribution} of $\beta
_{i}=E\left[  Y_{it}|D_{it}=1,\alpha_{i},\beta_{i}\right]  -E\left[
Y_{it}|D_{it}=0,\alpha_{i},\beta_{i}\right]  $, which indicates the
distribution of smoking effects in this example. For the identification, they
assume strict exogeneity that mothers do not change their smoking behavior
from their previous babies' birth weight. Furthermore, their estimation result
is somewhat implausible. It is interpreted that smoking has a positive effect
on infant birth weight for approximately 30\% mothers. They conjecture that
this might result from a misspecification problem such as the strict
exogeneity condition, i.i.d. idiosyncratic shock, etc.%

\begin{center}
\begin{table}[tbp]%

\caption{Data used in the recent literature}\label{dataintheliterature}%

\begin{tabular}
[c]{l|l|l}\hline\hline
& {Data} & { \# of obs.}\\\hline
{ Evans and Ringel (1999)} & \multicolumn{1}{|l|}{{ NCHS
(1989-1992)}} & ${ 10.5}${ \ million}\\
{ Almond et al. (2005)} & \multicolumn{1}{|l|}{{ NCHS(1989-1991,
PA only)}} & ${ 491,139}$\\
{ Abrevaya (2006)} & \multicolumn{1}{|l|}{{ matched panel
constructed from NCHS (1989-1998)}} & ${296,218}$\\
{ Arellano and Bonhomme (2011)} & { matched panel \#3 in Abrevaya
(2006)} & ${ 1,445}$\\
{ Jun et al. (2013)} & { matched panel \#3 in Abrevaya (2006)} &
${ 2,113}$\\
{ Hoderlein and Sasaki (2013)} & \multicolumn{1}{|l|}{{ random
sample from NCHS (1989-1999)}} & ${ 100,000}$\\\hline
\end{tabular}%
\end{table}%

\end{center}

Most existing studies used the Natality Data by the National Center for Health
Statistics (NCHS) for its large sample size and a wealth of information on
covariates. The birth data is based on birth records from every live birth in
the U.S. and contains detailed information on birth outcomes, maternal
prenatal behavior and medical status, and demographic
attributes.\footnote{Unfortunately the Natality Data does not provide
information on mothers' income and weight.} Table 1
describes the data used in the recent literature.

While some studies such as Hoderlein and Sasaki (2011) and Caetano (2012) use
the number of cigarettes per day as a continuous treatment variable, most
applied research uses a binary variable for smoking. The literature, including
Evans and Farrelly (1998), found that individuals, especially women, tend to
underreport their cigarette consumption. On the other hand, smoking
participation has shown to be more accurately reported among adults in the
literature. Moreover, the literature has pointed out that the number of
cigarettes may not be a good proxy for the level of nicotine intake. Previous
studies, including Chaloupka and Warner (2000), Evans and Farrelly (1998),
Farrelly et al. (2004), Adda and Cornaglia (2006), and Abrevaya and Puzzello
(2012) discussed that although an increase in cigarette taxes leads to a lower
percentage of smokers and less cigarettes consumed by smokers, it causes
individuals to purchase cigarettes that contain more tar and nicotine as
compensatory behavior.

Although many recent studies are based on the same NCHS data set, their
estimates of average smoking effects are quite varied, ranging from -144 grams
to -600 grams depending on their estimation methods and samples. Table 2 summarizes their estimates.%
\begin{table}[tbp]%

\begin{center}
\caption{Estimated average smoking effects on infant birth weight}\label{estintheliterature}%

\begin{tabular}
[c]{l|l}\hline\hline
& {\small Estimate (g)}\\\hline
{\small Evans and Ringel (1999)} & \multicolumn{1}{|c}{{\small -600} $-$
{\small -360}}\\
{\small Almond et al. (2005)} & \multicolumn{1}{|c}{{\small -203.2}}\\
{\small Abrevaya (2006)} & \multicolumn{1}{|c}{{\small -144} $-$
{\small -178}}\\
{\small Arellano and Bonhomme (2011)} & \multicolumn{1}{|c}{{\small -161}%
}\\\hline
\end{tabular}

\end{center}%

\end{table}%

\subsection{Data}

I\ use the NCHS Natality dataset. My sample\ consists of births to women who
were in their first trimester during the period between two years before and
two years after the tax increase. In other words, I\ consider births to women
who conceived babies in MA between October 1990 and September
1994.\footnote{To trace the month of conception, I\ use information on the
month of birth and the clinical estimate of gestation weeks.} I\ define the
instrument as an indicator of whether the agent faces the high tax rate from
the tax hike during the first trimester of pregnancy. Since the tax increase
occurred in MA in January of 1993, the instrument $Z$ can be written as
\begin{equation}
Z=\left\{
\begin{array}
[c]{cl}%
1, & \text{if a baby is conceived in Oct. 1992 or later}\\
0, & \text{if a baby is conceived before Oct. 1992}%
\end{array}
\right. \label{ZZZ}%
\end{equation}
The first trimester of pregnancy has received particular attention in the
medical literature on the effects of smoking. Mainous and Hueston (1994)
demonstrated that smokers who quit smoking within the first trimester showed
reductions in the proportion of preterm deliveries and low birth weight
infants, compared with those who smoked beyond the first trimester. Also,
Fingerhut et al. (1990) showed that approximately 70\% of women who quit
smoking during pregnancy do so as soon as they are aware of their pregnancy,
which is mostly the first trimester of pregnancy.

I\ take only singleton births into account and focus on births to mothers who
are white, Hispanic or black, and whose age is between 15 and 44. The
covariates that I\ use to control for observed characteristics include
mothers' race, education, age, martial status, birth year, sex of the baby,
the "Kessner" prenatal care index, pregnancy history, information on various
diseases such as anemia, cardiac, diabete alcohol use, etc.\footnote{As an
index measure for the quality of prenatal care, the Kessner index is
calculated based on month of pregnancy care started, \ number of prenatal
visits, and length of gestation. If the value 1 in the Kessner index indicates
'adequate' prenatal care, while the value 2 and the value 3 indicate
'intermediate' and 'inadequate' prenatal care, respectively. For details, see
Abrevaya (2006).}

Descriptive statistics for this sample are reported in Table 3. After
the tax increase, the smoking rate of pregnant women decreased from 23\% to
16\%. As expected, babies of nonsmokers are on average heavier than babies of
smokers by 214 grams and furthermore, nonsmokers' infant birth weight
stochastically dominate smokers' infant birth weight as shown in Figure
14. Also, smokers are on average 1.63 years
younger, 1.27 years less educated than nonsmokers, and less likely to have
adequate prenatal care in the Kessner index. Regarding race, black or Hispanic
pregnant women are less likely to smoke than white women.%

\begin{table}[tbp]%

\begin{center}
\caption{Means and Standard Deviations}\label{mean}
\begin{tabular}
[c]{cccccccc}\hline\hline
&  & \multicolumn{3}{c}{{\small Before/After Tax Increase}} &
\multicolumn{3}{c}{{\small Smoking/Nonsmoking}}\\\cline{3-8}
& {\small Entire sample} & {\small After} & {\small Before} & {\small Diff.} &
{\small Smokers} & {\small Nonsmokers} & {\small Diff.}\\\hline
{\small \# of obs.} & {\small 297,031} & {\small 144,251} & {\small 152,780} &
& {\small 57,602} & {\small 239,429} & \\
{\small Smoking} & {\small 0.19} & {\small 0.16} & {\small 0.23} &
{\small -0.07} &  &  & \\
{\small (proportion)} & {\small [0.40]} & {\small [0.36]} & {\small [0.42]} &
{\small (-50.64)} &  &  & \\
{\small Birth weight} & {\small 3416.81} & {\small 3416.73} & {\small 3416.88}
& {\small -0.15} & {\small 3244.31} & {\small 3458.30} & {\small -214.00}\\
{\small (grams)} & {\small [556.07]} & {\small [556.09]} & {\small [556.07]} &
{\small (-0.07)} & {\small [561.28]} & {\small [546.75]} & {\small (-82.57)}\\
{\small Age} & {\small 28.51} & {\small 28.70} & {\small 28.33} & {\small .37}
& {\small 27.19} & {\small 28.82} & {\small -1.63}\\
{\small (years)} & {\small [5.70]} & {\small [5.75]} & {\small [5.65]} &
{\small (17.58)} & {\small [5.67]} & {\small [5.66]} & {\small (-62.07)}\\
{\small Education} & {\small 13.46} & {\small 13.54} & {\small 13.38} &
{\small 0.15} & {\small 12.43} & {\small 13.71} & {\small -1.27}\\
& {\small [2.50]} & {\small [2.49]} & {\small [2.52]} & {\small (16.48)} &
{\small [2.16]} & {\small [2.52]} & {\small (-112.00)}\\
{\small Married} & {\small 0.74} & {\small 0.74} & {\small 0.75} &
{\small -0.004} & {\small 0.58} & {\small 0.78} & {\small -.20}\\
& {\small [0.43]} & {\small [0.74]} & {\small [.44]} & {\small (-2.64)} &
{\small [.49]} & {\small [0.41]} & {\small (-90.41)}\\
{\small Black} & {\small 0.10} & {\small 0.10} & {\small 0.10} &
{\small -0.005} & {\small 0.07} & {\small 0.11} & {\small -0.03}\\
& {\small [0.30]} & {\small [0.29]} & {\small [.30]} & {\small (-4.22)} &
{\small [0.26]} & {\small [0.31]} & {\small (-27.90)}\\
{\small Hispanic} & {\small 0.10} & {\small 0.10} & {\small 0.10} &
{\small 0.002} & {\small 0.06} & {\small 0.11} & {\small -0.06}\\
& {\small [0.30]} & {\small [0.30]} & {\small [0.30]} & {\small (2.23)} &
{\small [.24]} & {\small [0.32]} & {\small (-45.34)}\\
{\small Kessner=1} & {\small 0.84} & {\small .84} & {\small 0.83} &
{\small 0.01} & {\small 0.78} & {\small 0.85} & {\small -0.08}\\
& {\small [0.37]} & {\small [0.36]} & {\small [0.37]} & {\small (7.96)} &
{\small [.42]} & {\small [0.35]} & {\small (-41.69)}\\
{\small Kessner=2} & {\small 0.13} & {\small 0.13} & {\small 0.14} &
{\small -0.01} & {\small 0.18} & {\small 0.12} & {\small 0.05}\\
& {\small [0.34]} & {\small [0.34]} & {\small [0.34]} & {\small (-5.75)} &
{\small [0.38]} & {\small [0.33]} & {\small (30.35)}\\
{\small Gestation} & {\small 39.27} & {\small 39.25} & {\small 39.29} &
{\small -0.04} & {\small 39.14} & {\small 39.30} & {\small -0.17}\\
{\small (weeks)} & {\small [2.04]} & {\small [2.01]} & {\small [2.07]} &
{\small (-5.88)} & {\small [2.24]} & {\small [1.99]} & {\small (-16.29)}%
\\\hline\hline
\end{tabular}

\end{center}%

\scriptsize
Note: The table reports means and standard deviations (in brackets) for the sample used in this study. The columns showing differences in means (by assignment or treatment status) report the t-statistic (in parentheses) for the null hypothesis of equality in means.
\end{table}%
%

\begin{figure}
[ptb]
\begin{center}
\includegraphics[
natheight=5.833200in,
natwidth=7.778100in,
height=2.9456in,
width=3.9167in
]%
{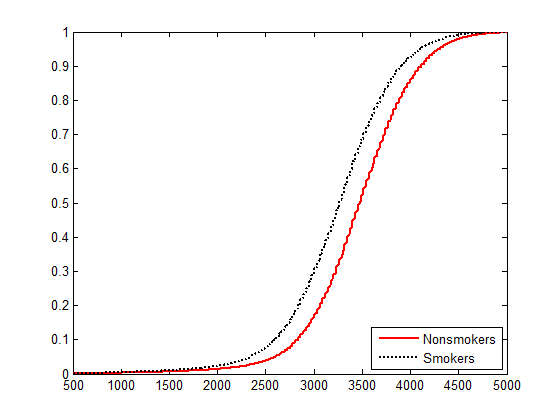}%
\caption{Distribution functions of infant birth weight of smokers and
nonsmokers}%
\label{distributionfunctionsof}%
\end{center}
\end{figure}

\subsection{Estimation}

Using the earlier notation, let $Y$ be observed infant birth weight and $D$
the nonsmoking indicator defined as
\[
D=\left\{
\begin{tabular}
[c]{ll}%
$1,$ & for a nonsmoker\\
$0,$ & for a smoker
\end{tabular}
\right.
\]
In addition, let $D_{z}$ denote a potential nonsmoking indicator given $Z=z.$
Let $Y_{0}$ be the potential infant birth weight if the mother is a smoker,
while $Y_{1}$ the potential infant birth weight if the mother is not a smoker.
As defined in (14), $Z$ is a tax increase indicator during the first
trimester.\ The $k\times1$ vector $X$ of covariates\ consists of binary
indicators for mother's race, age, education, marital status, birth order, sex
of the baby, "Kessner" prenatal care index, drinking status, and medical risk
factors. Since the treatment variable is nonsmoking here, the estimated effect
is the benefit of smoking cessation, which is in turn equal to the absolute
value of the adverse effect of smoking.\ To identify marginal distributions,
I\ impose the standard LATE assumptions following Abadie et al. (2002):

\begin{assumption}
\label{LATEassumption}For almost all values of $X:$

(i) Independence: $\left(  Y_{1},Y_{0},D_{1},D_{0}\right)  $ is jointly
independent of $Z$ given $X$.

(ii) Nontrivial Assignment: $\Pr\left(  Z=1|X\right)  \in\left(  0,1\right)
.$

(iii) First-stage: $E\left[  D_{1}|X\right]  \neq E\left[  D_{0}|X\right]  .$

(iv) Monotonicity: $\Pr\left(  D_{1}\geq D_{0}|X\right)  =1.$
\end{assumption}

Assumption 2(i) implies that the tax increase exogenously
affects the smoking status conditional on observables and that any effect of
the tax increase on infant birth weight must be via the change in smoking
behavior. This is plausible in my application since the tax increase acts as
an exogenous shock.\footnote{The state cigarette tax rate and tax increases
have been widely recognized as a valid instrument in the literature such as
Evans and Ringel (1999), Lien and Evans (2005) and Hoderlein and Sasaki
(2011), among others.} Assumption 2(ii) and (iii) obviously
hold in this sample. Assumption 2(iv) is plausible since an
increase in cigarette tax rates would never encourage smoking for each individual.

\subsubsection{The Marginal Treatment Effect and Local Average Treatment
Effect}

First, I estimate marginal effects of smoking cessation to see how the mean
effect varies with the individual's tendency to smoke. The marginal treatment
effect (MTE) is defined as follows:%
\[
MTE(x,p)=E[Y_{1}-Y_{0}|X=x,P\left(  Z,X\right)  =p].
\]
where $P\left(  Z,X\right)  =P(D=1|Z,X),$ which is the probability of not
smoking conditional on $Z$ and $X.$ In Heckman and Vytlacil (2005), the MTE is
recovered as follows:
\[
MTE(x,p)=\frac{\partial}{\partial p}E\left[  Y|X=x,P\left(  Z,X\right)
=p\right]  .
\]
Since the propensity score $p\left(  Z,X\right)  =\Pr\left(  D=1|Z,X\right)  $
is unobserved for each agent, I\ estimate it using the probit specification:%
\begin{equation}
p\left(  Z,X\right)  =\Phi\left(  \alpha+\beta Z+X^{\prime}\gamma\right)
.\label{FS1}%
\end{equation}
Then with the estimated propensity score $\widehat{p}\left(  Z,X\right)  $ in
(15), I\ estimate the following outcome equation:%
\begin{equation}
Y=\mu\left(  \widehat{p}\left(  Z,X\right)  ,X\right)  +u\label{SS1}%
\end{equation}
I\ estimate the equation (16) using a series approximation. This method
is especially convenient to estimate MTE $\frac{\partial\mu}{\partial p}.$ The
estimation results for the regressions (15) and (16)\ are
reported in Table C.1 and Table C.2, respectively, in Appendix C. Figure
15 shows estimated marginal treatment effects for
each propensity to not smoke. It is observed that the positive effect of
smoking cessation on infant birth weight increases as the tendency to smoke
increases. That is, the benefit of quitting smoking on child health is larger
for women who will still smoke despite facing higher tax rates. In turn, the
adverse effect of smoking on infant birth weight is\ more severe for women
with the higher tendency to smoke during pregnancy.%

\begin{figure}
[ptb]
\begin{center}
\includegraphics[
natheight=5.833200in,
natwidth=7.778100in,
height=2.9456in,
width=3.9159in
]%
{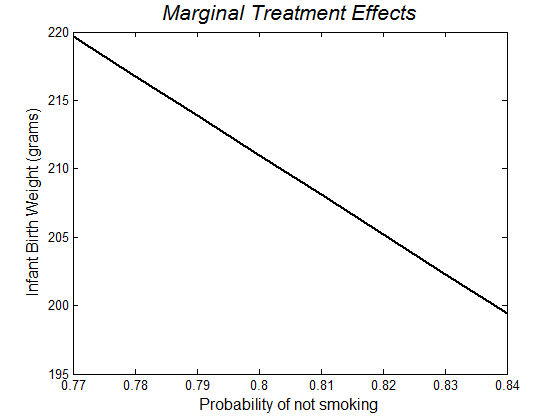}%
\caption{Marginal smoking effects}%
\label{marginaltreatmenteffects}%
\end{center}
\end{figure}

Next, I\ estimate LATE from the MTE. The LATE is interpreted as the benefit of
smoking cessation for compliers, women who change their smoking status from
smoker to nonsmoker in response to the tax increase. It is obtained from
marginal treatment effects as follows: for $\overline{p}\left(  x\right)
=\Pr\left(  D=1|Z=1,X=x\right)  $ and $\underline{p}=\Pr\left(
D=1|Z=0,X=x\right)  ,$
\[
E[Y_{1}-Y_{0}|X=x,D_{1}>D_{0}]=\frac{1}{\overline{p}\left(  x\right)
-\underline{p}\left(  x\right)  }\int_{\underline{p}\left(  x\right)
}^{\overline{p}\left(  x\right)  }MTE(x,p)dp.
\]
Table 4 presents estimated LATE for the entire sample and three
subgroups of white women, women aged 26-35, and women with some college or
college graduates (SCCG). The estimated benefit of smoking cessation is
noticeably small for SCCG women, compared to the entire sample and women whose
age is between 26 and 35. These MTE and LATE estimates show that births to
less educated women or women with a higher tendency to smoke are on average
more vulnerable to smoking. The literature, such as Deaton (2003) and Park and
Kang (2008), has found a positive association between smoking behavior and
other unhealthy lifestyles, and between higher education and a healthier
lifestyle. Given this association, my MTE and LATE estimates suggest that
births to women with an unhealthier lifestyle on average are more vulnerable
to smoking.%

\begin{table}[tbp]%

\begin{center}
\caption{Local Average Treatment Effects (grams)}\label{LATE}%
\begin{tabular}
[c]{l|l}\hline\hline
Dep. var.: birth weight (grams) & {\small LATE}\\\hline
{\small The entire sample} & \multicolumn{1}{|c}{{\small 209}}\\
{\small White} & \multicolumn{1}{|c}{{\small 133}}\\
{\small Age26-35} & \multicolumn{1}{|c}{{\small 183}}\\
{\small Some college and college graduates (SCCG)} &
\multicolumn{1}{|c}{{\small 112}}\\\hline
\end{tabular}

\end{center}%

\end{table}%

\subsubsection{Quantile Treatment Effects for Compliers}

In this subsection, I estimate the effect of smoking on quantiles of infant
birth weight through the quantile treatment effect (QTE) parameter. $q$-QTE
measures the difference in the $q$-quantile of $Y_{1}$ and $Y_{0}$, which is
written as $Q_{q}\left(  Y_{1}\right)  -Q_{q}\left(  Y_{0}\right)  $ where
$Q_{q}\left(  Y_{d}\right)  $ denotes the $q$-quantile of $Y_{d}$ for
$d\in\left\{  0,1\right\}  $.

Lemma 4 forms a basis for causal inferences for compliers under
Assumption 2.

\begin{lemma}
[Abadie et al. (2002)]\label{AAI2002} Given Assumption 2(i),
\[
\left(  Y_{1},Y_{0}\right)  \perp\!\!\!\perp D|X,D_{1}>D_{0}%
\]

\end{lemma}

Lemma 4 allows QTE to provide causal interpretations for
compliers. Let $Q_{q}\left(  Y|X,D,D_{1}>D_{0}\right)  $ denote the
$q$-quantile of $Y$ given $X$\ and $D$ for compliers. Then by Lemma
4,%
\[
Q_{q}\left(  Y|X,D=1,D_{1}>D_{0}\right)  -Q_{q}\left(  Y|X,D=0,D_{1}%
>D_{0}\right)
\]
represents the causal effect of smoking cessation on the $q$-quantile infant
birth weight for compliers. Now I estimate the quantile regression model based
on the following specification for the $q$-quantile of $Y$\ given $X$ and $D$
for compliers $:$ for $q\in\left(  0,1\right)  ,$
\begin{equation}
Q_{q}\left(  Y|X,D,D_{1}>D_{0}\right)  =\alpha_{q}+\beta_{q}\left(  X\right)
D+X^{\prime}\gamma_{q},\label{6.model}%
\end{equation}
where $\beta_{q}\left(  X\right)  =\beta_{1q}+X^{\prime}\beta_{2q}$,
$\beta_{q}=\left(
\begin{array}
[c]{c}%
\beta_{1q}\\
\beta_{2q}%
\end{array}
\right)  $, $\left(  \alpha_{q},\beta_{1q}\right)  \in\mathbb{R}%
\times\mathbb{R},$ $\beta_{2q}\in\mathbb{R}^{k}$ and $\gamma_{q}\in
\mathbb{R}^{k}$.

I\ use Abadie et al. (2002)'s estimation procedure. They proposed an
estimation method for moments involving $(Y,D,X)$ for compliers by using
weighted moments. See Section 3 of Abadie et al. (2002) for details about the
estimation procedure and asymptotic distribution of the estimator. Following
their estimation strategy, I\ estimate the equation (17).\footnote{I\ follow the same computation method as in Abadie et al. (2002).
They used Barrodale-Roberts (1973) linear programming algorithm for quantile
regression and a biweight kernel for the estimation of standard errors.} The
estimation results for the equation (17) are documented in Table
C.3 in Appendix C.

Smoking is estimated to have significantly negative effects on all quantiles
of birth weight. The estimated causal effect of smoking on the $q$-quantile of
infant birth weight is $-195$ grams at $q=0.15$, $-214$ grams at $q=0.25$, and
$-234$ grams at $q=0.50$.\ The effect significantly differs by women's race,
education, age, and the quality of prenatal care. This heterogeneity also
varies across quantile levels of birth weight. For the low quantiles $q=0.15$
and $0.25$, the adverse effect of smoking is estimated to be the largest for
births whose mothers are black and get inadequate prenatal care. In education,
the\ adverse smoking effect is much less severe for college graduates compared
to women with other education background.\ At $q=0.15$, as women's age
increases up to 35 years, the adverse effect of smoking becomes less severe,
but it increases with women's age for births to women who are older than 35
years old.

Controlling for the smoking status, compared to white women, black women bear
lighter babies for all quantiles and Hispanic women bear similar weight babies
at low quantiles $q=0.15,$ $0.25$ but lighter babies at higher $q>0.5. $ Also,
at low quantiles $q=0.15$ and $0.25$, as mothers' education level increases,
the birth weight noticeably increases except for post graduate women. Married
women are more likely to give births to heavier babies for low quantiles
$q=0.15,$ $0.25,$ $0.50,$ but lighter babies at high quantiles $q=0.75,0.85$.
One should be cautious about interpreting the results at high quantiles. At
high quantiles, heavier babies do not necessarily mean healthier babies
because high birth weight could be also problematic.\footnote{High birth
weight is defined as a birth weight of
$>$%
4000 grams or greater than 90 percentiles for gestational age. The causes of
HBW are gestational diabetes, maternal obesity, grand multiparity, etc. The
rates of birth injuries and infant mortality rates are higher among HBW
infants than normal birth weight infants.} The prenatal care seems to be
associated with birth weight very differently at both ends of quantiles (at
$q=0.15$ and at $q=0.85$). At $q=0.15$, women with better prenatal care tend
to have lighter babies, while at $q=0.85$ women with better prenatal care are
more likely to bear heavier infants. This suggests that women with higher
medical risk factors are more likely to have more intense prenatal care.%
\begin{table}[tbp]%

\begin{center}
\caption{Quantiles of potential outcomes and quantile treatment effects
(grams)}\label{quantile}%
\begin{tabular}
[c]{l|cccccc}\hline\hline
{\small \ (grams)} &  & ${\small Q}_{0.15}$ & ${\small Q}_{0.25}$ &
${\small Q}_{0.5}$ & ${\small Q}_{0.75}$ & ${\small Q}_{0.85}$\\\hline
{\small Entire Sample} & {\small QTE} & {\small 195} & {\small 214} &
{\small 234} & {\small 259} & {\small 292}\\
& ${\small Q}\left(  Y_{0}\right)  $ & {\small 2760} & {\small 2927} &
{\small 3220} & {\small 3515} & {\small 3675}\\
& ${\small Q}\left(  Y_{1}\right)  $ & {\small 2955} & {\small 3141} &
{\small 3454} & {\small 3774} & {\small 3967}\\\hline
{\small White} & {\small QTE} & {\small 204} & {\small 212} & {\small 212} &
{\small 227} & {\small 255}\\
& ${\small Q}\left(  Y_{0}\right)  $ & {\small 2815} & {\small 2974} &
{\small 3300} & {\small 3589} & {\small 3731}\\
& ${\small Q}\left(  Y_{1}\right)  $ & {\small 3019} & {\small 3186} &
{\small 3512} & {\small 3816} & {\small 3986}\\\hline
{\small SCCG} & {\small QTE} & {\small 109} & {\small 165} & {\small 187} &
{\small 244} & {\small 194}\\
& ${\small Q}\left(  Y_{0}\right)  $ & {\small 2908} & {\small 3031} &
{\small 3316} & {\small 3566} & {\small 3798}\\
& ${\small Q}\left(  Y_{1}\right)  $ & {\small 3017} & {\small 3196} &
{\small 3503} & {\small 3810} & {\small 3992}\\\hline
{\small Age 26-35} & {\small QTE} & {\small 233} & {\small 180} & {\small 179}
& {\small 262} & {\small 283}\\
& ${\small Q}\left(  Y_{0}\right)  $ & {\small 2781} & {\small 3008} &
{\small 3331} & {\small 3557} & {\small 3720}\\
& ${\small Q}\left(  Y_{1}\right)  $ & {\small 3014} & {\small 3188} &
{\small 3510} & {\small 3818} & {\small 4003}\\\hline
\end{tabular}

\end{center}%

\end{table}%

To estimate marginal distributions of $Y_{0}$ and $Y_{1,}$ I\ first estimate
the model (17) for a fine grid of $q\ $with $999$ points from
$0.001$ to $0.999$ and obtain quantile curves of $Y_{0}$ and $Y_{1}$ on the
fine grid. Note that fitted quantile curves are non-monotonic as shown in
Figure 16(a). I sort the estimated values of the quantile
curves in an increasing order as proposed by Chernozhukov et al. (2009). They
showed that this procedure improves the estimates of quantile functions and
distribution functions in finite samples. Figure 16(b) shows
the monotonized quantile curves for $Y_{0}$ and $Y_{1}$, respectively. The
marginal distribution functions of $Y_{0}$ and $Y_{1}$ are obtained by
inverting the monotonized quantile curves.%

\begin{figure}
[ptb]
\begin{center}
\includegraphics[
height=2.8305in,
width=6.8623in
]%
{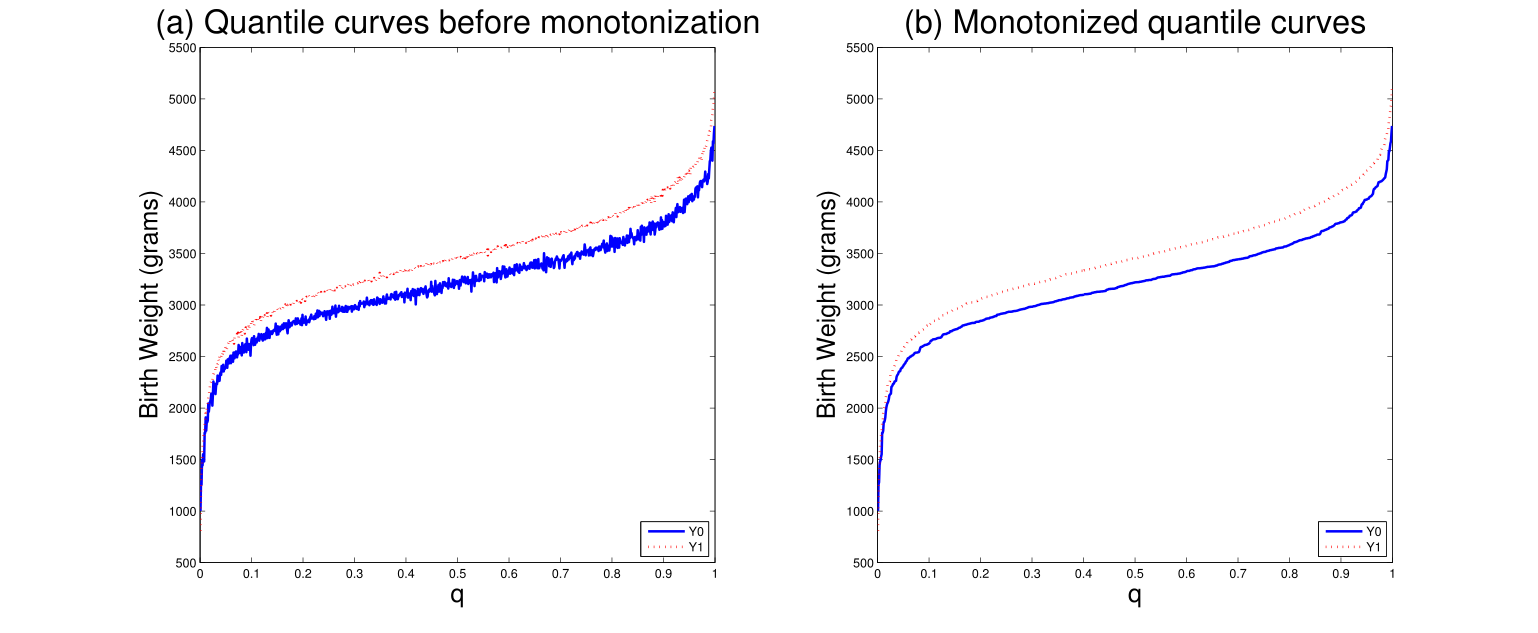}%
\caption{Estimated quantile curves}%
\label{quantilecurves}%
\end{center}
\end{figure}

Table 5 presents estimates of quantiles for potential outcomes
and QTE.\ One noticeable observation is that\ for SCCG women, low quantiles
($q<0.5$) of birth weight from smokers are remarkably higher compared to those
for the entire sample or other subgroups, while their nonsmokers' birth weight
quantiles are similar to those in other groups.\ This leads to the lower
quantile smoking effects for this college education group compared to other
groups at low quantiles.\ 

I\ also obtain the proportion of potential low birth weight infants to smokers
and nonsmokers, $F_{0}\left(  2,500\right)  $ and $F_{1}\left(  2,500\right)
$, respectively. As shown in Table 6, 6.5\% of babies to smokers
would have low birth weight, while 4\% babies to nonsmokers would have low
birth weight. Similar results are obtained for white women and women aged
26-35. A surprising result is obtained for SCCG women. Only 3.5\% \ of babies
to SCCG women who smoke would have low birth weight. This implies that SCCG
women who smoke are less likely to have low birth weight infants than women
with\ less education who smoke. One possible explanation for this is that
women with higher education are more likely to have healthier lifestyles and
this substantially lowers the risk of having low infant birth weight for smoking.%

\begin{table}[tbp] \centering

\begin{center}
\caption{The proportion of potential low birth weight infants}\label{rateofLBW}%
\begin{tabular}
[c]{l|cc}\hline\hline
(\%) & $F_{0}\left(  2,500\right)  $ & $F_{1}\left(  2,500\right)  $\\\hline
{\small Entire Sample} & 6.5 & 4\\
{\small White} & 7 & 3\\
{\small SCCG} & 3.5 & 2.9\\
{\small Age 26-35} & 5.7 & 3.2\\\hline
\end{tabular}

\end{center}%

\end{table}%
%

\begin{table}[tbp] \centering

\begin{center}
\caption{The proportion of potential low birth weight infants}\label{rateofLBW}%
\begin{tabular}
[c]{l|cc}\hline\hline
(\%) & $F_{0}\left(  2,500\right)  $ & $F_{1}\left(  2,500\right)  $\\\hline
{\small Entire Sample} & 6.5 & 4\\
{\small White} & 7 & 3\\
{\small SCCG} & 3.5 & 2.9\\
{\small Age 26-35} & 5.7 & 3.2\\\hline
\end{tabular}

\end{center}%

\end{table}%

\subsubsection{\textbf{Bounds on the Distribution and Quantiles of Treatment
Effects for Compliers}}

Recall the sharp lower bound under MTR: for $\delta\geq0,$
\begin{equation}
F_{\Delta}^{L}\left(  \delta\right)  =\underset{\left\{  a_{k}\right\}
_{k=-\infty}^{\infty}}{\sup}\sum\limits_{k=-\infty}^{\infty}\max\left\{
F_{1}\left(  a_{k+1}\right)  -F_{0}\left(  a_{k}\right)  ,0\right\}
,\label{lower}%
\end{equation}
where $0\leq a_{k+1}-a_{k}\leq\delta$ for each integer $k$. To compute the new
sharp lower bound from the estimated marginal distribution functions, I\ plug
in the estimates of marginal distribution functions $\widehat{F}_{0}$ and
$\ \widehat{F}_{1}$ proposed in the previous subsection. I\ follow the same
computation procedure as in the numerical example of Section 4. I\ discuss the
procedure in Appendix B in detail.

I\ propose the following plug-in estimators of my new lower bound and Makarov
bounds based on the estimators of marginal distributions $\widehat{F}_{0}$ and
$\ \widehat{F}_{1}$ proposed in the previous subsection.$\footnote{Fan and
Park (2010a, 2010b) proposed the same type plug-in estimators for Makarov
bounds and studied their asymptotic properties. They used empirical
distributions to estimate marginal distributions point-identified in
randomized experiments.}$ Note that the infinite sum in the lower bound under
MTR in Corollary 1 reduces to the finite sum for the bounded support.
For any fixed $\delta>0,$ the consistency of my estimators is immediate.

In Figure 17, I\ plot my new lower bound and Makarov bounds for the entire
sample. One can see substantial identification gains from the distance between
my new lower bound and the Makarov lower bound. The most remarkable
improvement arises around $q=0.5$ and the refinement gets smaller as $q$
approaches $0$ and $1,$ in turn as $\delta$ approaches $0$ and $2000$. This
can be intuitively understood through Figure 7(c). As $\delta$ gets closer to $2000,$ the number of triangles, which is one
source of identification gains, decreases to one in the bounded support of
each potential outcome. This causes the new lower bound to converge to the
Makarov lower bound as $\delta$ approaches $2000$. Also, as $\delta$ converges
to $0,$ the identification gain generated by each triangle,\ which is written
as max$\left\{  F_{1}(y)-F_{0}(y-\delta),0\right\}  ,$ converges to $0$ under
MTR, which implies $F_{1}(y)\leq F_{0}(y)$ for each $y\in\mathbb{R}.$%

\begin{center}
\includegraphics[
natheight=12.972200in,
natwidth=23.333500in,
height=2.6221in,
width=4.6942in
]%
{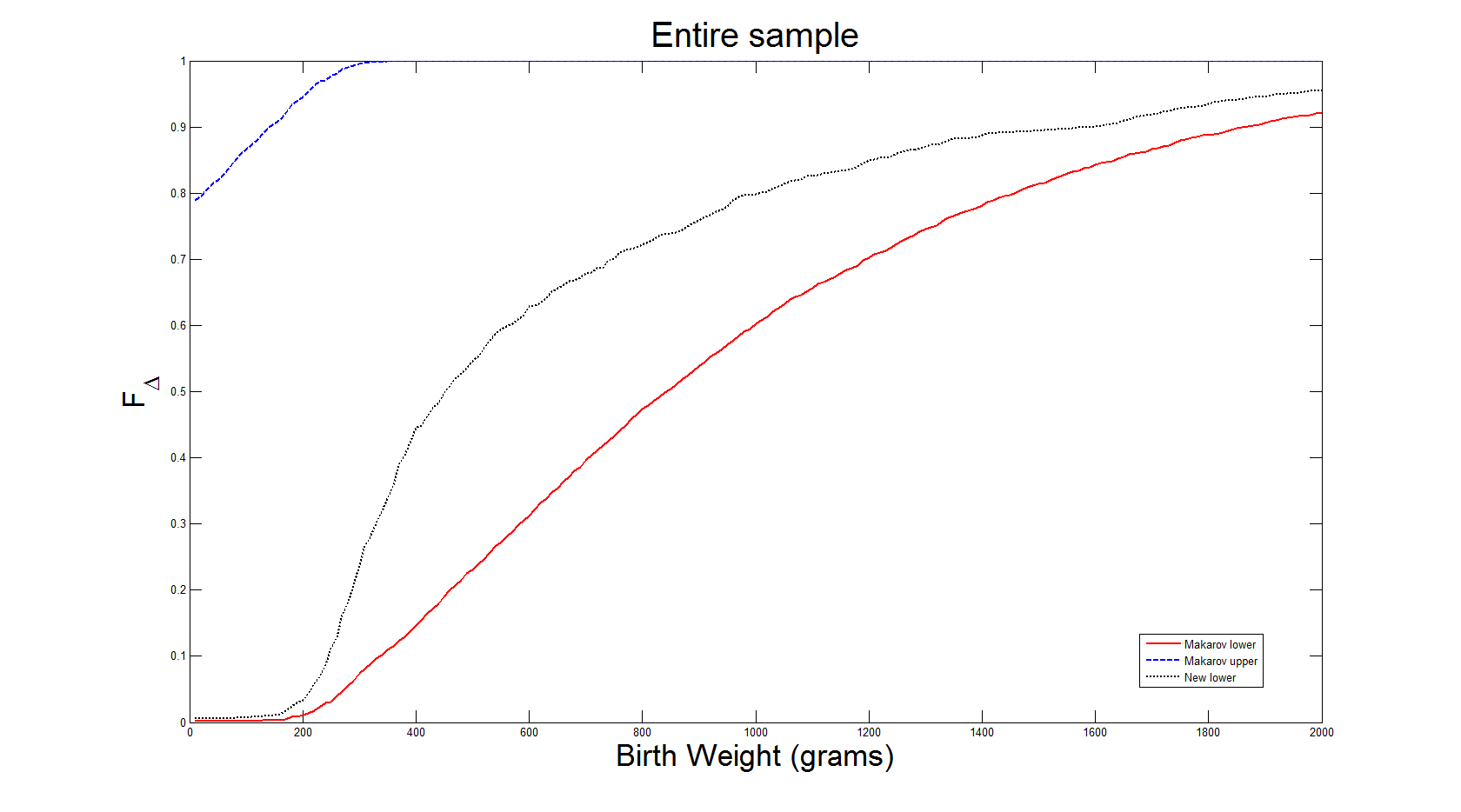}%
\end{center}

\begin{center}
Figure 17: Bounds on the effect of smoking on birth weight for the entire sample
\end{center}

The quantiles of smoking effects can be obtained by inverting these DTE
bounds. Specifically, the upper and lower bounds on the quantile of treatment
effects are obtained by inverting the lower bound and upper bound on the DTE,
respectively. Note that quantiles of smoking effects show $q$-quantiles of the
difference $(Y_{1}-Y_{0})$, while QTE gives the difference between the
$q$-quantiles of $Y_{1}$ and those of $Y_{0}$. These two parameters typically
have different values. Fan and Park (2009) pointed out that QTE\ is identical
to the quantile of treatment effects under strong
conditions.\footnote{Specifically, QTE = the quantile of treatment effects
when (i) two potential outcomes are perfectly positively dependent
$Y_{1}=F_{1}^{-1}\left(  F_{0}\left(  Y_{0}\right)  \right)  $ AND (ii)
$F_{1}^{-1}\left(  q\right)  -F_{0}^{-1}\left(  q\right)  $ is nondecreasing
in $q$.} The bounds on the quantile of treatment effects are reported in Table
7 with comparison to QTE, already reported in Table 5. In the
entire sample, my new bounds on the quantiles of the treatment effect show
$33$\% - $45$\% refinement\ for $q=0.15$, $0.25$, $0.5$, $0.75$ compared to
Makarov bounds. For the entire sample, my new bounds yield $[0,457]$ grams for
the median of the benefit of smoking cessation on infant birth weight, while
Makarov bounds yield $[0,843]$ grams. Compared to Makarov bounds, my new
bounds are more informative and show that $(457,843]$ should be excluded from
the identification region for the median of the effect.

It is worth noting that my new bounds on the quantile of the effects of
smoking are much tighter for SCCG women, compared to the entire sample and
other subsamples. For $q\leq0.5$, the refinement rate ranges from 51\% to 64\%
compared to Makarov bounds. For SCCG women, my new sharp bounds on the median
are $[0,299]$ grams, while Makarov bounds on the median are $[0,764]$ grams.
The higher identification gains result from relatively heavier potential
nonsmokers' infant birth weight, which leads to the shorter distance between
two potential outcomes distributions as reported in Table 5. Note
that the shorter distance between marginal distributions of potential outcomes
improves both my new lower bound and the Makarov lower bound.\footnote{To
develop intuition, recall Figure 7(c). The size
of the lower bound on each triangle's probability is related to the distance
between marginal distribution functions of $Y_{0}$ and $Y_{1}$. To see this,
consider two marginal distribution functions $F_{1}^{A}$ and $F_{1}^{B}$ of
$Y_{1}$ with $F_{1}^{A}\left(  y\right)  \leq F_{1}^{B}\left(  y\right)  $ for
all $y\in\mathbb{R}$\ and fix the marginal distribution $F_{0}$ of $Y_{0}$
where $\left(  Y_{0},Y_{1}\right)  $ satisfies MTR. Since MTR implies
stochastic dominance of $Y_{1}$ over $Y_{0}$ for each $y\in\mathbb{R}%
,F_{1}^{A}\left(  y\right)  <F_{1}^{B}\left(  y\right)  \leq F_{0}\left(
y\right)  .$Thus,%
\[
\max\left\{  F_{1}^{A}\left(  y\right)  -F_{0}\left(  y-\delta\right)
,0\right\}  <\max\left\{  F_{1}^{B}\left(  y\right)  -F_{0}\left(
y-\delta\right)  ,0\right\}  .
\]
\ Since the probability lower bound on the triangle is written as
$\max\left\{  F_{1}\left(  y\right)  -F_{0}\left(  y-\delta\right)  \right\}
$ for some $y\in\mathbb{R},$ the above inequality shows that the closer
marginal distributions $F_{0}$ and $F_{1}$ generates higher probability lower
bound on each triangle.}

\begin{center}
Table 7: QTE and bounds on the quantiles of smoking effects%

\begin{tabular}
[c]{l|llllll}\hline\hline
{\small Dep. var.= Birth weight (grams) } &  & ${\small Q}_{0.15}$ &
${\small Q}_{0.25}$ & ${\small Q}_{0.5}$ & ${\small Q}_{0.75}$ &
${\small Q}_{0.85}$\\\hline
{\small Entire Sample} & {\small QTE} & \multicolumn{1}{c}{{\small 195}} &
\multicolumn{1}{c}{{\small 214}} & \multicolumn{1}{c}{{\small 234}} &
\multicolumn{1}{c}{{\small 259}} & \multicolumn{1}{c}{{\small 292}}\\
& {\small Makarov} & \multicolumn{1}{c}{{\small [0,405]}} &
\multicolumn{1}{c}{{\small [0,524]}} & \multicolumn{1}{c}{{\small [0,843]}} &
\multicolumn{1}{c}{{\small [0,1317]}} & \multicolumn{1}{c}{{\small [80,1634]}%
}\\
& {\small New} & \multicolumn{1}{c}{{\small [0,265]}} &
\multicolumn{1}{c}{{\small [0,304]}} & \multicolumn{1}{c}{{\small [0,457]}} &
\multicolumn{1}{c}{{\small [0,882]}} & \multicolumn{1}{c}{{\small [80,1204]}%
}\\\hline
{\small White} & {\small QTE} & \multicolumn{1}{c}{{\small 204}} &
\multicolumn{1}{c}{{\small 212}} & \multicolumn{1}{c}{{\small 212}} &
\multicolumn{1}{c}{{\small 227}} & \multicolumn{1}{c}{{\small 255}}\\
& {\small Makarov} & \multicolumn{1}{c}{{\small [0,383]}} &
\multicolumn{1}{c}{{\small [0,505]}} & \multicolumn{1}{c}{{\small [0,833]}} &
\multicolumn{1}{c}{{\small [0,1274]}} & \multicolumn{1}{c}{{\small [65,1588]}%
}\\
& {\small New} & \multicolumn{1}{c}{{\small [0,265]}} &
\multicolumn{1}{c}{{\small [0,308]}} & \multicolumn{1}{c}{{\small [0,450]}} &
\multicolumn{1}{c}{{\small [0,891]}} & \multicolumn{1}{c}{{\small [65,1239]}%
}\\\hline
{\small SCCG} & {\small QTE} & \multicolumn{1}{c}{{\small 109}} &
\multicolumn{1}{c}{{\small 165}} & \multicolumn{1}{c}{{\small 187}} &
\multicolumn{1}{c}{{\small 244}} & \multicolumn{1}{c}{{\small 194}}\\
& {\small Makarov} & \multicolumn{1}{c}{{\small [0,311]}} &
\multicolumn{1}{c}{{\small [0,428]}} & \multicolumn{1}{c}{{\small [0,764]}} &
\multicolumn{1}{c}{{\small [0,1183]}} & \multicolumn{1}{c}{{\small [69,1453]}%
}\\
& {\small New} & \multicolumn{1}{c}{{\small [0,114]}} &
\multicolumn{1}{c}{{\small [0,193]}} & \multicolumn{1}{c}{{\small [0,299]}} &
\multicolumn{1}{c}{{\small [0,579]}} & \multicolumn{1}{c}{{\small [69,792]}%
}\\\hline
{\small Age 26-35} & {\small QTE} & \multicolumn{1}{c}{{\small 233}} &
\multicolumn{1}{c}{{\small 180}} & \multicolumn{1}{c}{{\small 179}} &
\multicolumn{1}{c}{{\small 262}} & \multicolumn{1}{c}{{\small 283}}\\
& {\small Makarov} & \multicolumn{1}{c}{{\small [0,336]}} &
\multicolumn{1}{c}{{\small [0,458]}} & \multicolumn{1}{c}{{\small [0,807]}} &
\multicolumn{1}{c}{{\small [0,1324]}} & \multicolumn{1}{c}{{\small [79,1621]}%
}\\
& {\small New} & \multicolumn{1}{c}{{\small [0,239]}} &
\multicolumn{1}{c}{{\small [0,276]}} & \multicolumn{1}{c}{{\small [0,406]}} &
\multicolumn{1}{c}{{\small [0,746]}} & \multicolumn{1}{c}{{\small [79,1204]}%
}\\\hline
\end{tabular}

\end{center}

Although QTE is placed within the identification region for $q=0.15$ to $0.85
$ and for all groups, at $q=0.15$, QTE is very close to the upper bound on the
quantile of smoking effects for SCCG and age 26-35 subgroups. Furthermore, at
$q=0.10$, QTE is placed outside of the improved identification region for SCCG
group and age 26-35. This implies that QTE is not identical to the quantile of
treatment effects in my example and so one should not interpret the value of
QTE as a quantile of smoking effects.

Despite the large improvement of my bounds over Makarov bounds, the difference
in the quantiles of the smoking effects between SCCG women and others is still
inconclusive from my bounds. The sharp upper bound on the quantile of the
effect for the SCCG group is quite lower than that for the entire sample while
the sharp lower bound is $0$ for both groups; the identification region for
the SCCG group is contained in that for the entire sample. Since the two
identification regions overlap, one cannot conclude that the effect at each
quantile level $q$ is smaller for the SCCG group. This can be further
investigated by developing formal test procedures for the partially identified
quantile of treatment effects or by establishing tighter bounds under
additional plausible restrictions. I\ leave these issues for future research.

My empirical analysis shows that smoking is on average more dangerous for
infants to women with a higher tendency to smoke. Also, women with SCCG are
less likely to have low birth weight babies when they smoke. The estimated
bounds on the median of the effect of smoking on infant birth weight are
$[-457$,$0]$ grams and $[-299,0]$ grams for the entire sample and for women
with SCCG, respectively.

Based on my observations, I\ suggest that policy makers pay particular
attention to smoking women with low education in their antismoking policy
design, since these women's infants are more likely to have low weight.
Considering the association between higher education and better personal
health care as shown in Park and Kang (2008), it appears that smoking on
average does less harm to infants to mothers with a healthier lifestyle. Based
on this interpretation, healthy lifestyle campaigns need to be combined with
antismoking campaigns to reduce the negative effect of smoking on infant birth weight.

\subsection{Testability and Inference on the Bounds}

\subsubsection{Testability of MTR}

My empirical analysis relies on the assumption that smoking of pregnant women
has nonpositive effects on infant birth weight with probability one. This MTR
assumption is not only plausible but also testable in my setup. While a formal
econometric test procedure is beyond the scope of this paper, I\ briefly
discuss testable implications. First, MTR implies stochastic dominance of
$Y_{1}$ over $Y_{0}$. Since I\ point-identify their marginal distributions for
compliers, stochastic dominance can be checked from the estimated marginal
distribution functions. Except for very low $q$-quantiles with $q<0.006$ where
the quantile curves estimates are imprecise as noted in subsection 5.4, my
estimated marginal distribution functions satisfy the stochastic dominance for
the entire sample and all subgroups. Second, under MTR my new lower bound
should be lower than the Makarov upper bound. If MTR is not satisfied, then my
new lower bound is not necessarily lower than the Makarov upper bound. In my
estimation result, my new lower bound is lower than the Makarov upper bound
for all $\delta>0$ and in all subgroups.

\subsubsection{Inference and Bias Correction}

Asymptotic properties of my estimators other than consistency have not been
covered in this paper. The complete asymptotic theory for the estimators can
be investigated by adopting arguments from Abadie et al. (2002), Koenker and
Xiao (2002), Angrist et al. (2005), and Fan and Park (2010).\ Abadie et al.
(2002) provided asymptotic properties for their weighted quantile regression
coefficients for the fixed quantile level $q,$ while Koenker and Xiao (2002)
and Angrist et al. (2005) focused on the standard quantile regression
$\emph{process}.$ Fan and Park (2010) derived asymptotic properties for the
plug-in estimators of Makarov bounds. Since they estimated marginal
distribution functions using empirical distributions in the context of
randomized experiments, their arguments follow standard empirical process
theory. To investigate asymptotic properties of the bounds estimators and the
estimated maximizer or minimizer for the bounds, I\ am currently
extending\ the asymptotic analysis on the quantile regression process
presented by Koenker and Xiao (2002) and Angrist et al. (2005) to the quantile
curves which are obtained from the weighted quantile regression of Abadie et
al. (2002).

Canonical bootstrap procedures may be invalid for inference in this setting.
Fan and Park (2010) found that asymptotic distributions of their plug-in
estimators for Makarov bounds discontinuously change around the boundary where
the true lower and upper Makarov bounds reach zero and one, respectively.
Specifically, they estimated the Makarov lower bound $\underset{y}{\sup}%
\max\left\{  F_{1}\left(  y\right)  -F_{0}\left(  y-\delta\right)  ,0\right\}
$ using empirical distribution functions $\widehat{F}_{0}$ and $\widehat
{F}_{1}$. They found that the asymptotic distribution of their estimator of
the Makarov lower bound is discontinuous on the boundary where $\sup
_{y}\left\{  F_{1}\left(  y\right)  -F_{0}\left(  y-\delta\right)  \right\}
=0.$ Since my improved lower bound under MTR is written as the supremum of the
sum of $\max\left\{  F_{1}\left(  a_{k}\right)  -F_{0}\left(  a_{k-1}\right)
,0\right\}  $ over integers $k$, the asymptotic distribution of my plug-in
estimator is likely to suffer discontinuities near multiple boundaries where
$F_{1}\left(  a_{k}\right)  -F_{0}\left(  a_{k-1}\right)  =0$ for each integer
$k$. To avoid the failure of the standard bootstrap, I\ recommend subsampling
or the fewer than $n$ bootstrap procedure following Politis et al. (1999),
Andrews (2000), Andrews and Han (2009).

Although the estimator $\widehat{F}_{\Delta}^{NL}$ is consistent, it may have a nonnegligible bias in small
samples.\footnote{Since $\max\left(  x,0\right)  $ is a convex function, by
Jensen's inequality my plug-in estimator is upward biased. This has been also
pointed out in Fan and Park (2009) for their estimator of Makarov bounds.}
I\ suggest that one use a bias-adjusted estimator based on subsampling when
the sample size is small in practice. Let
\[
\widehat{F}_{\Delta,n,b,j}^{NL}\left(  \delta\right)  =\underset{0\leq
y\leq\delta}{\sup}\sum_{k=\left\lfloor \frac{500-y}{\delta}\right\rfloor
}^{\left\lfloor \frac{5500-y}{\delta}\right\rfloor +1}\max\left(  \widehat
{F}_{1}^{n,b,j}\left(  y+k\delta\right)  -\widehat{F}_{0}^{n,b,j}\left(
y+\left(  k-1\right)  \delta\right)  ,0\right)  ,
\]
where for $d=0,1,$ $\widehat{F}_{d}^{n,b,j}$ is an estimator of $F_{d}$ from
the $j$th subsample $\{\left(  Y_{j_{1}},D_{j_{1}}\right)  ,...,\left(
Y_{j_{b}},D_{j_{b}}\right)  \}$ with the subsample size $b$ out of $n$
observations s.t. $j_{1}\neq$ $j_{2}\neq\ldots\neq j_{b},$ $b<n$ and
$j=1,...,\left(
\begin{array}
[c]{c}%
n\\
b
\end{array}
\right)  $. Then the subsampling bias-adjusted estimator $\widetilde
{F}_{\Delta}^{NL}\left(  \delta\right)  $ is
\begin{align*}
\widetilde{F}_{\Delta}^{NL}\left(  \delta\right)   & =\widehat{F}_{\Delta
}^{NL}\left(  \delta\right)  -\frac{1}{q_{n}}\sum\limits_{j=1}^{q_{n}}\left\{
\widehat{F}_{\Delta,n,b,j}^{NL}\left(  \delta\right)  -\widehat{F}_{\Delta
}^{NL}\left(  \delta\right)  \right\} \\
& =2\widehat{F}_{\Delta}^{NL}\left(  \delta\right)  -\frac{1}{q_{n}}%
\sum\limits_{j=1}^{q_{n}}\widehat{F}_{\Delta,n,b,j}^{NL}\left(  \delta\right)
,
\end{align*}
where $q_{n}=\left(
\begin{array}
[c]{c}%
n\\
b
\end{array}
\right)  .$

\section{Conclusion}

In this paper, I have proposed a novel approach to identifying the DTE under
general support restrictions on the potential outcomes. My approach involves
formulating the problem as an optimal transportation linear program and
embedding support restrictions into the cost function with an infinite
Lagrange multiplier by taking advantage of their linearity in the entire joint
distribution. I have developed the dual formulation for $\{0,1,\infty
\}$-valued costs to overcome the technical challenges associated with
optimization over the space of joint distributions. This contrasts sharply
with the existing copula approach, which requires one to find out the joint
distributions achieving sharp bounds given restrictions.

I have characterized the identification region under general support
restrictions and derived sharp bounds on the DTE for economic examples. My
identification result has been applied to the empirical analysis of the
distribution of smoking effects on infant birth weight. I have proposed an
estimation procedure for the bounds. The empirical results have shown that MTR
has a substantial power to identify the distribution of smoking effects when
the marginal distributions of the potential outcomes are given.

In some cases, information concerning the relationship between potential
outcomes cannot be represented by support restrictions. Moreover, it is also
sometimes the case that the joint distribution function itself is of interest.
In a companion paper, I propose a method to identify the DTE and the joint
distribution when weak stochastic dependence restrictions among unobservables
are imposed in triangular systems, which consist of an outcome equation and a
selection equation.

\bibliographystyle{chicago}
\bibliography{acompat,camel}

\begin{thebibliography}{999999999}                                                                                        %
\bibitem[AHV2005]{AHV2005}Aakvik, A., J. Heckman and E. Vytlacil (2005).
\textquotedblleft Estimating Treatment Effects for Discrete Outcomes When
Responses to Treatment Vary among Observationally Identical Persons: An
Application to Norwegian Vocational Rehabilitation Programs,\textquotedblright%
\ \emph{Journal of Econometrics}, 125, 15--51.

\bibitem[A2002]{A2002}Abadie, A. (2002). \textquotedblleft Bootstrap Tests for
Distributional Treatment Effects in Intrumental Variable
Models,\textquotedblright\ \emph{Journal of the American Statistical
Association}, 97, 284-292.

\bibitem[A2003]{A2003}Abadie, A. (2003). \textquotedblleft Semiparametric
Instrumental Variable Estimation of Treatment Response
Models,\textquotedblright\ \emph{Journal of Econometrics}, 113, 231--263.

\bibitem[AAI2002]{AAI2002}Abadie, A., J. Angrist and G. Imbens (2002).
\textquotedblleft Instrumental Variables Estimates of the Effect of Subsidized
Training on the Quantiles of Trainee Earnings,\textquotedblright%
\ \emph{Econometrica}, 70, 91-117.

\bibitem[AH2007]{AH2007}Abbring, J. H. and J. Heckman (2007).
\textquotedblleft Econometric evaluation of social programs, Part III:
Distributional treatment effects, dynamic treatment effects, dynamic discrete
choice, and general equilibrium policy evaluation,\textquotedblright%
\ $Handbook$ $of$ $Econometrics$, 6B, 5145--5301.

\bibitem[A2006]{A2006}Abrevaya, J. (2006). \textquotedblleft Estimating the
Effect of Smoking on Birth Outcomes Using a Matched Panel Data
Approach,\textquotedblright\ \emph{Journal of Applied Econometrics}, 21, 489-519.

\bibitem[AD2008]{AD2008}Abrevaya, J. and C. Dahl (2008). \textquotedblleft The
Effects of Birth Inputs on Birthweight," \emph{Journal of Business and
Economic Statistics}, 26, 379-397.

\bibitem[AP2012]{AP2012}Abrevaya, J. and L. Puzzello (2012). "Taxes, Cigarette
Consumption, and Smoking Intensity: Comment," \emph{American Economic Review},
102, 1751-1763.

\bibitem[AC2006]{AC2006}Adda, J. and F. Cornaglia (2006). \textquotedblleft
Taxes, Cigarette Consumption, and Smoking Intensity,"\emph{\ American Economic
Review}, 96, 1013-1028.

\bibitem[AC2011]{AC2011}Almond, D. and J. Currie (2011). \textquotedblleft
Human Capital Development Before Age Five," \emph{The Handbook of Labor
Economics}, 4, 2011, 1316-1486.

\bibitem[ACL2005]{ACL2005}Almond, D., K. Chay and D. Lee (2005).
\textquotedblleft The Costs of Low Birth Weight,\textquotedblright\ \emph{The
Quarterly Journal of Economics},120 (3), 1031--1083.

\bibitem[A2000]{A2000}Andrews, D. W. K. (2000). \textquotedblleft
Inconsistency of the Bootstrap when a Parameter is on the Boundary of the
Parameter Space," 68, 399--405.

\bibitem[AH2009]{AH2009}Andrews, D. W. K and S. Han. (2009). \textquotedblleft
Invalidity of the Bootstrap and the m out of n Bootstrap for Confidence
Interval Endpoints Defined by Moment Inequalities,\textquotedblright%
\ \emph{Econometrics Journal}, 12, 172--199.

\bibitem[AG2009a]{AG2009a}Andrews, D.W.K. and P. Guggenberger (2009a).
\textquotedblleft Incorrect Asymptotic Size of Subsampling Procedures Based on
Post- Consistent Model Selection Estimators,\textquotedblright\ \emph{Journal
of Econometrics}, 152, 19-27.

\bibitem[AG2009b]{AG2009b}--------- (2009b). \textquotedblleft Hybrid and
size-corrected subsampling methods,\textquotedblright\ \emph{Econometrica},
77, 721-762.

\bibitem[AB2012]{AB2012}Arellano, M. and S. Bonhomme (2012). \textquotedblleft
Identifying Distributional Characteristics in Random Coefficients Panel Data
Models,\textquotedblright\ \emph{Review of Economic Studies}, 79, 987-1020.

\bibitem[BLR2008]{BLR2008}Bandiera, O., V. Larcinese and I. Rasul (2008).
\textquotedblleft Heterogeneous Class Size Effects: New Evidence from a Panel
of University Students,\textquotedblright\ \emph{Economic Journal}, 120, 1365--1398.

\bibitem[BSV2008]{BSV2008}Bhattacharya, J., A. Shaikh and E. Vytlacil (2008).
\textquotedblleft Treatment Effect Bounds under Monotonicity Assumptions: An
Application to Swan-Ganz Catheterization,$^{\text{\textquotedblright}}$
\emph{American Economic Review}, 98, 351-56.

\bibitem[BSV2012]{BSV2012}Bhattacharya, J., A. Shaikh and E. Vytlacil (2012).
\textquotedblleft Treatment Effect Bounds: An Application to Swan--Ganz
Catheterization,\textquotedblright\ \emph{Journal of Econometrics}, 168, 223-243.

\bibitem[B2010]{B2010}Boes, S. (2010). \textquotedblleft Convex Treatment
Response and Treatment Selection,\textquotedblright\ SOI Working Paper 1001,
University of Zurich.

\bibitem[BNW2006]{BNW2006}Byrd, R., J. Nocedal and R. Waltz (2006).
\textquotedblleft KNITRO: AN Integrated Package for Nonlinear
Optimization,\textquotedblright\ in Large-Scale Nonlinear Optimization,
Springer Verlag.

\bibitem[C2010]{C2010}Carlier, G. (2010). \textquotedblleft Optimal
Transportation and Economic Applications,\textquotedblright\ Letcure Notes.

\bibitem[C2012]{C2012}Caetano, C. (2012). \textquotedblleft A Test of
Endogeneity without Instrumental Variables,\textquotedblright\ Working Paper.

\bibitem[CC2012]{CC2012}Conlon, C. (2012). \textquotedblleft A Dynamic Model
of Prices and Margins in the LCD TA\ Industry,\textquotedblright\ Working Paper.

\bibitem[CHH2003]{CHH2003}Carneiro, P., K. T. Hansen and J. Heckman (2003).
\textquotedblleft Estimating Distributions of Treatment Effects with an
Application to the Returns to Schooling and Measurement of the Effects of
Uncertainty on College Choice,\textquotedblright\ \emph{International Economic
Review}, 44, 361--422.

\bibitem[CH2005]{CH2005}Chernozhukov, V. and C. Hansen (2005),
\textquotedblleft An IV Model of Quantile Treatment Effects,\textquotedblright%
\ \emph{Econometrica}, 73, 245--261.

\bibitem[CCH2010]{CCH2010}Chernozhukov, V., P.-A. Chiappori and M. Henry
(2010). \textquotedblleft Introduction,\textquotedblright\ \emph{Economic
Theory}, 42, 271-273.

\bibitem[CFG2010]{CFG2010}Chernozhukov, V., I. Fern\'{a}ndez-Val and A.
Galichon (2010). \textquotedblleft Quantile and Probability Curves without
Crossing,\textquotedblright\ \emph{Econometrica}, 78, 1093-1125.

\bibitem[CLR2013]{CLR2013}Chernozhukov, V., S. Lee and A. M. Rosen (2013).
\textquotedblleft Intersection Bounds: Estimation and
Inference,\textquotedblright\ \emph{Econometrica}, 81, 667-737.

\bibitem[CMN2010]{CMN2010}Chiappori, P.-A., R. J. McCann and L. P. Nesheim
(2010). \textquotedblleft Hedonic Price Equilibria, Stable Matching, and
Optimal Transport: Equivalence, Topology, and Uniqueness,\textquotedblright%
\ \emph{Economic Theory}, 42, 317--354.

\bibitem[C1999]{C1999}Currie, J. and R. Hyson (1999). \textquotedblleft Is the
Impact of Shocks Cusioned by Socioeconomic Status? The Case of Low Birth
Weight,\textquotedblright\ \emph{American Economic Review}, 89 (2), 245--250.

\bibitem[CM2007]{CM2007}Currie, J. and E. Moretti (2007). \textquotedblleft
Biology as Destiny? Short- and Long-run Determinants of Intergenerational
Transmission of Birth Weight,\textquotedblright\ \emph{Journal of Labor
Economics,} 25 (2), 231--264.

\bibitem[D2003]{D2003}Deaton, A. (2003). \textquotedblleft Health, Inequality,
and Economic Development,\textquotedblright\ \emph{Journal of Economic
Literature}, 41, 113-158.

\bibitem[DL2008]{DL2008}Ding, W. and S. Lehrer (2008). \textquotedblleft Class
Size and Student Achievement: Experimental Estimates of Who Benefits and Who
Loses from Reductions.\textquotedblright\ Queen's Economic Department Working
Paper No. 1046, Queen's University.

\bibitem[DFS2012]{DFS2012}Dub\'{e}, J.-P., J. Fox and C.-L. Su (2012).
\textquotedblleft Improving the Numerical Performance of Static and Dynamic
Aggregate Discrete Choice Random Coefficients Demand
Estimation,\textquotedblright\ \emph{Econometrica}, 80, 2231-2367.

\bibitem[DDK2011]{DDK2011}Duflo, E., P. Dupas, and M. Kremer (2011).
\textquotedblleft Peer Effects, Teacher Incentives, and the Impact of
Tracking: Evidence from a Randomized Evaluation in Kenya,\textquotedblright%
\ \emph{American Economic Review}, 101, 1739-74.

\bibitem[E2005]{E2005}Ekeland, I. (2005). \textquotedblleft An Optimal
Matching Problem,\textquotedblright\ ESAIM Controle Optimal et Calcul des
Variations, 11, 57-51.

\bibitem[E2010]{E2010}Ekeland, I. (2010). \textquotedblleft Existence,
Uniqueness, and Efficiency of Equilibrium in Hedonic Markets with
Multidimensional Types,\textquotedblright\ \emph{Economic Theory}, 42, 275-315.

\bibitem[EGH2010]{EGH2010}Ekeland, I., A. Galichon and M. Henry (2010).
\textquotedblleft Optimal Transportation and the Falsifiability of
Incompletely Specified Economic Models,\textquotedblright\ \emph{Economic
Theory}, 42, 355-374.

\bibitem[ER1999]{ER1999}Evans, W. and J. S. Ringel (1999). \textquotedblleft
Can Higher Cigarette Taxes Improve Birth Outcomes?,\textquotedblright%
\ \emph{Jounal of Public Economics}, 72, 135--154.

\bibitem[FP2009]{FP2009}Fan, Y. and S. S. Park (2009). \textquotedblleft
Partial Identification of the Distribution of Treatment Effects and its
Confidence Sets,\textquotedblright\ in Thomas B. Fomby and R. Carter Hill
(ed.) Nonparametric Econometric Methods (Advances in Econometrics, Volume 25),
Emerald Group Publishing Limited, pp.3-70.

\bibitem[FP2010]{FP2010}--------- (2010). \textquotedblleft Sharp Bounds on
the Distribution of Treatment Effects and Their Statistical
Inference,\textquotedblright\ \emph{Econometric Theory}, 26, 931-951.

\bibitem[FT2011]{FT2011}French, E. and C. Taber (2011). \textquotedblleft
Identification of Models of the Labor Market,\textquotedblright%
\emph{\ Handbook of Labor Economics}, 4, 537-617.

\bibitem[FW2010]{FW2010}Fan, Y. and J. Wu (2010). \textquotedblleft Partial
Identification of the Distribution of Treatment Effects in Switching Regime
Models and Its Confidence Sets,\textquotedblright\ \emph{Review of Economic
Studies}, 77, 1002-1041.

\bibitem[FKK1990]{FKK1990}Fingerhut, L. A., J. C. Kleinman and J. S. Kendrick
(1990). \textquotedblleft Smoking before, during, and after
Pregnancy,\textquotedblright\ \emph{American Journal of Public Health}, 80, 541--544.

\bibitem[F2007]{F2007}Firpo, S. (2007). \textquotedblleft Efficient
Semiparametric Estimation of Quantile Treatment Effects\textquotedblright,
\emph{Econometrica}, 75, 259--276.

\bibitem[FR2008]{FR2008}Firpo, S., and G. Ridder (2008). \textquotedblleft
Bounds on Functionals of the Distribution of Treatment
Effects,\textquotedblright\ Working paper, FGV Brazil.

\bibitem[FNS1987]{FNS1987}Frank, M. J., R. B. Nelson and B. Schweizer (1987).
\textquotedblleft Best-possible Bounds for the Distribution of a Sum-a Problem
of Kolmogorov,\textquotedblright\ \emph{Probability Theory Related Fields},
74, 199-211.

\bibitem[GH2008]{GH2008}Galichon, A. and M. Henry. (2008). \textquotedblleft
Inference in Incomplete Models,\textquotedblright\ Working Paper.

\bibitem[GH2011]{GH2011}Galichon, A. and M. Henry. (2011). \textquotedblleft
Set Identification in Models with Multiple Equilibria,\textquotedblright%
\ \emph{Review of Economic Studies}, 78, 1264-1298.

\bibitem[GS2012]{GS2012}Galichon, A. and B. Salani\'{e} (2012).
\textquotedblleft Cupid's Invisible Hand: Social Surplus and Identification in
Matching Models,\textquotedblright\ Working Paper.

\bibitem[GKP2011]{GKP2011}Gundersen, C., B. Kreider and J. Pepper. (2011).
\textquotedblleft The Impact of the National School Lunch Program on Child
Health:\ A Nonparametric Bounds Analysis,\textquotedblright\ \emph{Journal of
Econometrics}, 166, 79-91.

\bibitem[H2012]{H2012}Haan, M. (2012). \textquotedblleft The Effect of
Additional Funds for Low-Ability Pupils - A Nonparametric Bounds
Analysis,\textquotedblright\ CESifo Working Paper.

\bibitem[H1990]{H1990}Heckman, J. J. (1990). \textquotedblleft Varieties of
Selection Bias,\textquotedblright\ \emph{American Economic Review}, Papers and
Proceedings, 80, 313--318.

\bibitem[HSC1997]{HSC1997}Heckman, J. J., J. A. Smith and N. Clements (1997).
"Making the Most Out of Programme Evaluations and Social Experiments:
Accounting for Heterogeneity in Programme Impacts,\textquotedblright%
\ \emph{Review of Economic Studies}, 64, 487--535.

\bibitem[HV2007]{HV2007}Heckman, J. J. and E. Vytlacil (2007).
\textquotedblleft Econometric Evaluation of Social Programs, part II: Causal
Models, Structural Models and Econometric Policy Evaluation,\textquotedblright%
\ \emph{Handbook of Econometrics}, 6, 4779-4874.

\bibitem[HEV2011]{HEV2011}Heckman, J. J., P. Eisenhauer and E. Vytlacil
(2011). \textquotedblleft Generalized Roy Model \ and Cost-Benefit Analysis of
Social Programs,\textquotedblright\ Working Paper.

\bibitem[HM2012]{HM2012}Henry, M. and I. Mourifi\'{e} (2012).
\textquotedblleft Sharp Bounds in the Binary Roy Model,\textquotedblright%
\ Working Paper.

\bibitem[HIR2003]{HIR2003}Hirano, K., G. Imbens and G. Ridder (2003).
\textquotedblleft Efficient Estimation of Average Treatment Effects Using the
Estimated Propensity Score,\textquotedblright\ \emph{Econometrica}, 71, 1161-1189.

\bibitem[HS2012]{HS2012}Hodelein, S. and Y. Sasaki (2013). \textquotedblleft
Outcome Conditioned Treatment Effects,\textquotedblright\ CEMMAP Working Paper
CWP 39/13.

\bibitem[IA1994]{IA1994}Imbens, G. W. and J. D. Angrist (1994).
\textquotedblleft Identification and Estimation of Local Average Treatment
Effects,\textquotedblright\ \emph{Econometrica}, 62, 467--75.

\bibitem[IR1997]{IR1997}Imbens G. W., and D. B. Rubin (1997).
\textquotedblleft Estimating Outcome Distributions for Compliers in
Instrumental Variables Models,\textquotedblright\ \emph{Review of Economic
Studies}, 64, 555--574.

\bibitem[IW2009]{IW2009}Imbens, G. W. and J. M. Wooldridge (2009).
\textquotedblleft Recent Developments in the Econometrics of Program
Evaluation,\textquotedblright\ \emph{Journal of Economic Literature}, 47, 5--86.

\bibitem[JLS2013]{JLS2013}Jun, S. J., \ Y. Lee and Y. Shin (2013).
\textquotedblleft Testing for Distributional Treatment Effects: A Set
Identification Approach,\textquotedblright\ Working Paper.

\bibitem[KP1996]{KP1996}Koenker, R. and B. Park (1996). \textquotedblleft An
Interior Point Algorithm for Nonlinear Quantile Regression,\textquotedblright%
\ \emph{Journal of Econometrics}, 71, 265-283.

\bibitem[LACWC2013]{LACWC2013}Lau, C., N. Ambalavanan, H. Chakraboty, M. S.
Wingate and W.A. Carlo (2013). \textquotedblleft Extremely Low Birth Weight
and Infant Mortality Rates in the United States,\textquotedblright%
\ \emph{Pediatrics}, 131, 855-860.

\bibitem[LE(2005)]{LE(2005)}Lien, D. S. and W. N. Evans (2005).
\textquotedblleft Estimating the Impact of Large Cigarette Tax Hikes: The Case
of Maternal Smoking and Infant Birth Weight,\textquotedblright\ \emph{Journal
of Human Resources}, 40, 373-392.

\bibitem[MH1994]{MH1994}Mainous, A. G. and W.J. Hueston (1994).
\textquotedblleft The Effect of Smoking Cessation during Pregnancy on Preterm
Delivery and Low Birthweight,\textquotedblright\ \emph{The Journal of Family
Practice}, 38, 262-266.

\bibitem[M1981]{M1981}Makarov, G. D. (1981). \textquotedblleft Estimates for
the Distribution Function of a Sum of Two Random Variables when the Marginal
Distributions are Fixed,\textquotedblright\ \emph{Theory of Probability and
its Applications},\ 26, 803-806.

\bibitem[M1997]{M1997}Manski, C. F. (1997). \textquotedblleft Monotone
Treatment Response,\textquotedblright\ \emph{Econometrica}, 65, 1311-1334.

\bibitem[M2000]{M2000}Manski, C. F. and J. Pepper (2000). \textquotedblleft
Monotone Instrumental Variables: With an Application to the Returns to
Schooling,\textquotedblright\ \emph{Econometrica}, 68, 997-1010.

\bibitem {M2007}Massachusetts Department of Public Health (2007).
\textquotedblleft Accomplishments of the Massachusetts Tobacco Control
Program,\textquotedblright\ online at http://www.mass.gov/eohhs/docs/dph/tobacco-control/accomplishments.pdf.

\bibitem[M1781]{M1781}Monge, G. (1781). \textquotedblleft M\'{e}moire sur la
th\'{e}orie des d\'{e}blais et remblais,\textquotedblright\ In Histoire de
l'Acad\'{e}mie Royale des Sciences de Paris, pp. 666-704.

\bibitem[N2006]{N2006}Nelsen, R.B. (2006). An Introduction to Copulas. Springer.

\bibitem[NB2008]{NB2008}Newhouse, J. P., R. H. Brook, N. Duan, E. B. Keeler,
A. Leibowitz, W. G. Manning, M. S. Marquis, C. N. Morris, C. E. Phelps and J.
E. Rolph (2008). \textquotedblleft Attrition in the RAND Health Insurance
Experiment: a response to Nyman,\textquotedblright\ \emph{Journal of Health
Politics, Policy and Law}, 33, 295-308.

\bibitem[OW2011]{OW2011}Orzechowski and Walker (2011). \textquotedblleft The
Tax Burden On Tobacco,\textquotedblright\ \emph{The Tax Burden onTobacco:
Historical Compilation}, 46.

\bibitem[PK2008]{PK2008}Park, C. and C. Kang (2008). \textquotedblleft Does
Education Induce Healthy Lifestyle?,\textquotedblright\emph{\ Journal of
Health Economics,} 27, 1516-1531.

\bibitem[P2010]{P2010}Park, S. S. (2010). \textquotedblleft Heterogeneous
Effects of Class Size Reduction: Re-Visiting Project STAR,\textquotedblright%
\ Working Paper.

\bibitem[P2013]{P2013}Park, B. G. (2013). \textquotedblleft Nonparametric
Identification and Estimation of the Extended Roy Model,\textquotedblright%
\ Working Paper.

\bibitem[PH1989]{PH1989}Permutt, T. and J.R. Hebel (1989). \textquotedblleft
Simultaneous Equation Estimation in a Clinical Trial of the Effect of Smoking
on Birth Weight," \emph{Biometrics}, 45, 619-622.

\bibitem[PRW1999]{PRW1999}Politis, D. N., J.P. Romano and M. Wolf (1999).
\textquotedblleft Subsampling,\textquotedblright\ Springer-Verlag.

\bibitem[S2011]{S2011}Suri, T. (2011). \textquotedblleft Selection and
Comparative Advantage in Technology Adoption,\textquotedblright%
\ \emph{Econometrica}, 79, 159-209.

\bibitem[S2012]{S2012}Simon, D. (2012). \textquotedblleft Does Early Life
Exposure to Cigarette Smoke Permanently Harm Childhood Health? Evidence from
Cigarette Tax Hikes,\textquotedblright\ Working Paper.

\bibitem[V2003]{V2003}Villani, C. (2003). \textquotedblleft Topics in Optimal
Transportation,\textquotedblright\ Graduate Studies in Mathematics 58,
American Mathematical Society.

\bibitem[WD1990]{WD1990}Williamson, R. C. and T. Downs. (1990).
\textquotedblleft Probabilistic Arithmetic I: Numerical Methods for
Calculating Convolutions and Dependency Bounds,\textquotedblright%
\ \emph{International Journal of Approximate Reasoning}, 4, 89--158.
\end{thebibliography}

\setcounter{figure}{0} \renewcommand{\thefigure}{A.\arabic{figure}}
\setcounter{table}{0} \renewcommand{\thetable}{A.\arabic{table}}
\setcounter{equation}{0} \renewcommand{\theequation}{A.\arabic{equation}}
\appendix

\section*{Appendix A}

In Appendix A. I\ provide technical proofs for Theorem 1,
Corollary 1 and Corollary 2. Throughout Appendix A, the
function $\varphi$ is assumed to be bounded and continuous without loss of
generality by Lemma 2.

\subsection*{Proof of Theorem 1}

Since the proofs of characterization of $F_{\Delta}^{U}$ and $F_{\Delta}^{L}$
are very similar, I\ provide a proof for characterization of $F_{\Delta}^{L}$
only. Let
\begin{align*}
I\left[  \pi\right]   & =\int\left\{  \boldsymbol{1}\left\{  y_{1}%
-y_{0}<\delta\right\}  +\lambda\left(  1-\boldsymbol{1}_{C}\left(  y_{0}%
,y_{1}\right)  \right)  \right\}  d\pi,\\
J\left(  \varphi,\psi\right)   & =\int\varphi d\mu_{0}+\int\psi d\mu_{1},
\end{align*}
for $\lambda=\infty.$ To prove Theorem 1, I\ introduce Lemma A.1:

\subparagraph{Lemma A.1}

For any function $f:\mathbb{R}\rightarrow\mathbb{R},$ $s\in\lbrack0,1],$ and
nonnegative integer $k,$ define $A_{k}^{+}$ and $A_{k}^{-}$ to be level sets
of a function $f$ as follows:
\begin{align*}
A_{k}^{+}\left(  f,s\right)   & =\left\{  y\in\mathbb{R};f(y)>s+k\right\}  ,\\
A_{k}^{-}\left(  f,s\right)   & =\left\{  y\in\mathbb{R};f(y)\leq-\left(
s+k\right)  \right\}  .
\end{align*}
Then for the following dual problems%
\[
\underset{\pi\in\Pi\left(  \mu_{0},\mu_{1}\right)  }{\inf}I\left[  \pi\right]
=\underset{\left(  \varphi,\psi\right)  \in\Phi_{c}}{\sup}J\left(
\varphi,\psi\right)  ,
\]
each $\left(  \varphi,\psi\right)  \in\Phi_{c}$ can be represented as a
continuous convex combination of a continuum of pairs of the form
\[
\left(  \sum\limits_{k=0}^{\infty}\boldsymbol{1}_{A_{k}^{+}\left(
\varphi,s\right)  }-\sum\limits_{k=0}^{\infty}\boldsymbol{1}_{A_{k}^{-}\left(
\varphi,s\right)  },\sum\limits_{k=0}^{\infty}\boldsymbol{1}_{A_{k}^{+}\left(
\psi,s\right)  }-\sum\limits_{k=0}^{\infty}\boldsymbol{1}_{A_{k}^{-}\left(
\psi,s\right)  }\right)  \in\Phi_{c}%
\]

\subparagraph{Proof of Lemma A.1}

By Lemma 2,%
\[
\underset{\pi\in\Pi\left(  \mu_{0},\mu_{1}\right)  }{\inf}I\left[  \pi\right]
=\underset{\left(  \varphi,\psi\right)  \in\Phi_{c}}{\sup}J\left(
\varphi,\psi\right)  ,
\]
where $\Phi_{c}$ is the set of all pairs $\left(  \varphi,\psi\right)  $ in
$L^{1}\left(  dF_{0}\right)  $ $\times L^{1}\left(  dF_{1}\right)  $ such that%
\begin{equation}
\varphi\left(  y_{0}\right)  +\psi\left(  y_{1}\right)  \leq\boldsymbol{1}%
\left\{  y_{1}-y_{0}<\delta\right\}  +\lambda\left(  1-\boldsymbol{1}%
_{C}\left(  y_{0},y_{1}\right)  \right)  \text{ for all }\left(  y_{0}%
,y_{1}\right)  .\label{[A.3]}%
\end{equation}
Note that $\Phi_{c}$\ is a convex set.\ From the definition of $A_{k}%
^{+}\left(  f,s\right)  $ and $A_{k}^{-}\left(  f,s\right)  ,$ for any
function $f:\mathbb{R}\rightarrow\mathbb{R}$ and $s\in(0,1],$%
\begin{equation}
\ldots\subseteq A_{1}^{+}\left(  f,s\right)  \subseteq A_{0}^{+}\left(
f,s\right)  \subseteq\left(  A_{0}^{-}\left(  f,s\right)  \right)
^{c}\subseteq\left(  A_{1}^{-}\left(  f,s\right)  \right)  ^{c}\subseteq
\ldots.,\label{[A.4.5]}%
\end{equation}
as illustrated in Figure A.1.%
\begin{align*}
&
{\includegraphics[
natheight=2.458700in,
natwidth=4.125200in,
height=2.4993in,
width=4.1753in
]%
{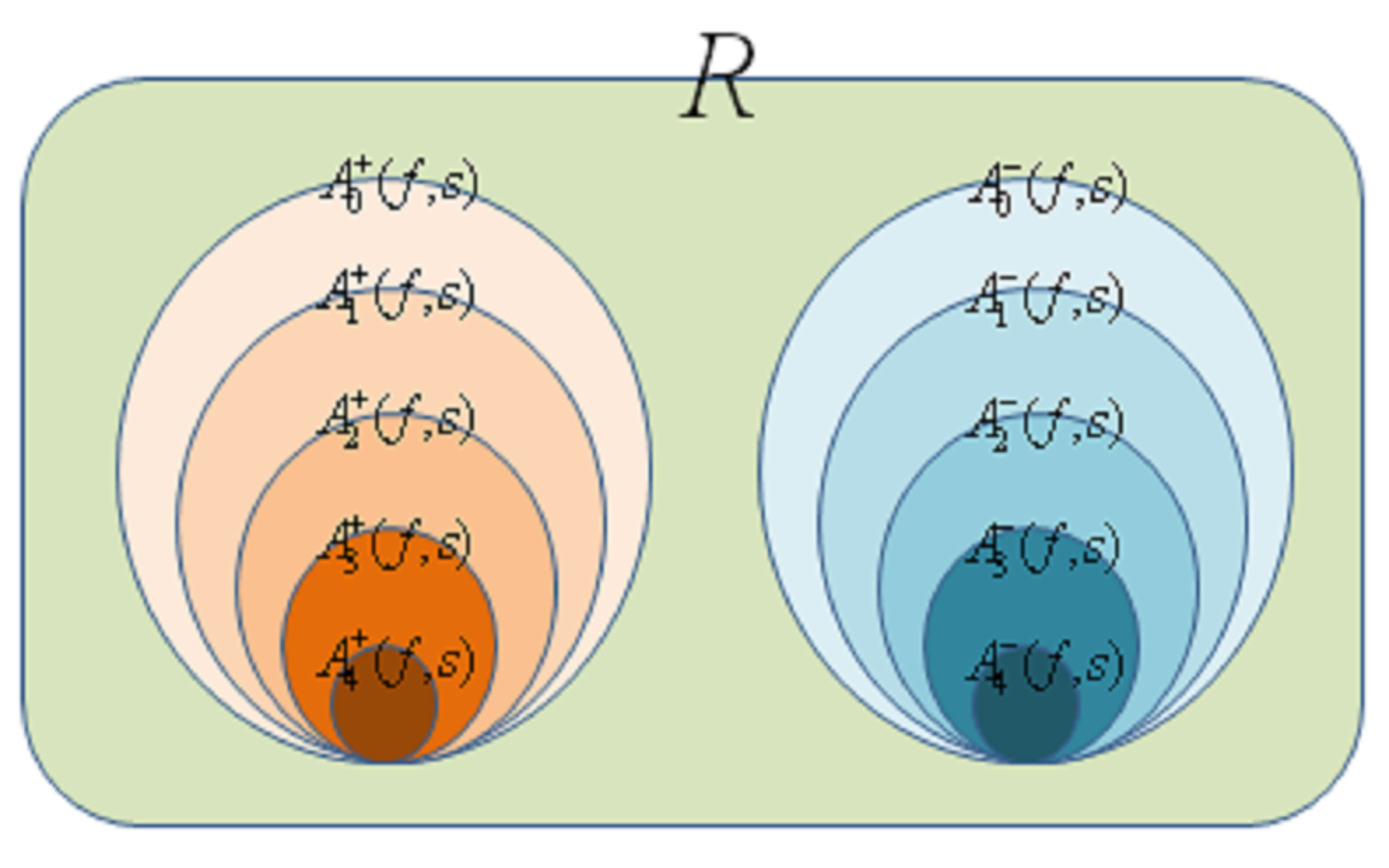}%
}
\\
& \text{Figure A.1: Monotonicity of }\left\{  A_{k}^{+}\left(  f,s\right)
\right\}  _{k=0}^{\infty}\text{ and }\left\{  A_{k}^{-}\left(  f,s\right)
\right\}  _{k=0}^{\infty}%
\end{align*}
\qquad Let%
\begin{align*}
\varphi_{+}\left(  x\right)   & =\max\left\{  \varphi(x),0\right\}  \geq0,\\
\varphi_{-}\left(  x\right)   & =\min\left\{  \varphi(x),0\right\}  \leq0.
\end{align*}
By the layer cake representation theorem, $\varphi_{+}\left(  x\right)  $ can
be written as%
\begin{align}
\varphi_{+}\left(  x\right)   & =\int_{0}^{\varphi_{+}\left(  x\right)
}ds\label{cake}\\
& =\int_{0}^{\infty}\boldsymbol{1}\left\{  \varphi_{+}\left(  x\right)
>s\right\}  ds\nonumber\\
& =\sum_{k=0}^{\infty}\int_{0}^{1}\boldsymbol{1}\left\{  \varphi_{+}\left(
x\right)  >s+k\right\}  ds\nonumber\\
& =\int_{0}^{1}\sum_{k=0}^{\infty}\boldsymbol{1}\left\{  \varphi_{+}\left(
x\right)  >s+k\right\}  ds\nonumber\\
& =\int_{0}^{1}\sum_{k=0}^{\infty}\boldsymbol{1}\left\{  \varphi\left(
x\right)  >s+k\right\}  ds\nonumber\\
& =\int_{0}^{1}\sum_{k=0}^{\infty}\boldsymbol{1}_{A_{k}^{+}\left(
\varphi,s\right)  }\left(  x\right)  ds,\nonumber
\end{align}
where $A_{k}^{+}\left(  f,s\right)  =\left\{  y\in\mathbb{R};f(y)>s+k\right\}
$ for any function $f.$ The fourth equality in (A.3) follows from
Fubini's theorem. Similarly, the nonpositive function $\varphi_{-}\left(
x\right)  $ can be represented as%
\begin{align*}
\varphi_{-}\left(  x\right)   & =-\int_{0}^{\infty}\boldsymbol{1}\left\{
\varphi_{-}\left(  x\right)  \leq-s\right\}  ds\\
& =-\sum_{k=0}^{\infty}\int_{0}^{1}\boldsymbol{1}\left\{  \varphi_{-}\left(
x\right)  \leq-\left(  s+k\right)  \right\}  ds\\
& =-\int_{0}^{1}\sum_{k=0}^{\infty}\boldsymbol{1}\left\{  \varphi_{-}\left(
x\right)  \leq-\left(  s+k\right)  \right\}  ds\\
& =-\int_{0}^{1}\sum_{k=0}^{\infty}\boldsymbol{1}\left\{  \varphi\left(
x\right)  \leq-\left(  s+k\right)  \right\}  ds\\
& =-\int_{0}^{1}\sum_{k=0}^{\infty}\boldsymbol{1}_{A_{k}^{-}\left(
\varphi,s\right)  }\left(  x\right)  ds.
\end{align*}
where $A_{k}^{-}\left(  f,s\right)  =\left\{  y\in\mathbb{R};f(y)\leq-\left(
s+k\right)  \right\}  $ for any function $f.$ Similarly, $\psi_{+}\left(
x\right)  $ and $\psi_{-}\left(  x\right)  $ are written as follows:%
\begin{align*}
\psi_{+}\left(  x\right)   & =\int_{0}^{1}\sum_{k=0}^{\infty}\boldsymbol{1}%
_{A_{k}^{+}\left(  \psi,s\right)  }\left(  x\right)  ds,\\
\psi_{-}\left(  x\right)   & =-\int_{0}^{1}\sum_{k=0}^{\infty}\boldsymbol{1}%
_{A_{k}^{-}\left(  \psi,s\right)  }\left(  x\right)  ds.
\end{align*}

For any $\left(  \varphi,\psi\right)  \in$ $\Phi_{c},$ one can write
\begin{align*}
& \left(  \varphi,\psi\right) \\
& =\left(  \varphi_{+}+\varphi_{-},\psi_{+}+\psi_{-}\right) \\
& =\int_{0}^{1}\left(  \sum_{k=0}^{\infty}\left(  \boldsymbol{1}_{A_{k}%
^{+}\left(  \varphi,s\right)  }-\boldsymbol{1}_{A_{k}^{-}\left(
\varphi,s\right)  }\right)  ds,\sum_{k=0}^{\infty}\left(  \boldsymbol{1}%
_{A_{k}^{+}\left(  \psi,s\right)  }-\boldsymbol{1}_{A_{k}^{-}\left(
\psi,s\right)  }\right)  \right)  ds,
\end{align*}
which is a continuous convex combination of a continuum of pairs of
\[
\left(  \sum_{k=0}^{\infty}\left(  \boldsymbol{1}_{A_{k}^{+}\left(
\varphi,s\right)  }-\boldsymbol{1}_{A_{k}^{-}\left(  \varphi,s\right)
}\right)  ,\sum_{k=0}^{\infty}\left(  \boldsymbol{1}_{A_{k}^{+}\left(
\psi,s\right)  }-\boldsymbol{1}_{A_{k}^{-}\left(  \psi,s\right)  }\right)
\right)  _{s\in\lbrack0,1]}.
\]
\qquad To see if $\left(  \sum_{k=0}^{\infty}\left(  \boldsymbol{1}_{A_{k}%
^{+}d\left(  \varphi,s\right)  }-\boldsymbol{1}_{A_{k}^{-}\left(
\varphi,s\right)  }\right)  ,\sum_{k=0}^{\infty}\left(  \boldsymbol{1}%
_{A_{k}^{+}\left(  \psi,s\right)  }-\boldsymbol{1}_{A_{k}^{-}\left(
\psi,s\right)  }\right)  \right)  \in\Phi_{c},$ check the following: for any
$s\in\lbrack0,1]$ and $\lambda=\infty,$%
\begin{align}
& \sum_{k=0}^{\infty}\left(  \boldsymbol{1}_{A_{k}^{+}\left(  \varphi
,s\right)  }\left(  y_{0}\right)  -\boldsymbol{1}_{A_{k}^{-}\left(
\varphi,s\right)  }\left(  y_{0}\right)  \right)  +\sum_{k=0}^{\infty}\left(
\boldsymbol{1}_{A_{k}^{+}\left(  \psi,s\right)  }\left(  y_{1}\right)
-\boldsymbol{1}_{A_{k}^{-}\left(  \psi,s\right)  }\left(  y_{1}\right)
\right) \label{[A.5]}\\
& \leq\boldsymbol{1}\left\{  y_{1}-y_{0}<\delta\right\}  +\lambda\left(
1-\boldsymbol{1}_{C}\left(  y_{0},y_{1}\right)  \right)  \text{.}\nonumber
\end{align}
The nontrivial case to check is when the LHS in (A.4) \ is positive.
Consider the case where $s+t<\varphi\left(  y_{0}\right)  \leq s+t+1$ and
$-\left(  s+t\right)  <\psi\left(  y_{1}\right)  \leq-\left(  s+t-1\right)  $
for some nonnegative integer $t$ and $s\in\lbrack0,1]$. Then,%
\begin{align*}
\sum_{k=0}^{\infty}\left(  \boldsymbol{1}_{A_{k}^{+}\left(  \varphi,s\right)
}\left(  y_{0}\right)  -\boldsymbol{1}_{A_{k}^{-}\left(  \varphi,s\right)
}\left(  y_{0}\right)  \right)   & =t+1,\\
\sum_{k=0}^{\infty}\left(  \boldsymbol{1}_{A_{k}^{+}\left(  \psi,s\right)
}\left(  y_{1}\right)  -\boldsymbol{1}_{A_{k}^{-}\left(  \psi,s\right)
}\left(  y_{1}\right)  \right)   & =-t,
\end{align*}
and so the LHS in (A.4) is $1$. Also, it follows from (A.1)
that for $\left(  y_{0},y_{1}\right)  \in\mathbb{R}\times$ $\mathbb{R}$ s.t.
$s+t\leq\varphi\left(  y_{0}\right)  <s+t+1$ and $-\left(  s+t\right)
<\psi\left(  y_{1}\right)  ,$
\[
0<\varphi\left(  y_{0}\right)  +\psi\left(  y_{1}\right)  \leq\boldsymbol{1}%
\left\{  y_{1}-y_{0}<\delta\right\}  +\lambda\left(  1-\boldsymbol{1}%
_{C}\left(  y_{0},y_{1}\right)  \right)  ,
\]
and thus (A.4) is satisfied in this case from the following:%
\[
\boldsymbol{1}\left\{  y_{1}-y_{0}<\delta\right\}  +\lambda\left(
1-\boldsymbol{1}_{C}\left(  y_{0},y_{1}\right)  \right)  \geq1.
\]
\qquad Consider another case where $s+t\leq\varphi\left(  y_{0}\right)
<s+t+1$ and $-\left(  s+t-1\right)  <\psi\left(  y_{1}\right)  \leq-\left(
s+t-2\right)  $ for some nonnegative integer $t$ and $s\in\lbrack0,1]$. Then
the LHS in (A.4) is $2.$ Moreover, since $\varphi\left(  y_{0}\right)
+\psi\left(  y_{1}\right)  >1$, for $\left(  y_{0},y_{1}\right)  \in
\mathbb{R}\times$ $\mathbb{R}$ s.t. $s+t\leq\varphi\left(  y_{0}\right)
<s+t+1$ and $-\left(  s+t-1\right)  <\psi\left(  y_{1}\right)  ,$ by
(A.1)
\[
1<\varphi\left(  y_{0}\right)  +\psi\left(  y_{1}\right)  \leq\boldsymbol{1}%
\left\{  y_{1}-y_{0}<\delta\right\}  +\lambda\left(  1-\boldsymbol{1}%
_{C}\left(  y_{0},y_{1}\right)  \right)  ,
\]
and thus (A.4) is also satisfied from the following:%
\[
\boldsymbol{1}\left\{  y_{1}-y_{0}<\delta\right\}  +\lambda\left(
1-\boldsymbol{1}_{C}\left(  y_{0},y_{1}\right)  \right)  =\infty.
\]
\qquad Similarly, it can be proven that (A.4) is also satisfied for
other nontrivial cases. Therefore it concludes that each $\left(  \varphi
,\psi\right)  \in\Phi_{c}$ can be written as a continuous convex combination
of a continuum of pairs of the form
\[
\left(  \sum\limits_{k=0}^{\infty}\left(  \boldsymbol{1}_{A_{k}^{+}\left(
\varphi,s\right)  }-\boldsymbol{1}_{A_{k}^{-}\left(  \varphi,s\right)
}\right)  ,\sum\limits_{k=0}^{\infty}\left(  \boldsymbol{1}_{A_{k}^{+}\left(
\psi,s\right)  }-\boldsymbol{1}_{A_{k}^{-}\left(  \psi,s\right)  }\right)
\right)  .
\]
${\small \blacksquare}$

\subparagraph{Proof of Theorem 1}

\bigskip\bigskip By Lemma A.1, $\left(  \varphi,\psi\right)  \in$ $\Phi_{c}$
can be represented as a continuous convex combination of a continuum of pairs
of the form
\[
\left(  \sum\limits_{k=0}^{\infty}\left(  \boldsymbol{1}_{A_{k}^{+}\left(
\varphi,s\right)  }-\boldsymbol{1}_{A_{k}^{-}\left(  \varphi,s\right)
}\right)  ,\sum\limits_{k=0}^{\infty}\left(  \boldsymbol{1}_{A_{k}^{+}\left(
\psi,s\right)  }-\boldsymbol{1}_{A_{k}^{-}\left(  \psi,s\right)  }\right)
\right)  ,
\]
with $\ $%
\begin{align*}
& \text{ }\sum\limits_{k=0}^{\infty}\left(  \boldsymbol{1}_{A_{k}^{+}\left(
\varphi,s\right)  }\left(  y_{0}\right)  -\boldsymbol{1}_{A_{k}^{-}\left(
\varphi,s\right)  }\left(  y_{0}\right)  \right)  +\sum\limits_{k=0}^{\infty
}\left(  \boldsymbol{1}_{A_{k}^{+}\left(  \psi,s\right)  }\left(
y_{1}\right)  -\boldsymbol{1}_{A_{k}^{-}\left(  \psi,s\right)  }\left(
y_{1}\right)  \right) \\
& \leq\boldsymbol{1}\left\{  y_{1}-y_{0}<\delta\right\}  +\lambda\left(
1-\boldsymbol{1}_{C}\left(  y_{0},y_{1}\right)  \right)  .
\end{align*}
Since $\Phi_{c}$ is a convex set and $J\left(  \varphi,\psi\right)
=\int\varphi dF_{0}+\int\psi dF_{1}$ is a linear functional, for all $\left(
\varphi,\psi\right)  \in$ $\Phi_{c}$, there exists $s\in(0,1]$ such that%
\begin{equation}
J\left(  \sum\limits_{k=0}^{\infty}\left(  \boldsymbol{1}_{A_{k}^{+}\left(
\varphi,s\right)  }-\boldsymbol{1}_{A_{k}^{-}\left(  \varphi,s\right)
}\right)  ,\sum\limits_{k=0}^{\infty}\left(  \boldsymbol{1}_{A_{k}^{+}\left(
\psi,s\right)  }-\boldsymbol{1}_{A_{k}^{-}\left(  \psi,s\right)  }\right)
\right)  \geq J\left(  \varphi,\psi\right)  .\label{[A.2]}%
\end{equation}
Thus, the value of $\underset{\left(  \varphi,\psi\right)  \in\Phi_{c}}{\sup
}J\left(  \varphi,\psi\right)  $ is unchanged even if one restricts the
supremum to pairs of the form $\left(  \sum\limits_{k=0}^{\infty}\left(
\boldsymbol{1}_{A_{k}^{+}\left(  \varphi,s\right)  }-\boldsymbol{1}_{A_{k}%
^{-}\left(  \varphi,s\right)  }\right)  ,\sum\limits_{k=0}^{\infty}\left(
\boldsymbol{1}_{A_{k}^{+}\left(  \psi,s\right)  }-\boldsymbol{1}_{A_{k}%
^{-}\left(  \psi,s\right)  }\right)  \right)  .$ Hence for all $\left(
y_{0},y_{1}\right)  \in\mathbb{R}^{2}$,%
\begin{align*}
& \sum\limits_{k=0}^{\infty}\left(  \boldsymbol{1}_{A_{k}^{+}\left(
\varphi,s\right)  }\left(  y_{0}\right)  -\boldsymbol{1}_{A_{k}^{-}\left(
\varphi,s\right)  }\left(  y_{0}\right)  \right)  +\sum\limits_{k=0}^{\infty
}\left(  \boldsymbol{1}_{A_{k}^{+}\left(  \psi,s\right)  }\left(
y_{1}\right)  -\boldsymbol{1}_{A_{k}^{-}\left(  \psi,s\right)  }\left(
y_{1}\right)  \right) \\
& \leq\boldsymbol{1}\left\{  y_{1}-y_{0}<\delta\right\}  +\lambda\left(
1-\boldsymbol{1}_{C}\left(  y_{0},y_{1}\right)  \right)  ,
\end{align*}
which implies that for each $y_{1}\in\mathbb{R},$%
\begin{align*}
-\infty & <\underset{y_{0}\in\mathbb{R}}{\sup}\left[  \sum\limits_{k=0}%
^{\infty}\left(  \boldsymbol{1}_{A_{k}^{+}\left(  \varphi,s\right)  }\left(
y_{0}\right)  -\boldsymbol{1}_{A_{k}^{-}\left(  \varphi,s\right)  }\left(
y_{0}\right)  \right)  -\boldsymbol{1}\left(  y_{1}-y_{0}<\delta\right)
-\lambda\left(  1-\boldsymbol{1}_{C}\left(  y_{0},y_{1}\right)  \right)
\right] \\
& \leq-\sum\limits_{k=0}^{\infty}\left(  \boldsymbol{1}_{A_{k}^{+}\left(
\psi,s\right)  }\left(  y_{1}\right)  -\boldsymbol{1}_{A_{k}^{-}\left(
\psi,s\right)  }\left(  y_{1}\right)  \right)  .
\end{align*}
\qquad Define $\left\{  A_{k,D}^{+}\left(  \varphi,s\right)  \right\}
_{k=0}^{\infty},$ $\left\{  A_{k,D}^{-}\left(  \varphi,s\right)  \right\}
_{k=0}^{\infty}$ as follows:%
\begin{align}
A_{k,D}^{+}\left(  \varphi,s\right)   & =%
\begin{tabular}
[c]{l}%
$\left\{  y_{1}\in\mathbb{R}|\exists y_{0}\in A_{k}^{+}\left(  \varphi
,s\right)  \text{ s.t. }y_{1}-y_{0}\geq\delta\text{ and }\left(  y_{0}%
,y_{1}\right)  \in C\right\}  $\\
$\cup\left\{  y_{1}\in\mathbb{R}|\exists y_{0}\in A_{k+1}^{+}\left(
\varphi,s\right)  \text{ s.t. }y_{1}-y_{0}<\delta\text{ and }\left(
y_{0},y_{1}\right)  \in C\right\}  $\\
$\text{for any integer }k\geq0,$%
\end{tabular}
\label{[A.6]}\\
A_{0,D}^{-}\left(  \varphi,s\right)   & =%
\begin{tabular}
[c]{l}%
$\left\{  y_{1}\in\mathbb{R}|\forall y_{0}\leq y_{1}-\delta\text{ s.t.
}\left(  y_{0},y_{1}\right)  \in C,\text{ }y_{0}\in A_{0}^{-}\left(
\varphi,s\right)  \right\}  $\\
$\cap\left\{  y_{1}\in\mathbb{R}|\forall y_{0}>y_{1}-\delta\text{ s.t.
}\left(  y_{0},y_{1}\right)  \in C,\text{ }y_{0}\in\left(  A_{0}^{+}\left(
\varphi,s\right)  \right)  ^{c}\right\}  ,$%
\end{tabular}
\nonumber\\
A_{k,D}^{-}\left(  \varphi,s\right)   & =%
\begin{tabular}
[c]{l}%
$\left\{  y_{1}\in\mathbb{R}|\forall y_{0}\leq y_{1}-\delta\text{ s.t.
}\left(  y_{0},y_{1}\right)  \in C,\text{ }y_{0}\in A_{k}^{-}\left(
\varphi,s\right)  \right\}  $\\
$\cap\left\{  y_{1}\in\mathbb{R}|\forall y_{0}>y_{1}-\delta\text{ s.t.
}\left(  y_{0},y_{1}\right)  \in C,\text{ }y_{0}\in A_{k-1}^{-}\left(
\varphi,s\right)  \right\}  \text{ }$\\
$\text{for any integer }k>0.$%
\end{tabular}
\nonumber
\end{align}
Also, according to the definitions above and Figure A.1, if $y_{1}\in
A_{\rho,D}^{+}\left(  \varphi,s\right)  $ for some $\rho\geq0,$ then
\begin{align*}
& \underset{y_{0}\in\mathbb{R}}{\sup}\left[  \sum\limits_{k=0}^{\infty}\left(
\boldsymbol{1}_{A_{k}^{+}\left(  \varphi,s\right)  }\left(  y_{0}\right)
-\boldsymbol{1}_{A_{k}^{-}\left(  \varphi,s\right)  }\left(  y_{0}\right)
\right)  -\boldsymbol{1}\left\{  y_{1}-y_{0}<\delta\right\}  -\lambda\left(
1-\boldsymbol{1}_{C}\left(  y_{0},y_{1}\right)  \right)  \right] \\
& \geq\rho+1,
\end{align*}
and if $y_{1}\in A_{\rho,D}^{-}\left(  \varphi,s\right)  $ \ for some
$\rho\geq0,$%
\begin{align*}
& \underset{y_{0}\in\mathbb{R}}{\sup}\left[  \sum\limits_{k=0}^{\infty}\left(
\boldsymbol{1}_{A_{k}^{+}\left(  \varphi,s\right)  }\left(  y_{0}\right)
-\boldsymbol{1}_{A_{k}^{-}\left(  \varphi,s\right)  }\left(  y_{0}\right)
\right)  -\boldsymbol{1}\left\{  y_{1}-y_{0}<\delta\right\}  -\lambda\left(
1-\boldsymbol{1}_{C}\left(  y_{0},y_{1}\right)  \right)  \right] \\
& \leq-\left(  \rho+1\right)  .
\end{align*}
Hence, if $y_{1}\in A_{\rho,D}^{+}\left(  \varphi,s\right)  -A_{\rho+1,D}%
^{+}\left(  \varphi,s\right)  ,$ then
\begin{align*}
& \underset{y_{0}\in\mathbb{R}}{\sup}\left[  \sum\limits_{k=0}^{\infty}\left(
\boldsymbol{1}_{A_{k}^{+}\left(  \varphi,s\right)  }\left(  y_{0}\right)
-\boldsymbol{1}_{A_{k}^{-}\left(  \varphi,s\right)  }\left(  y_{0}\right)
\right)  -\boldsymbol{1}\left\{  y_{1}-y_{0}<\delta\right\}  -\lambda\left(
1-\boldsymbol{1}_{C}\left(  y_{0},y_{1}\right)  \right)  \right] \\
& =\rho+1,
\end{align*}
and if $y_{1}\in A_{\rho,D}^{-}\left(  \varphi,s\right)  -A_{\rho+1,D}%
^{-}\left(  \varphi,s\right)  ,$ then%
\begin{align*}
& \underset{y_{0}\in\mathbb{R}}{\sup}\left[  \sum\limits_{k=0}^{\infty}\left(
\boldsymbol{1}_{A_{k}^{+}\left(  \varphi,s\right)  }\left(  y_{0}\right)
-\boldsymbol{1}_{A_{k}^{-}\left(  \varphi,s\right)  }\left(  y_{0}\right)
\right)  -\boldsymbol{1}\left\{  y_{1}-y_{0}<\delta\right\}  -\lambda\left(
1-\boldsymbol{1}_{C}\left(  y_{0},y_{1}\right)  \right)  \right] \\
& =-\left(  \rho+1\right)  .
\end{align*}
Hence,%
\begin{align*}
& \sum\limits_{k=0}^{\infty}\left(  \boldsymbol{1}_{A_{k,D}^{+}\left(
\varphi,s\right)  }\left(  y_{1}\right)  -\boldsymbol{1}_{A_{k,D}^{-}\left(
\varphi,s\right)  }\left(  y_{1}\right)  \right) \\
& =\underset{y_{0}\in\mathbb{R}}{\sup}\left[  \sum\limits_{k=0}^{\infty
}\left(  \boldsymbol{1}_{A_{k}^{+}\left(  \varphi,s\right)  }\left(
y_{0}\right)  -\boldsymbol{1}_{A_{k}^{-}\left(  \varphi,s\right)  }\left(
y_{0}\right)  \right)  -\boldsymbol{1}\left\{  y_{1}-y_{0}<\delta\right\}
-\lambda\left(  1-\boldsymbol{1}_{C}\left(  y_{0},y_{1}\right)  \right)
\right] \\
& \leq-\sum\limits_{k=0}^{\infty}\left(  \boldsymbol{1}_{A_{k}^{+}\left(
\psi,s\right)  }\left(  y_{1}\right)  -\boldsymbol{1}_{A_{k}^{-}\left(
\psi,s\right)  }\left(  y_{1}\right)  \right)  .
\end{align*}
\qquad Now define
\begin{align*}
A_{k}\left(  \varphi,s\right)   & =\left\{
\begin{tabular}
[c]{ll}%
$A_{k}^{+}\left(  \varphi,s\right)  ,$ & if $k\geq0,$\\
$\left(  A_{-k-1}^{-}\left(  \varphi,s\right)  \right)  ^{c},$ & if $k<0,$%
\end{tabular}
\right. \\
A_{k}^{D}\left(  \varphi,s\right)   & =\left\{
\begin{tabular}
[c]{ll}%
$A_{k,D}^{+}\left(  \varphi,s\right)  ,$ & if $k\geq0,$\\
$\left(  A_{-k-1,D}^{-}\left(  \varphi,s\right)  \right)  ^{c},$ & if $k<0.$%
\end{tabular}
\right.
\end{align*}
Then for all $\left(  y_{0},y_{1}\right)  \in\mathbb{R}^{2},$%
\begin{align}
& \boldsymbol{1}\left\{  y_{1}-y_{0}<\delta\right\}  +\lambda\left(
1-\boldsymbol{1}_{C}\left(  y_{0},y_{1}\right)  \right) \label{[A.10]}\\
& \geq\sum\limits_{k=0}^{\infty}\left(  \boldsymbol{1}_{A_{k}^{+}\left(
\varphi,s\right)  }\left(  y_{0}\right)  -\boldsymbol{1}_{A_{k}^{-}\left(
\varphi,s\right)  }\left(  y_{0}\right)  \right)  -\sum\limits_{k=0}^{\infty
}\left(  \boldsymbol{1}_{A_{k,D}^{+}\left(  \varphi,s\right)  }\left(
y_{1}\right)  -\boldsymbol{1}_{A_{k,D}^{-}\left(  \varphi,s\right)  }\left(
y_{1}\right)  \right) \nonumber\\
& =\sum\limits_{k=0}^{\infty}\left\{  \left(  \boldsymbol{1}_{A_{k}^{+}\left(
\varphi,s\right)  }\left(  y_{0}\right)  -\boldsymbol{1}_{A_{k}^{-}\left(
\varphi,s\right)  }\left(  y_{0}\right)  \right)  -\left(  \boldsymbol{1}%
_{A_{k,D}^{+}\left(  \varphi,s\right)  }\left(  y_{1}\right)  -\boldsymbol{1}%
_{A_{k,D}^{-}\left(  \varphi,s\right)  }\left(  y_{1}\right)  \right)
\right\} \nonumber\\
& =\sum\limits_{k=0}^{\infty}\left\{  \boldsymbol{1}_{A_{k}^{+}\left(
\varphi,s\right)  }\left(  y_{0}\right)  +\left(  1-\boldsymbol{1}_{A_{k}%
^{-}\left(  \varphi,s\right)  }\left(  y_{0}\right)  \right)  -\boldsymbol{1}%
_{A_{k,D}^{+}\left(  \varphi,s\right)  }\left(  y_{1}\right)  -\left(
1-\boldsymbol{1}_{A_{k,D}^{-}\left(  \varphi,s\right)  }\left(  y_{1}\right)
\right)  \right\} \nonumber\\
& =\sum\limits_{k=0}^{\infty}\left\{  \left(  \boldsymbol{1}_{A_{k}^{+}\left(
\varphi,s\right)  }\left(  y_{0}\right)  +\boldsymbol{1}_{\left(  A_{k}%
^{-}\left(  \varphi,s\right)  \right)  ^{c}}\left(  y_{0}\right)  \right)
-\left(  \boldsymbol{1}_{A_{k,D}^{+}\left(  \varphi,s\right)  }\left(
y_{1}\right)  +\boldsymbol{1}_{\left(  A_{k,D}^{-}\left(  \varphi,s\right)
\right)  ^{c}}\left(  y_{1}\right)  \right)  \right\} \nonumber\\
& =\sum\limits_{k=0}^{\infty}\left(  \boldsymbol{1}_{A_{k}^{+}\left(
\varphi,s\right)  }\left(  y_{0}\right)  -\boldsymbol{1}_{A_{k,D}^{+}\left(
\varphi,s\right)  }\left(  y_{1}\right)  \right)  +\sum\limits_{k=0}^{\infty
}\left(  \boldsymbol{1}_{\left(  A_{k}^{-}\left(  \varphi,s\right)  \right)
^{c}}\left(  y_{0}\right)  -\boldsymbol{1}_{\left(  A_{k,D}^{-}\left(
\varphi,s\right)  \right)  ^{c}}\left(  y_{1}\right)  \right) \nonumber\\
& =\sum\limits_{k=0}^{\infty}\left(  \boldsymbol{1}_{A_{k}\left(
\varphi,s\right)  }\left(  y_{0}\right)  -\boldsymbol{1}_{A_{k}^{D}\left(
\varphi,s\right)  }\left(  y_{1}\right)  \right)  +\sum\limits_{k=-\infty
}^{-1}\left(  \boldsymbol{1}_{A_{k}\left(  \varphi,s\right)  }\left(
y_{0}\right)  -\boldsymbol{1}_{A_{k}^{D}\left(  \varphi,s\right)  }\left(
y_{1}\right)  \right) \nonumber\\
& =\sum\limits_{k=-\infty}^{\infty}\left(  \boldsymbol{1}_{A_{k}\left(
\varphi,s\right)  }\left(  y_{0}\right)  -\boldsymbol{1}_{A_{k}^{D}\left(
\varphi,s\right)  }\left(  y_{1}\right)  \right) \nonumber
\end{align}
Equalities in the third and sixth lines of (A.7) are satisfied
because $\varphi$ and $\psi$ are assumed to be bounded. To compress notation,
refer to $A_{k}\left(  \varphi,s\right)  $ and $A_{k}^{D}\left(
\varphi,s\right)  $ merely as $A_{k}$ and $A_{k}^{D}.$ Then,%
\begin{align*}
& \boldsymbol{1}\left\{  y_{1}-y_{0}<\delta\right\}  +\lambda\left(
1-\boldsymbol{1}_{C}\left(  y_{0},y_{1}\right)  \right) \\
& \geq\sum\limits_{k=-\infty}^{\infty}\left(  \boldsymbol{1}_{A_{k}}\left(
y_{0}\right)  -\boldsymbol{1}_{A_{k}^{D}}\left(  y_{1}\right)  \right)  .
\end{align*}
\qquad By taking integrals with respect to $dF$ to both side, one obtains the
following:%
\begin{align}
& \int\left\{  \boldsymbol{1}\left\{  y_{1}-y_{0}<\delta\right\}
-\lambda\left(  1-\boldsymbol{1}_{C}\left(  y_{0},y_{1}\right)  \right)
\right\}  d\pi\label{[A.11]}\\
& \geq\int\sum\limits_{k=-\infty}^{\infty}\left(  \boldsymbol{1}_{A_{k}%
}\left(  y_{0}\right)  -\boldsymbol{1}_{A_{k}^{D}}\left(  y_{1}\right)
\right)  d\pi\nonumber\\
& =\sum\limits_{k=-\infty}^{\infty}\int\left(  \boldsymbol{1}_{A_{k}}\left(
y_{0}\right)  -\boldsymbol{1}_{A_{k}^{D}}\left(  y_{1}\right)  \right)
d\pi\nonumber\\
& =\sum\limits_{k=-\infty}^{\infty}\left\{  \mu_{0}\left(  A_{k}\right)
-\mu_{1}\left(  A_{k}^{D}\right)  \right\}  .\nonumber
\end{align}
The third equality holds by Fubini's theorem because $\sum\limits_{k=-\infty
}^{\infty}\left\vert \boldsymbol{1}_{A_{k}}\left(  y_{0}\right)
-\boldsymbol{1}_{A_{k}^{D}}\left(  y_{1}\right)  \right\vert \leq
\sum\limits_{k=-\infty}^{\infty}\boldsymbol{1}_{A_{k}}\left(  y_{0}\right)
+\sum\limits_{k=-\infty}^{\infty}\boldsymbol{1}_{A_{k}^{D}}\left(
y_{1}\right)  <\infty$ for bounded functions $\varphi$ and $\psi.$ Now,
maximization of $\int\varphi\left(  y_{0}\right)  dF_{0}+\int\psi\left(
y_{1}\right)  dF_{1}$ over $\left(  \varphi,\psi\right)  \in\Phi_{c}$ is
equivalent to the that of $\sum\limits_{k=-\infty}^{\infty}\left\{
F_{0}\left(  A_{k}\right)  -F_{1}\left(  A_{k}^{D}\right)  \right\}  $ over
$\left\{  A_{k}\right\}  _{k=-\infty}^{\infty}$ with the following
monotonicity condition:%
\[
\ldots\subseteq A_{k+1}\subseteq A_{k}\subseteq A_{k-1}\subseteq\ldots.
\]
Therefore, it follows that
\begin{equation}
\inf_{F\in\Pi\left(  \mu_{0},\mu_{1}\right)  }I\left[  F\right]
=\underset{\left\{  A_{k}\right\}  _{k=-\infty}^{\infty}}{\sup}\sum
\limits_{k=-\infty}^{\infty}\left(  \mu_{0}\left(  A_{k}\right)  -\mu
_{1}\left(  A_{k}^{D}\right)  \right)  ,\label{[A.12]}%
\end{equation}
where
\begin{align*}
& \left\{  A_{k}\right\}  _{k=-\infty}^{\infty}\text{ is a monotonically
decreasing sequence of open sets,}\\
A_{k}^{D}  & =\left\{  y_{1}\in\mathbb{R}|\exists y_{0}\in A_{k}\text{ s.t.
}y_{1}-y_{0}\geq\delta\text{ and }\left(  y_{0},y_{1}\right)  \in C\right\} \\
& \cup\left\{  y_{1}\in\mathbb{R}|\exists y_{0}\in A_{k+1}\text{ s.t. }%
y_{1}-y_{0}<\delta\text{ and }\left(  y_{0},y_{1}\right)  \in C\right\}
\text{ for any integer }k.
\end{align*}
\qquad Note that the expression (A.9) can be equivalently written as
follows:%
\[
\inf_{\pi\in\Pi\left(  \mu_{0},\mu_{1}\right)  }I\left[  F\right]
=\underset{\left\{  A_{k}\right\}  _{k=-\infty}^{\infty}}{\sup}\sum
\limits_{k=-\infty}^{\infty}\max\left\{  \mu_{0}\left(  A_{k}\right)  -\mu
_{1}\left(  A_{k}^{D}\right)  ,0\right\}  .
\]
That is, \ $F_{0}\left(  A_{k}\right)  -F_{1}\left(  A_{k}^{D}\right)  \geq0$
for each integer $k$ at the optimum in the expression (A.9)$.$ This
is easily shown by proof by contradiction.

Suppose that there exists an integer $p$ s.t. $F_{0}\left(  A_{p}\right)
-F_{1}\left(  A_{p}^{D}\right)  <0$ at the optimum. If there exists an integer
$q>p$ s.t. $F_{0}\left(  A_{q}\right)  -F_{1}\left(  A_{q}^{D}\right)  >0,$
then there exists another monotonically decreasing sequence of open sets
$\left\{  \widetilde{A}_{k}\right\}  _{k=-\infty}^{\infty}$ s.t.
\[
\sum\limits_{k=-\infty}^{\infty}\left\{  \mu_{0}\left(  \widetilde{A}%
_{k}\right)  -\mu_{1}\left(  \widetilde{A}_{k}^{D}\right)  \right\}
>\sum\limits_{k=-\infty}^{\infty}\left\{  \mu_{0}\left(  A_{k}\right)
-\mu_{1}\left(  A_{k}^{D}\right)  \right\}  ,
\]
where $\widetilde{A}_{k}=A_{k}$ for $k<p$ and $\widetilde{A}_{k}=A_{k+1}$ for
$k\geq p.$ If there is no integer $q>p$ s.t. $F_{0}\left(  A_{q}\right)
-F_{1}\left(  A_{q}^{D}\right)  >0,$ then also there exists a monotonically
decreasing sequence of open sets $\left\{  \widehat{A}_{k}\right\}
_{k=-\infty}^{\infty}$ s.t.
\[
\sum\limits_{k=-\infty}^{\infty}\left\{  \mu_{0}\left(  \widehat{A}%
_{k}\right)  -\mu_{1}\left(  \widehat{A}_{k}^{D}\right)  \right\}  >\left\{
\mu_{0}\left(  A_{k}\right)  -\mu_{1}\left(  A_{k}^{D}\right)  \right\}  ,
\]
where $\widehat{A}_{k}=A_{k}$ for $k<p$ and $\widehat{A}_{k}=\phi$ for $k\geq
p.$ This contradicts the optimality of $\left\{  A_{k}\right\}  _{k=-\infty
}^{\infty}.$ ${\small \blacksquare}$

\subsection*{Proof of Corollary 1}

The proof consists of two parts: (i) deriving the lower bound and (ii)
deriving the upper bound.\newline\textbf{Part 1. The sharp lower bound}

First, I\ prove that in the dual representation%
\begin{align*}
& \inf_{F\in\Pi\left(  F_{0},F_{1}\right)  }\int\left\{  \boldsymbol{1}%
\left\{  y_{1}-y_{0}<\delta\right\}  +\lambda\left(  \boldsymbol{1}\left(
y_{1}<y_{0}\right)  \right)  \right\}  dF\\
& =\underset{\left(  \varphi,\psi\right)  \in\Phi_{c}}{\sup}\int\varphi\left(
y_{0}\right)  d\mu_{0}+\int\psi\left(  y_{1}\right)  d\mu_{1},
\end{align*}
the function $\varphi$\ is nondecreasing.

Recall that
\[
\varphi\left(  y_{0}\right)  =\underset{y_{1}\geq y_{0}}{\inf}\left\{
\boldsymbol{1}\left\{  y_{1}-y_{0}<\delta\right\}  -\psi\left(  y_{1}\right)
\right\}  .
\]
Pick $\left(  y_{0}^{\prime},y_{1}^{\prime}\right)  $ and $\left(
y_{0}^{\prime\prime},y_{1}^{\prime\prime}\right)  $ with $y_{0}^{\prime\prime
}>y_{0}^{\prime}$ in the support of the optimal joint distribution. Then,%
\begin{align}
\varphi\left(  y_{0}^{\prime}\right)   & =\underset{y_{1}\geq y_{0}}{\inf
}\left\{  \boldsymbol{1}\left\{  y_{1}-y_{0}^{\prime}<\delta\right\}
-\psi\left(  y_{1}\right)  \right\} \label{A.MTR.1}\\
& \leq\boldsymbol{1}\left\{  y_{1}^{\prime\prime}-y_{0}^{\prime}%
<\delta\right\}  -\psi\left(  y_{1}^{\prime\prime}\right) \nonumber\\
& \leq\boldsymbol{1}\left\{  y_{1}^{\prime\prime}-y_{0}^{\prime\prime}%
<\delta\right\}  -\psi\left(  y_{1}^{\prime\prime}\right) \nonumber\\
& =\varphi\left(  y_{0}^{\prime\prime}\right)  .\nonumber
\end{align}
The inequality in the second line of (A.10) is satisfied because
$y_{1}^{\prime\prime}\geq y_{0}^{\prime\prime}>y_{0}^{\prime}.$ The inequality
in the third line of (A.10) holds because $\boldsymbol{1}\left\{
y_{1}-y_{0}<\delta\right\}  $ is nondecreasing in $y_{0}$.

\begin{center}%
\begin{tabular}
[c]{l}%
{\parbox[b]{3.6685in}{\begin{center}
\includegraphics[
natheight=3.416900in,
natwidth=4.854200in,
height=2.591in,
width=3.6685in
]%
{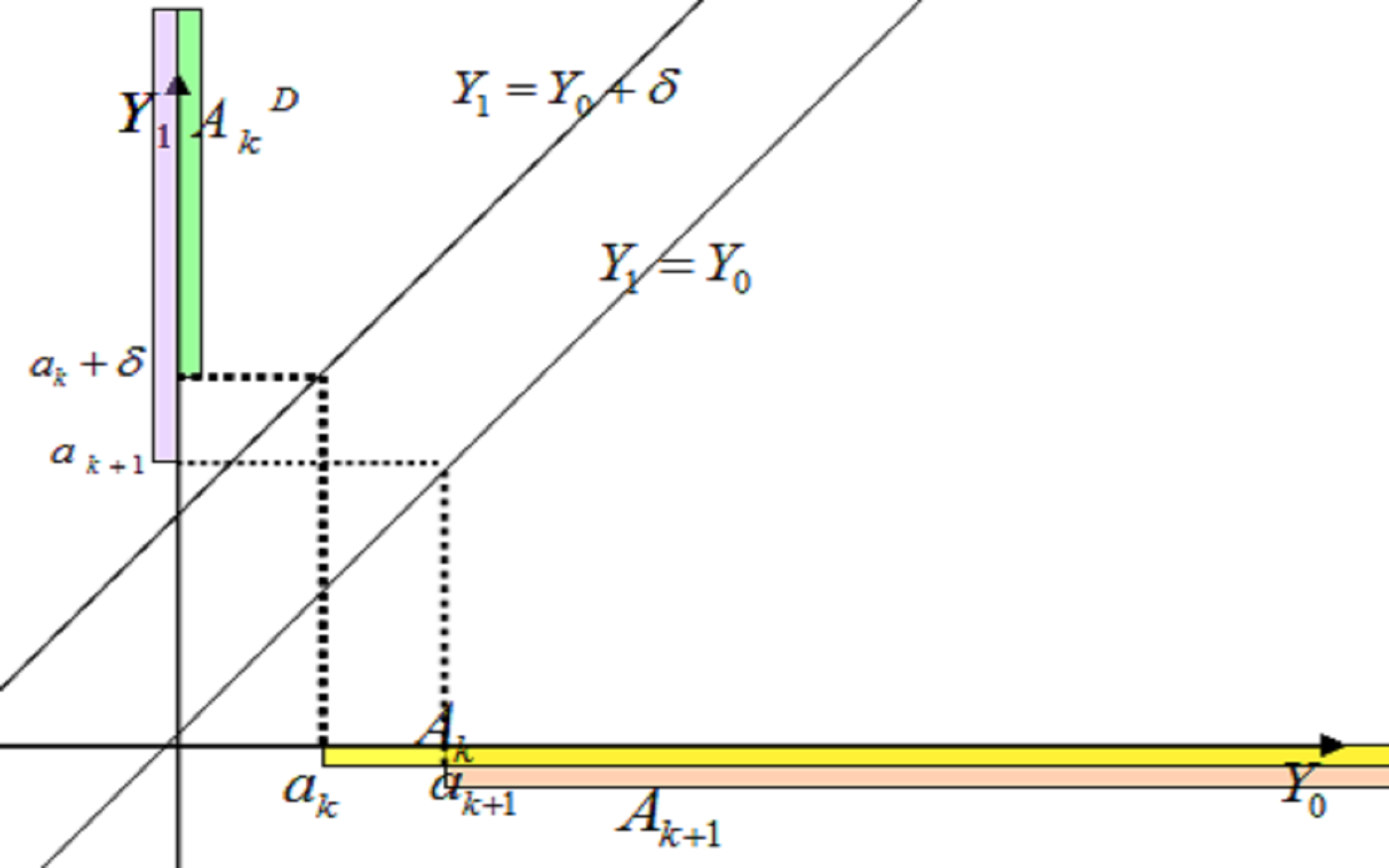}%
\\
{}%
\end{center}}}
\\
Figure. A.2: $A_{k}^{D}$ for $A_{k}=$ $\left(  a_{k},\infty\right)  $ and
$A_{k+1}=$ $\left(  a_{k+1},\infty\right)  $%
\end{tabular}

\end{center}

Since the function $\varphi$ is nondecreasing in the support of the optimal
joint distribution, $A_{k}$ reduces to $\left(  a_{k},\infty\right)  $ with
$a_{k}\leq a_{k+1}$ and $a_{k}\in\left[  -\infty,\infty\right]  $ where
$A_{k}=\phi$ for $a_{k}=\infty.$ By Theorem 1, for each integer
$k$ and $\delta>0,$
\begin{align*}
A_{k}^{D}  & =\left\{  y_{1}\in\mathbb{R}|\exists y_{0}>a_{k}\text{ s.t.
}y_{1}-y_{0}\geq\delta\right\}  \cup\left\{  y_{1}\in\mathbb{R}|\exists
y_{0}>a_{k+1}\text{ s.t. }0\leq y_{1}-y_{0}<\delta\right\} \\
& =\left(  a_{k}+\delta,\infty\right)  \cup\left(  a_{k+1},\infty\right) \\
& =\left(  \min\left\{  a_{k}+\delta,a_{k+1}\right\}  ,\infty\right)
\end{align*}
Then, $F_{0}\left(  A_{k}\right)  -F_{1}\left(  A_{k}^{D}\right)  =0$ for
$a_{k}=\infty$, while $F_{0}\left(  A_{k}\right)  -F_{1}\left(  A_{k}%
^{D}\right)  =\min\left\{  F_{1}\left(  a_{k}+\delta\right)  ,F_{1}\left(
a_{k+1}\right)  \right\}  -F_{0}\left(  a_{k}\right)  $ for $a_{k}<\infty.$
Therefore, By Theorem 1,
\begin{align*}
F_{\Delta}^{L}\left(  \delta\right)   & =\underset{\left\{  A_{k}\right\}
_{k=-\infty}^{\infty}}{\sup}\left[  \sum\limits_{k=-\infty}^{\infty}%
\max\left\{  \mu_{0}\left(  A_{k}\right)  -\mu_{1}\left(  A_{k}^{D}\right)
,0\right\}  \right] \\
& =\underset{\left\{  a_{k}\right\}  _{k=-\infty}^{\infty}}{\sup}\left[
\sum\limits_{k=-\infty}^{\infty}\max\left\{  \min\left\{  F_{1}\left(
a_{k}+\delta\right)  ,F_{1}\left(  a_{k+1}\right)  \right\}  -F_{0}\left(
a_{k}\right)  ,0\right\}  \right]  .
\end{align*}

Now I\ show that it is innocuous to assume that $a_{k+1}-a_{k}\leq\delta$ for
each integer $k.$\ Suppose that there exists an integer $l$ s.t.
$a_{l+1}>a_{l}+\delta$. Consider $\left\{  \widetilde{A}_{k}\right\}
_{k=-\infty}^{\infty}$ with $\widetilde{A}_{k}=\left(  \widetilde{a}%
_{k},\infty\right)  $ as follows:%
\begin{align*}
\widetilde{a}_{k}  & =a_{k}\text{ for }k\leq l,\\
\widetilde{a}_{l+1}  & =a_{l}+\delta,\\
\widetilde{a}_{k+1}  & =a_{k}\text{ for }k\geq l+1.
\end{align*}
It is obvious that $\widetilde{a}_{k+1}\leq\widetilde{a}_{k+2}$ for every
integer $k.$ $\widetilde{A}_{l}^{D}$ is given as%
\begin{align}
\widetilde{A}_{l}^{D}  & =\left(  \min\left\{  \widetilde{a}_{l}%
+\delta,\widetilde{a}_{l+1}\right\}  ,\infty\right) \label{equalities}\\
& =\left(  a_{l}+\delta,\infty\right) \nonumber\\
& =A_{l}^{D}%
\end{align}
The second equality in (A.11) follows from $\widetilde{a}%
_{l+1}=a_{l}+\delta=\widetilde{a}_{l}+\delta$, and the third equality holds
because
\begin{align*}
A_{l}^{D}  & =\left(  \min\left\{  a_{l}+\delta,a_{l+1}\right\}
,\infty\right) \\
& =\left(  a_{l}+\delta,\infty\right)  .
\end{align*}
This implies that
\begin{align*}
\max\left\{  \mu_{0}\left(  \widetilde{A}_{k}\right)  -\mu_{1}\left(
\widetilde{A}_{k}^{D}\right)  ,0\right\}   & =\max\left\{  \mu_{0}\left(
A_{k}\right)  -\mu_{1}\left(  A_{k}^{D}\right)  ,0\right\}  \text{ for }k\leq
l,\\
\max\left\{  \mu_{0}\left(  \widetilde{A}_{k+1}\right)  -\mu_{1}\left(
\widetilde{A}_{k+1}^{D}\right)  ,0\right\}   & =\max\left\{  \mu_{0}\left(
A_{k}\right)  -\mu_{1}\left(  A_{k}^{D}\right)  ,0\right\}  \text{ for }k\geq
l+1,
\end{align*}
Therefore,
\[
\sum\limits_{k=-\infty}^{\infty}\max\left\{  \mu_{0}\left(  A_{k}\right)
-\mu_{1}\left(  A_{k}^{D}\right)  ,0\right\}  \leq\sum\limits_{k=-\infty
}^{\infty}\max\left\{  \mu_{0}\left(  \widetilde{A}_{k}\right)  -\mu
_{1}\left(  \widetilde{A}_{k}^{D}\right)  ,0\right\}
\]
This means that for any sequence of sets $\left\{  A_{k}\right\}  _{k=-\infty
}^{\infty}$ with $a_{k+1}>a_{k}+\delta$ for some integer $k$, one can always
construct a seqeunce of sets $\left\{  \widetilde{A}_{k}\right\}  _{k=-\infty
}^{\infty}$ with $\widetilde{a}_{k+1}\leq\widetilde{a}_{k}+\delta$ for every
integer $k$ satisfying
\[
\sum\limits_{k=-\infty}^{\infty}\max\left\{  \mu_{0}\left(  \widetilde{A}%
_{k}\right)  -\mu_{1}\left(  \widetilde{A}_{k}^{D}\right)  ,0\right\}
\geq\sum\limits_{k=-\infty}^{\infty}\max\left\{  \mu_{0}\left(  A_{k}\right)
-\mu_{1}\left(  A_{k}^{D}\right)  ,0\right\}  .
\]

This can be intuitively understood by comparing Figure A.3(a) to Figure
A.3(b), where the sum of the lower bound on each triangle is equal to
$\sum\limits_{k=-\infty}^{\infty}\max\left\{  \mu_{0}\left(  A_{k}\right)
-\mu_{1}\left(  A_{k}^{D}\right)  ,0\right\}  $ and $\sum\limits_{k=-\infty
}^{\infty}\max\left\{  \mu_{0}\left(  \widetilde{A}_{k}\right)  -\mu
_{1}\left(  \widetilde{A}_{k}^{D}\right)  ,0\right\}  ,$ respectively.
Therefore, it is innocuous to assume $a_{k+1}\leq a_{k}+\delta$ at the optimum.

\begin{center}%
\begin{tabular}
[c]{ll}%
$%
{\includegraphics[
natheight=1.823000in,
natwidth=1.802300in,
height=2.7622in,
width=2.7311in
]%
{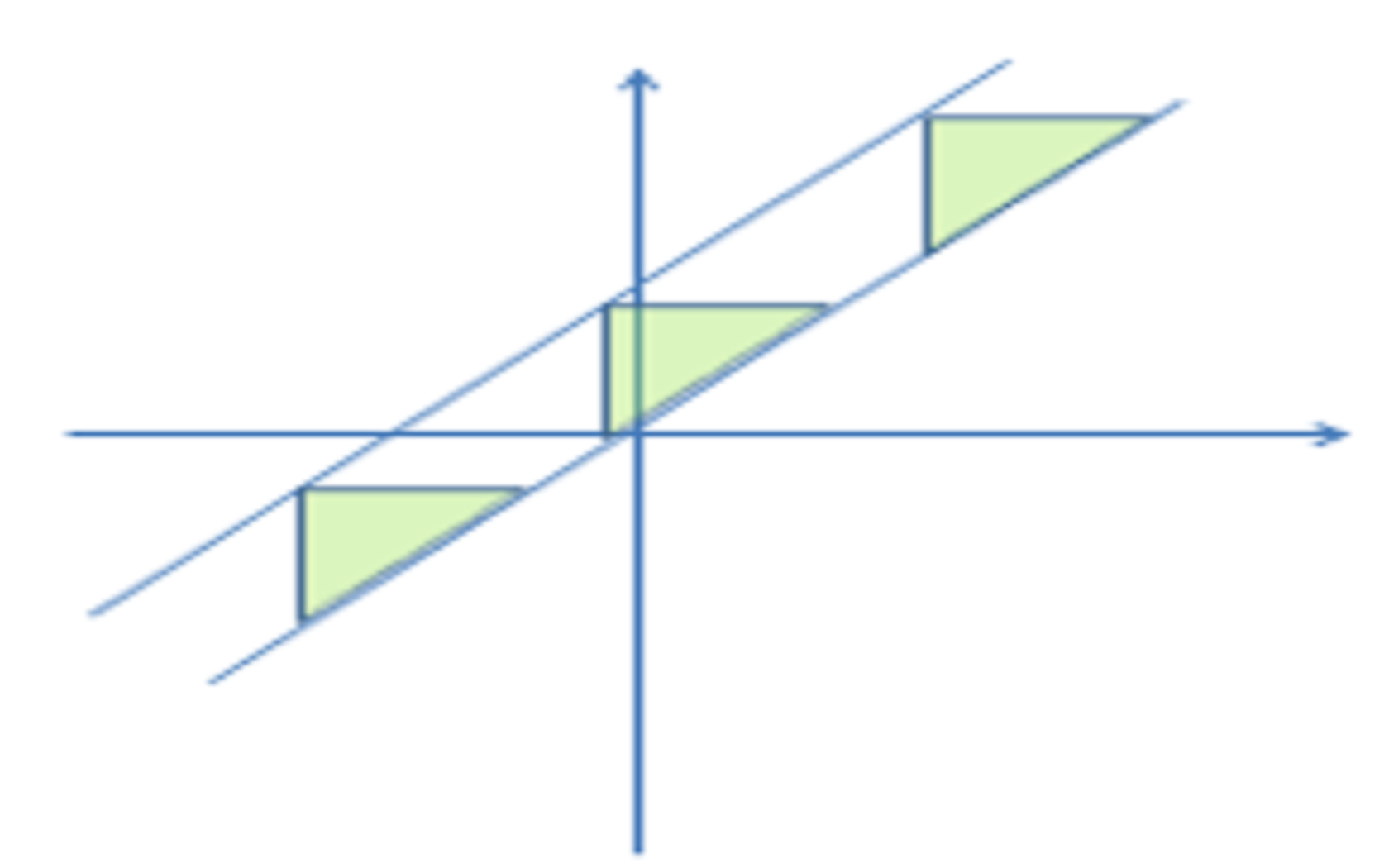}%
}
$ & $%
{\includegraphics[
natheight=1.802300in,
natwidth=1.718400in,
height=2.7311in,
width=2.6057in
]%
{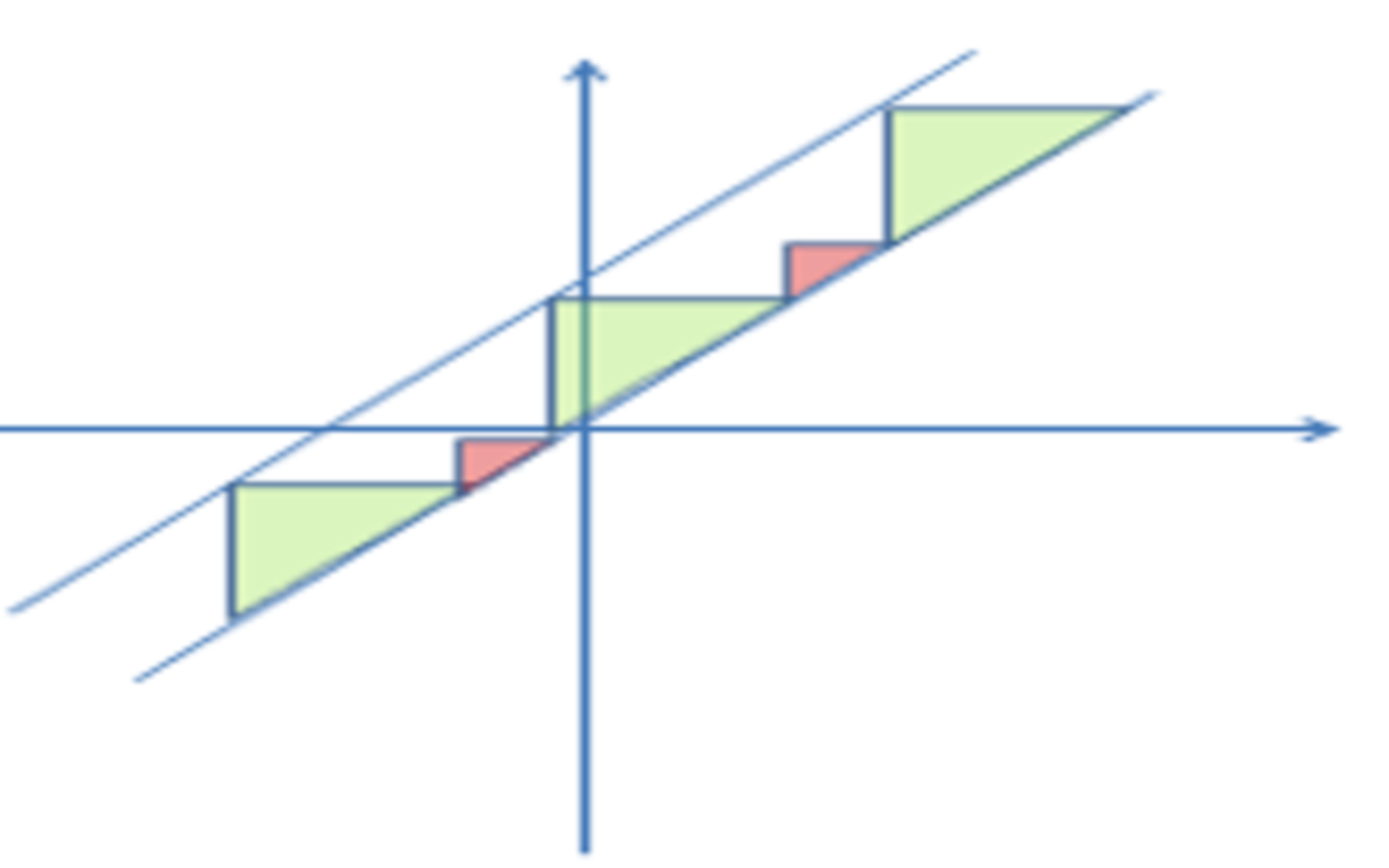}%
}
$\\
\multicolumn{1}{c}{(a)} & \multicolumn{1}{c}{(b)}\\
\multicolumn{2}{c}{Figure A.3:\ $a_{k+1}-a_{k}\leq\delta$ at the optimum}%
\end{tabular}
\\[0pt]
\end{center}

\textbf{Part 2. The upper bound}

First, I\ introduce the following lemma, which is useful for deriving the
upper bound under MTR.

\subparagraph{Lemma A.2}

(i) Let $f:\mathbb{R}\rightarrow\mathbb{R}$ be a continuous function. Suppose
that for any $x\in\mathbb{R},$ there exists $\varepsilon_{x}>0$ s.t.
$f(t_{0})\leq f(t_{1})$ whenever $x\leq t_{0}<t_{1}<x+\varepsilon_{x}.$ Then
$f$ is a nondecreasing function in $\mathbb{R}$. (ii) If there exists
$\varepsilon_{x}>0$ for any $x\in\mathbb{R}$\ s.t. $f(t_{0})\geq f(t_{1})$
whenever $x-\varepsilon_{x}\leq t_{0}<t_{1}<x,$ then $f$ is a nonincreasing
function in $\mathbb{R}$.

\subparagraph{Proof of Lemma A.2}

Since the proof of (ii) is very similar to the proof of (i), I\ provide only
the proof for (i). Suppose not. There exist $a$ and $b$ in $\mathbb{R}$ with
$a<b$ s.t. $f(a)>f(b)$. Define $V=\left\{  x\in\left[  a,b\right]
;f(a)>f(x)\right\}  .$ Since $V$ is a nonempty set with $b\in V$ and bounded
below by $a$, $V$ has an infimum $x_{0}\in\left[  a,b\right]  .$ Since $f$ is
continuous, $f(x_{0})=f(a).$ Note that $\ a\leq x_{0}<b.$ Pick $\varepsilon
_{x_{0}}>0$ satisfying $f(t_{0})\leq f(t_{1})$ whenever $x_{0}\leq t_{0}%
<t_{1}<x_{0}+\varepsilon_{x_{0}}.$ Since $x_{0}$ is an infimum of the set $V$,
there exists $t\in\left(  x_{0,}x_{0}+\varepsilon_{x_{0}}\right)  $ s.t.
$f(x_{0})>f(t).$ This is a contradiction. Thus, for any $a<b,$ $f(a)\leq
f(b).$ ${\small \blacksquare}$

\bigskip I\ prove that in the dual representation%
\begin{align*}
& \underset{F\in\Pi\left(  F_{0},F_{1}\right)  }{\inf}\int\left\{
\boldsymbol{1}\left\{  y_{1}-y_{0}>\delta\right\}  +\lambda\left(
\boldsymbol{1}\left(  y_{1}<y_{0}\right)  \right)  \right\}  d\pi\\
& =\underset{\left(  \varphi,\psi\right)  \in\Phi_{c}}{\sup}\int\varphi\left(
y_{0}\right)  d\mu_{0}+\int\psi\left(  y_{1}\right)  d\mu_{1},
\end{align*}
the function $\varphi$ is nonincreasing. Note that under $\Pr\left(
Y_{1}=Y_{0}\right)  =0$, $\Pr\left(  Y_{1}\geq Y_{0}\right)  =\Pr\left(
Y_{1}>Y_{0}\right)  =1$, and recall that
\[
\varphi\left(  y_{0}\right)  =\underset{y_{1}\geq y_{0}}{\inf}\left\{
\left\{  y_{1}-y_{0}>\delta\right\}  -\psi\left(  y_{1}\right)  \right\}  .
\]
Pick any $\left(  y_{0}^{\prime},y_{1}^{\prime}\right)  $ with $y_{1}^{\prime
}>y_{0}^{\prime}$ in the optimal support of the joint distribution. For any
$h$ s.t. $0<h<y_{1}^{\prime}-y_{0}^{\prime},$%
\begin{align}
\varphi\left(  y_{0}^{\prime}+h\right)   & =\underset{y_{1}>y_{0}^{\prime}%
+h}{\inf}\left\{  \boldsymbol{1}\left\{  y_{1}-\left(  y_{0}^{\prime
}+h\right)  >\delta\right\}  -\psi\left(  y_{1}\right)  \right\}
\label{AAA1}\\
& \leq\boldsymbol{1}\left\{  y_{1}^{\prime}-\left(  y_{0}^{\prime}+h\right)
>\delta\right\}  -\psi\left(  y_{1}^{\prime}\right) \nonumber\\
& \leq\boldsymbol{1}\left\{  y_{1}^{\prime}-y_{0}^{\prime}>\delta\right\}
-\psi\left(  y_{1}^{\prime}\right) \nonumber\\
& =\varphi\left(  y_{0}^{\prime}\right)  ,\nonumber
\end{align}
The inequality in the second line of (A.13) is satisfied because
$y_{1}^{\prime}>\left(  y_{0}^{\prime}+h\right)  ,$ and the inequality in the
third line of (A.13) holds since $\boldsymbol{1}\left\{  y_{1}%
-y_{0}>\delta\right\}  $ is nonincreasing in $y_{0}.$ By Lemma A.2, $\varphi$
is nonincreasing on $\mathbb{R}$.

\begin{center}%
\begin{tabular}
[c]{l}%
\raisebox{-0pt}{\parbox[b]{2.9231in}{\begin{center}
\includegraphics[
natheight=3.281100in,
natwidth=4.135500in,
height=2.3246in,
width=2.9231in
]%
{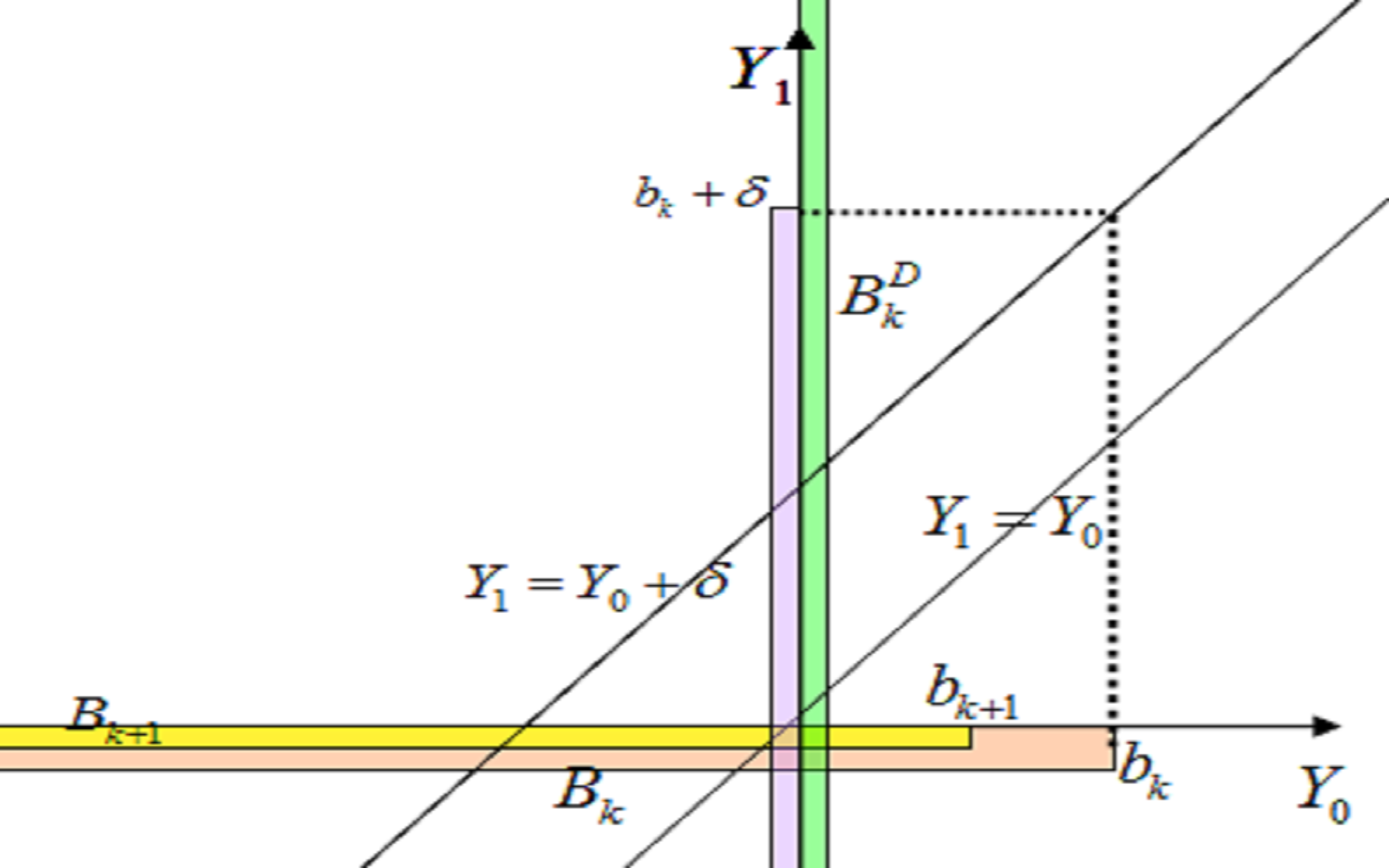}%
\\
{}%
\end{center}}}
\\
Figure A.4: $B_{k}^{D}$ for $B_{k}=\left(  -\infty,b_{k}\right)  $ and
$B_{k+1}=\left(  -\infty,b_{k+1}\right)  $%
\end{tabular}

\end{center}

Now, $B_{k}=\left\{  y\in\mathbb{R};\varphi>s+k\text{ }\right\}  =\left(
-\infty,b_{k}\right)  $ for each integer $k$, some $s\in(0,1]$ and $b_{k}%
\in\left[  -\infty,\infty\right]  ,$ in which $B_{k}=\phi$ for $b_{k}%
=-\infty.$ By Theorem 1, for each integer $k,$ $b_{k+1}\leq
b_{k}$\ and for $\delta>0,$%
\[
B_{k}^{D}=\left\{  y_{1}\in\mathbb{R};\exists y_{0}<b_{k}\text{ s.t. }0\leq
y_{1}-y_{0}<\delta\right\}  \cup\left\{  y_{1}\in\mathbb{R};\exists
y_{0}<b_{k+1}\text{ s.t. }y_{1}-y_{0}\geq\delta\right\}  .
\]
\newline If $b_{k}=-\infty,$ then $b_{k+1}=-\infty$ and so $B_{k}^{D}=\phi.$
For $b_{k}>-\infty,$ $B_{k}^{D}$ depends on the value of $b_{k+1}$ as
follows:
\[
B_{k}^{D}=\left\{
\begin{tabular}
[c]{ll}%
$\mathbb{R}$, & if $b_{k+1}>-\infty,$\\
$\left(  -\infty,b_{k}+\delta\right)  $, & if $b_{k+1}=-\infty.$%
\end{tabular}
\right.
\]
Pick any integer $k$. If $b_{k}=-\infty$, then%
\[
\max\left\{  \mu_{0}\left(  B_{k}\right)  -\mu_{1}\left(  B_{k}^{D}\right)
,0\right\}  =0.
\]
If $b_{k}>b_{k+1}>-\infty,$ then also
\[
\max\left\{  \mu_{0}\left(  B_{k}\right)  -\mu_{1}\left(  B_{k}^{D}\right)
,0\right\}  =0.
\]
If $b_{k}>b_{k+1}=-\infty,$ then%
\begin{align*}
& \max\left\{  \mu_{0}\left(  B_{k}\right)  -\mu_{1}\left(  B_{k}^{D}\right)
,0\right\} \\
& =\max\left\{  F_{0}\left(  b_{k}\right)  -F_{1}\left(  b_{k}+\delta\right)
,0\right\}  .
\end{align*}
Consequently, by Theorem 1, the sharp upper bound under MTR can
be written as%
\begin{align*}
F_{\Delta}^{U}\left(  \delta\right)   & =1-\underset{\left\{  B_{k}\right\}
_{k=-\infty}^{\infty}}{\sup}\sum\limits_{k=-\infty}^{\infty}\max\left\{
\mu_{0}\left(  B_{k}\right)  -\mu_{1}\left(  B_{k}^{D}\right)  ,0\right\} \\
& =1-\underset{b_{k}}{\sup}\max\left\{  F_{0}\left(  b_{k}\right)
-F_{1}\left(  b_{k}+\delta\right)  ,0\right\} \\
& =1+\underset{y}{\inf}\max\left\{  F_{1}\left(  y\right)  -F_{0}\left(
y-\delta\right)  ,0\right\}  .
\end{align*}
${\small \blacksquare}$

\subsection*{Proof of Corollary 2}

Since monotonicity of $\varphi$ can be shown very similarly as in the proof of
Corollary 1, I\ do not provide the proof. As given in Corollary
2, the sharp lower bound under concave treatment response is
identical to the sharp lower bound under MTR and the proof is also the same.
The sharp upper bound under convex treatment response is equal to the Makarov
upper bound by the same token as the upper bound under MTR. Thus, I\ do not
provide their proofs. Also, since the sharp lower bound under convex treatment
response is derived very similarly to the sharp upper bound under concave
treatment response, I\ provide a proof only for the sharp upper bound under
concave treatment response.

Consider\ a concave treatment response restriction $\Pr\left\{  \frac{Y_{0}%
-w}{t_{0}-t_{W}}\geq\frac{Y_{1}-Y_{0}}{t_{1}-t_{0}},Y_{1}\geq Y_{0}\geq
w\right\}  =1$ for any $w$ in the support of $W$ and $\left(  t_{1}%
,t_{0},t_{W}\right)  \in\mathbb{R}^{3}$ s.t. $t_{W}<t_{0}<t_{1}.$ The support
satisfying$\left\{  \frac{Y_{0}-w}{t_{0}-t_{W}}\geq\frac{Y_{1}-Y_{0}}%
{t_{1}-t_{0}},Y_{1}\geq Y_{0}\geq w\right\}  $ corresponds to the intersection
of the regions below the straight line $Y_{1}=\frac{t_{1}-t_{W}}{t_{0}-t_{W}%
}Y_{0}-\frac{t_{1}-t_{0}}{t_{0}-t_{W}}w$ and above the straight line
$Y_{1}=Y_{0}$ as shown in Figure A.5. Note that $\frac{t_{1}-t_{W}}%
{t_{0}-t_{W}}>1$ and the two straight lines intersect at $\left(  w,w\right)
$.

\begin{center}%
\begin{tabular}
[c]{c}%
{\includegraphics[
natheight=2.896300in,
natwidth=2.698200in,
height=2.9239in,
width=2.7259in
]%
{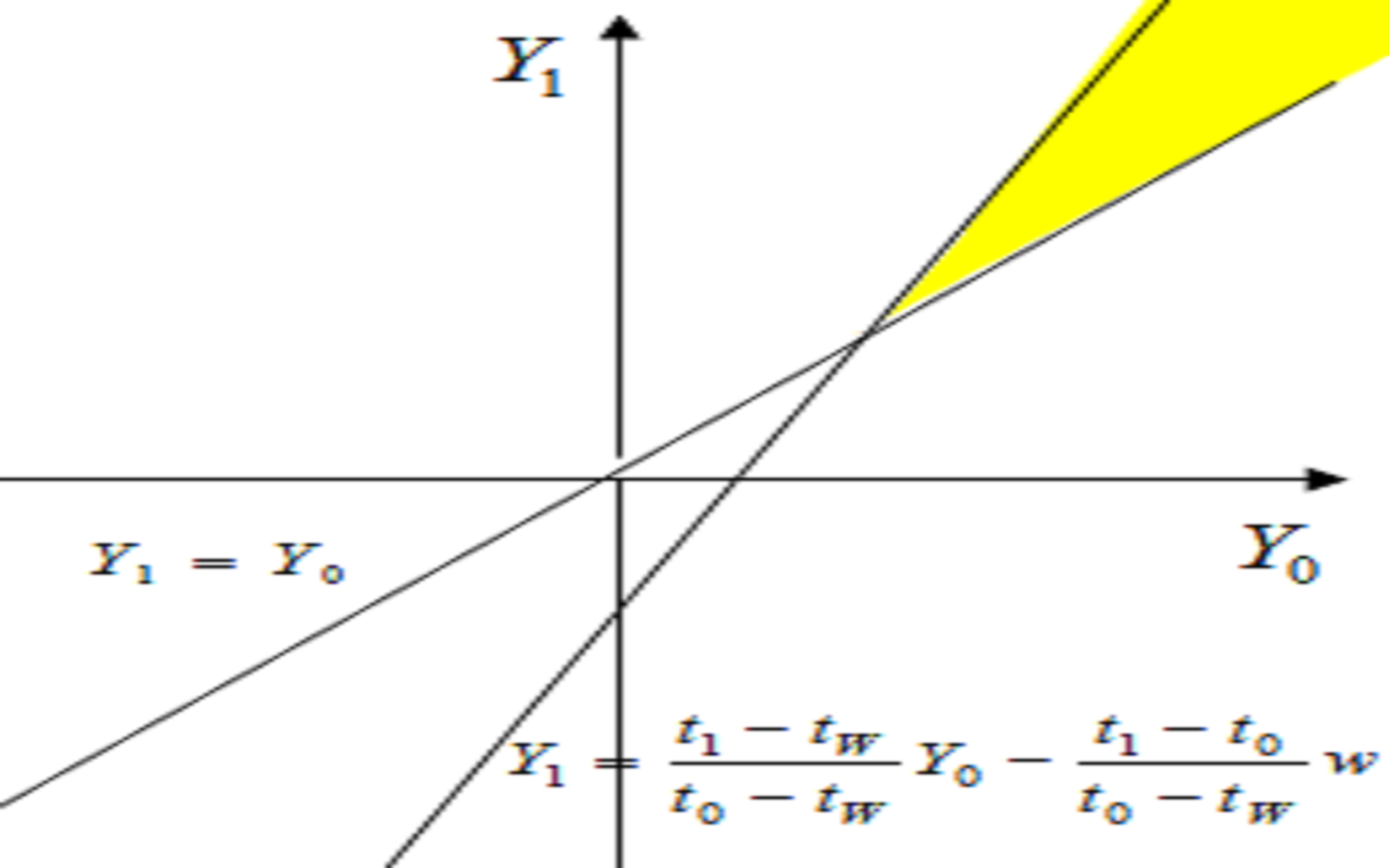}%
}
\\
Figure A.5: Support under concave treatment response
\end{tabular}

\end{center}

The function $\varphi$ can be readily shown to be nonincreasing. Thus, at the
optimum $B_{k}$ $=\left(  -\infty,b_{k}\right)  $ with $b_{k+1}\leq b_{k}$ and
$b_{k}\in\left[  -\infty,\infty\right]  $ for every integer $k$. By Theorem
1, for $\delta>0,$ $B_{k}^{D}$ is written as
\[
B_{k}^{D}=%
\begin{tabular}
[c]{l}%
$\left\{  y_{1}\in\mathbb{R}|\exists y_{0}<b_{k}\text{ s.t. }0\leq y_{1}%
-y_{0}<\delta\text{ and }\left(  t_{0}-t_{W}\right)  y_{1}-\left(  t_{1}%
-t_{W}\right)  y_{0}\leq-\left(  t_{1}-t_{0}\right)  w\right\}  $\\
$\cup\left\{  y_{1}\in\mathbb{R}|\exists y_{0}<b_{k+1}\text{ s.t. }y_{1}%
-y_{0}\geq\delta\text{ and }\left(  t_{0}-t_{W}\right)  y_{1}-\left(
t_{1}-t_{W}\right)  y_{0}\leq-\left(  t_{1}-t_{0}\right)  w\right\}  \text{
.}$%
\end{tabular}
\]

Note that $Y_{1}=Y_{0}+\delta$ and $Y_{1}=\frac{t_{1}-t_{W}}{t_{0}-t_{W}}%
Y_{0}-\frac{t_{1}-t_{0}}{t_{0}-t_{W}}w$ intersect at $\left(  \frac
{t_{0}-t_{W}}{t_{1}-t_{0}}\delta+y_{-1},\frac{t_{1}-t_{W}}{t_{1}-t_{0}}%
\delta+w\right)  .$ I\ consider the following three cases: a) $b_{k+1}\leq
b_{k}\leq\frac{t_{0}-t_{W}}{t_{1}-t_{0}}\delta+w,$ b) $b_{k+1}\leq\frac
{t_{0}-t_{W}}{t_{1}-t_{0}}\delta+w\leq b_{k},$ and c) $\frac{t_{0}-t_{W}%
}{t_{1}-t_{0}}\delta+w\leq b_{k+1}\leq b_{k}.$

\begin{center}%
\begin{tabular}
[c]{lll}%
{\includegraphics[
natheight=2.896300in,
natwidth=3.010400in,
height=1.7651in,
width=1.8334in
]%
{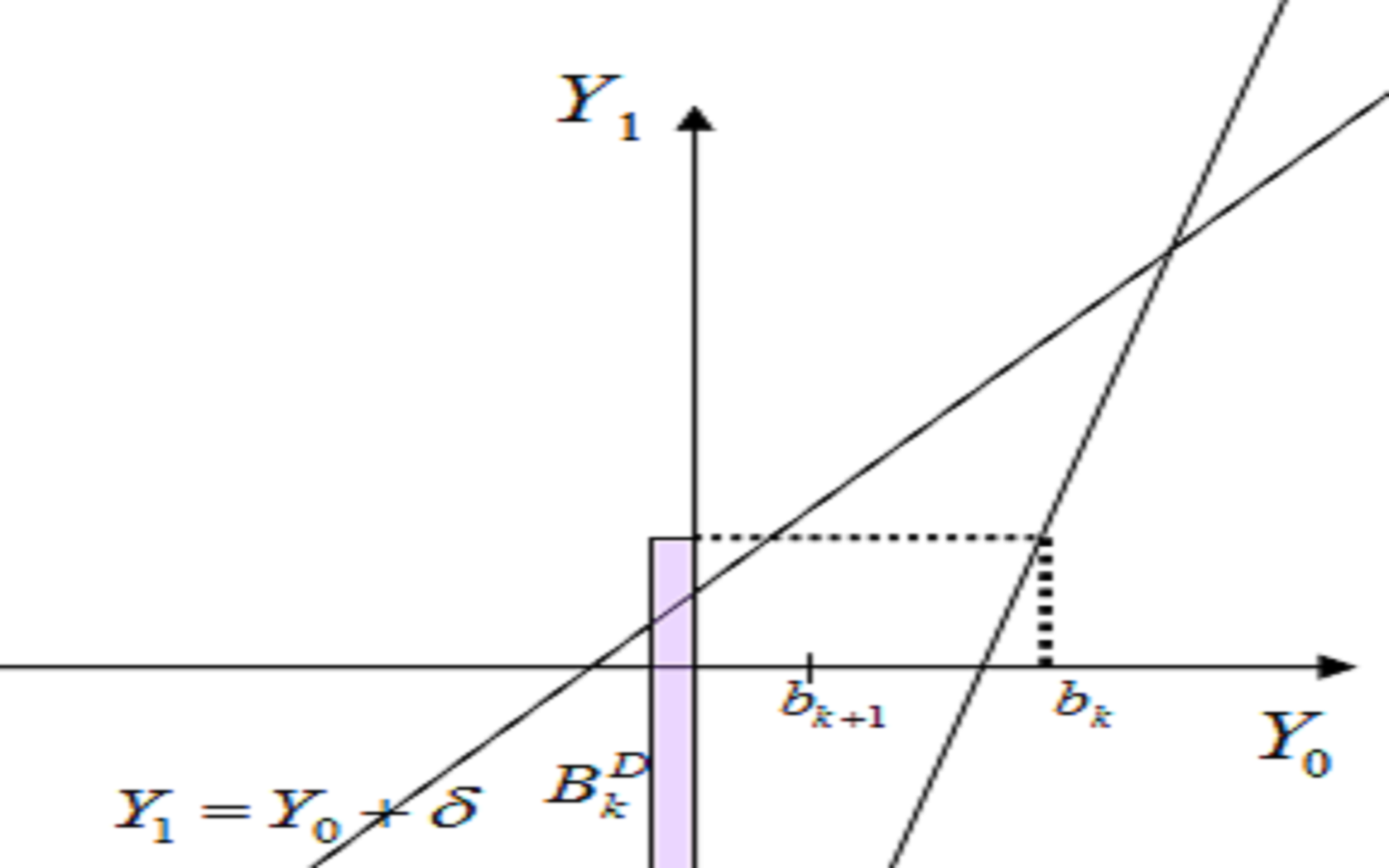}%
}
&
{\includegraphics[
natheight=3.135800in,
natwidth=3.114200in,
height=1.753in,
width=1.7409in
]%
{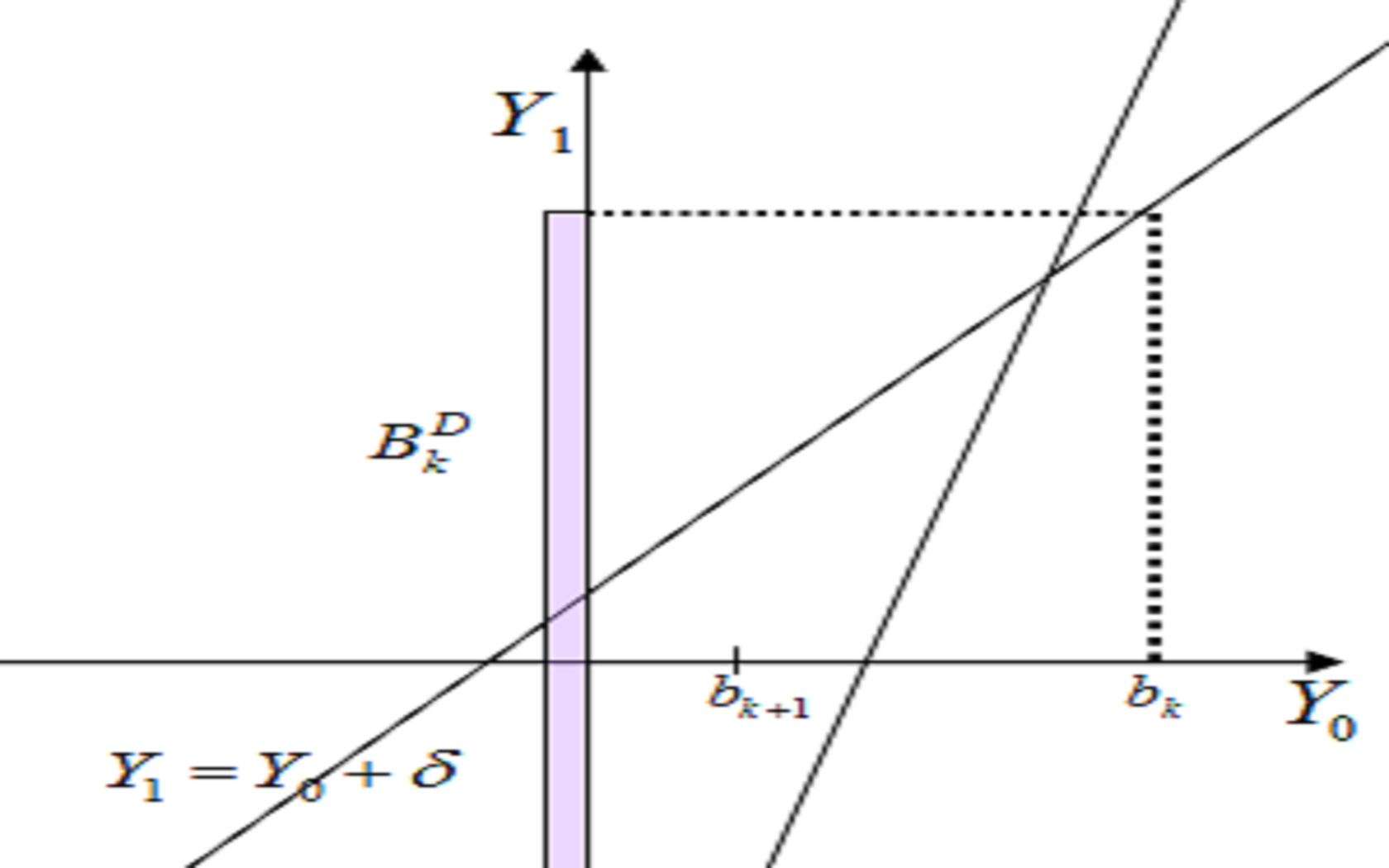}%
}
&
{\includegraphics[
natheight=3.416900in,
natwidth=3.885600in,
height=1.8057in,
width=2.0487in
]%
{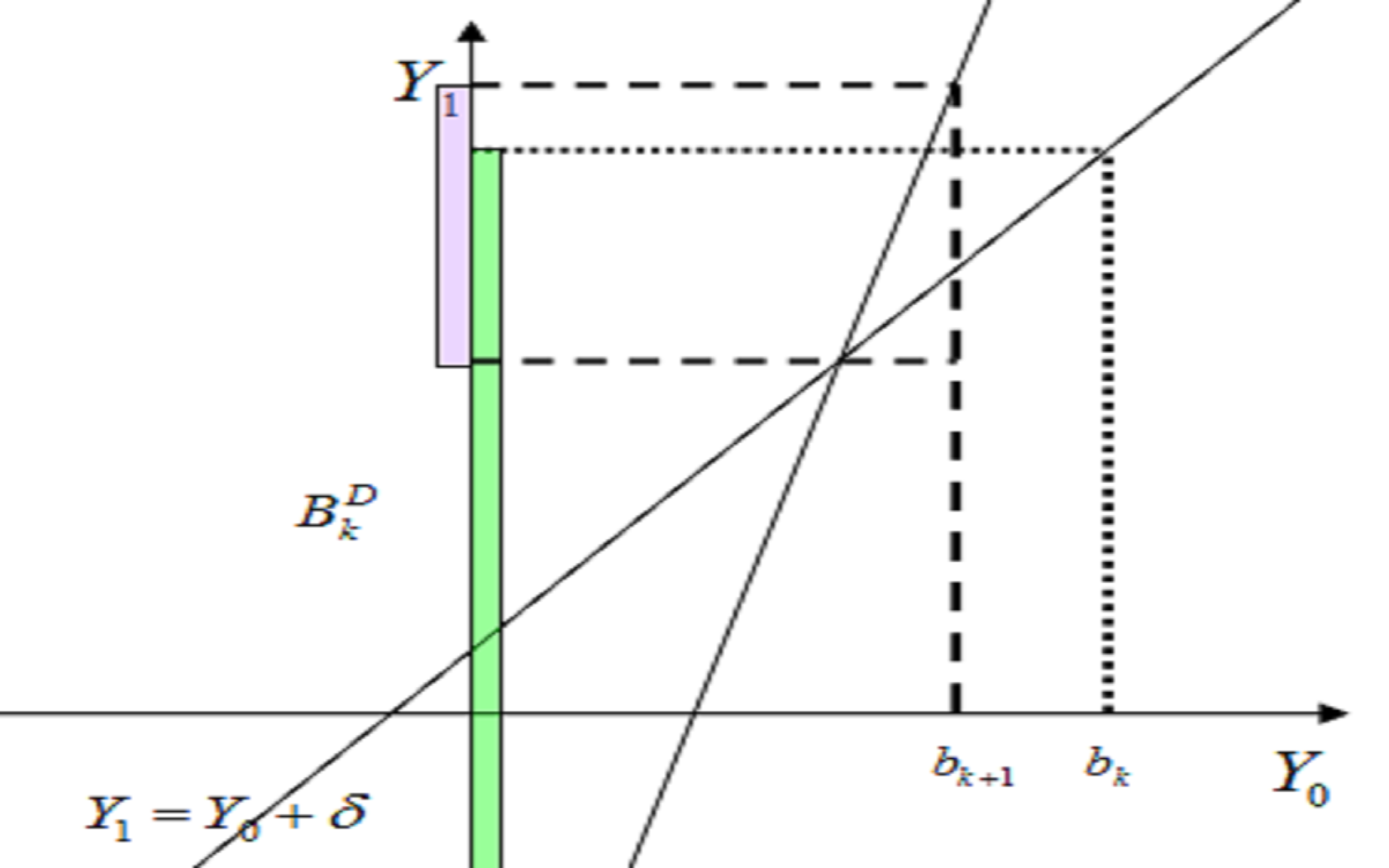}%
}
\\
\multicolumn{1}{c}{(a)} & \multicolumn{1}{c}{(b)} & \multicolumn{1}{c}{(c)}%
\end{tabular}

Figure. A.6: $B_{k}^{D}$ for $B_{k}=\left(  -\infty,b_{k}\right)  $ and
$B_{k+1}=\left(  -\infty,b_{k+1}\right)  $
\end{center}

\textbf{Case a)}{\Large \ }$b_{k+1}\leq b_{k}\leq\frac{t_{0}-t_{W}}%
{t_{1}-t_{0}}\delta+w$

If $b_{k+1}\leq b_{k}\leq\frac{t_{0}-t_{W}}{t_{1}-t_{0}}\delta+w,$ as
illustrated in Figure A.5(a), for any $y_{0}<b_{k+1}\leq\frac{t_{0}-t_{W}%
}{t_{1}-t_{0}}\delta+w,$ there exists no $y_{1}\in\mathbb{R}$ s.t.
$y_{1}-y_{0}\geq\delta$ and $\left(  t_{0}-t_{W}\right)  y_{1}-\left(
t_{1}-t_{W}\right)  y_{0}\leq-\left(  t_{1}-t_{0}\right)  w.$ Thus, for each
integer $k,$
\begin{align*}
B_{k}^{D}  & =\left(  -\infty,\frac{t_{1}-t_{W}}{t_{0}-t_{W}}b_{k}-\frac
{t_{1}-t_{0}}{t_{0}-t_{W}}w\right)  \cup\phi\\
& =\left(  -\infty,\frac{t_{1}-t_{W}}{t_{0}-t_{W}}b_{k}-\frac{t_{1}-t_{0}%
}{t_{0}-t_{W}}w\right)  .
\end{align*}

Let $\mu_{0,W}(\cdot|w)$ and $\mu_{1,W}(\cdot|w)$ denote conditional
distributions of $Y_{0}$ and $Y_{1}$ given $W=w$, while $F_{0,W}(\cdot|w)$ and
$F_{1,W}(\cdot|w)$ denote conditional distribution functions of $Y_{0}$ and
$Y_{1}$ given $W=w.$ Since $\Pr\left\{  \frac{Y_{0}-w}{t_{0}-t_{W}}\geq
\frac{Y_{1}-Y_{0}}{t_{1}-t_{0}}\right\}  =1,$ which is equivalent to
$\Pr\left\{  Y_{0}\geq\frac{t_{0}-t_{W}}{t_{1}-t_{W}}Y_{1}+\frac{t_{1}-t_{0}%
}{t_{1}-t_{W}}w\right\}  =1,$ implies
\[
F_{0,W}\left(  y|w\right)  \leq F_{1,W}\left(  \frac{t_{1}-t_{w}}{t_{0}-t_{w}%
}y-\frac{t_{1}-t_{0}}{t_{0}-t_{W}}w|w\right)  ,
\]
for each integer $k$,
\begin{align*}
& \mu_{0,W}\left(  B_{k}|w\right)  -\mu_{1,W}\left(  B_{k}^{D}|w\right) \\
& =F_{0,W}\left(  b_{k}|w\right)  -F_{1,W}\left(  \frac{t_{1}-t_{W}}%
{t_{0}-t_{W}}b_{k}-\frac{t_{1}-t_{0}}{t_{0}-t_{W}}w|w\right) \\
& \leq0.
\end{align*}

\textbf{Case b)}{\Large \ } $b_{k+1}\leq\frac{t_{0}-t_{W}}{t_{1}-t_{0}}%
\delta+w\leq b_{k}$

If $b_{k+1}\leq\frac{t_{0}-t_{W}}{t_{1}-t_{0}}\delta+w\leq b_{k},$ similar to
Case a, there exists no $y_{1}\in\mathbb{R}$ s.t. $y_{1}-y_{0}\geq\delta$ and
$\left(  t_{0}-t_{W}\right)  y_{1}-\left(  t_{1}-t_{W}\right)  y_{0}%
\leq-\left(  t_{1}-t_{0}\right)  w.$ Thus, for the same reason as in Case a,
\[
B_{k}^{D}=\left(  -\infty,\frac{t_{1}-t_{W}}{t_{0}-t_{W}}b_{k}-\frac
{t_{1}-t_{0}}{t_{0}-t_{W}}w\right)  ,
\]
and for every integer $k,$
\[
\mu_{0,W}\left(  B_{k}|w\right)  -\mu_{1,W}\left(  B_{k}^{D}|w\right)  \leq0.
\]

\textbf{Case c)}{\Large \ } $\frac{t_{0}-t_{W}}{t_{1}-t_{0}}\delta+w\leq
b_{k+1}\leq b_{k}$

If $\frac{t_{0}-t_{W}}{t_{1}-t_{0}}\delta+w\leq b_{k+1}\leq b_{k},$ then as
illustrated in Figure A.6(c),
\begin{align*}
B_{k}^{D}  & =\left(  -\infty,b_{k}+\delta\right)  \cup\left(  -\infty
,\frac{t_{1}-t_{W}}{t_{0}-t_{W}}b_{k+1}-\frac{t_{1}-t_{0}}{t_{0}-t_{W}%
}w\right) \\
& =\left(  -\infty,\max\left\{  b_{k}+\delta,\frac{t_{1}-t_{W}}{t_{0}-t_{W}%
}b_{k+1}-\frac{t_{1}-t_{0}}{t_{0}-t_{W}}w\right\}  \right)  .
\end{align*}
From Case a, b and c, it is innocuous to assume $\frac{t_{0}-t_{W}}%
{t_{1}-t_{0}}\delta+w\leq b_{k+1}\leq b_{k}$ for each integer $k.$

Furthermore, I\ show that it is innocuous to assume that $b_{k}+\delta
\leq\frac{t_{1}-t_{W}}{t_{0}-t_{W}}b_{k+1}-\frac{t_{1}-t_{0}}{t_{0}-t_{W}}w$
at the optimum. If there exists an integer $k$ s.t.
\[
b_{k}+\delta>\frac{t_{1}-t_{W}}{t_{0}-t_{W}}b_{k+1}-\frac{t_{1}-t_{0}}%
{t_{0}-t_{W}}w
\]
one can always construct $\left\{  \widetilde{B}_{k}\right\}  _{k=-\infty
}^{\infty}$ satisfying%
\begin{equation}
\sum\limits_{k=-\infty}^{\infty}\max\left\{  \mu_{0,W}\left(  B_{k}|w\right)
-\mu_{1,W}\left(  B_{k}^{D}|w\right)  ,0\right\}  \leq\sum\limits_{k=-\infty
}^{\infty}\max\left\{  \mu_{0,W}\left(  \widetilde{B}_{k}|w\right)  -\mu
_{1,W}\left(  \widetilde{B}_{k}^{D}|w\right)  ,0\right\}  ,\label{in3}%
\end{equation}
by defining $\widetilde{B}_{k}=\left(  -\infty,\widetilde{b}_{k}\right)  $ as
follows:
\begin{align*}
\widetilde{b}_{j}  & =b_{j}\text{ for }j\leq k,\\
\widetilde{b}_{k+1}  & =\frac{t_{0}-t_{W}}{t_{1}-t_{W}}\left(  b_{k}%
+\delta\right)  +\frac{t_{1}-t_{0}}{t_{1}-t_{W}}w,\\
\widetilde{b}_{j+1}  & =b_{j}\text{ for }j\geq k+1.
\end{align*}

\begin{center}
\bigskip%
\begin{tabular}
[c]{cc}%
{\includegraphics[
natheight=3.458400in,
natwidth=2.614300in,
height=2.7942in,
width=2.1197in
]%
{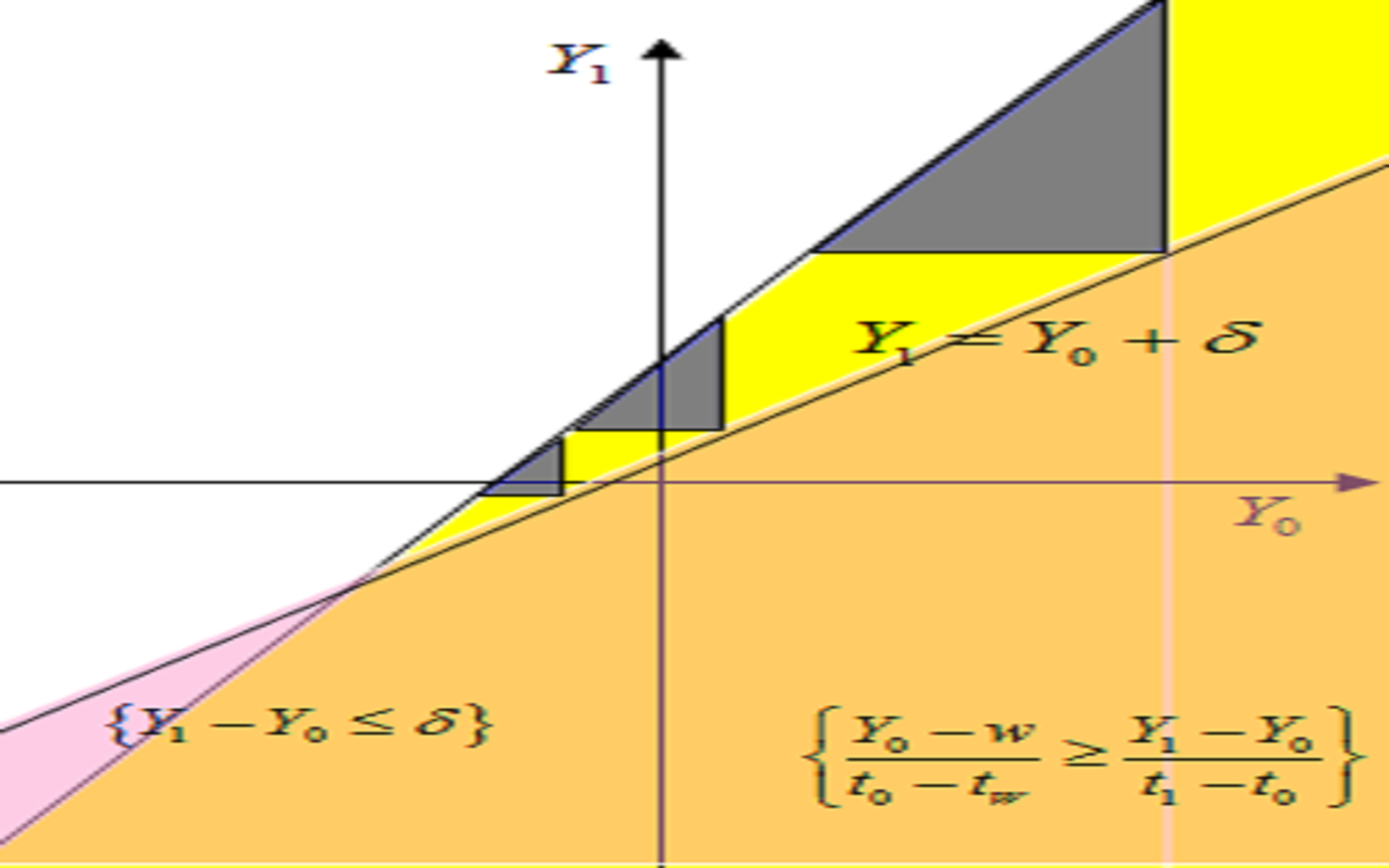}%
}
&
{\includegraphics[
natheight=3.479100in,
natwidth=2.489800in,
height=2.8106in,
width=2.0193in
]%
{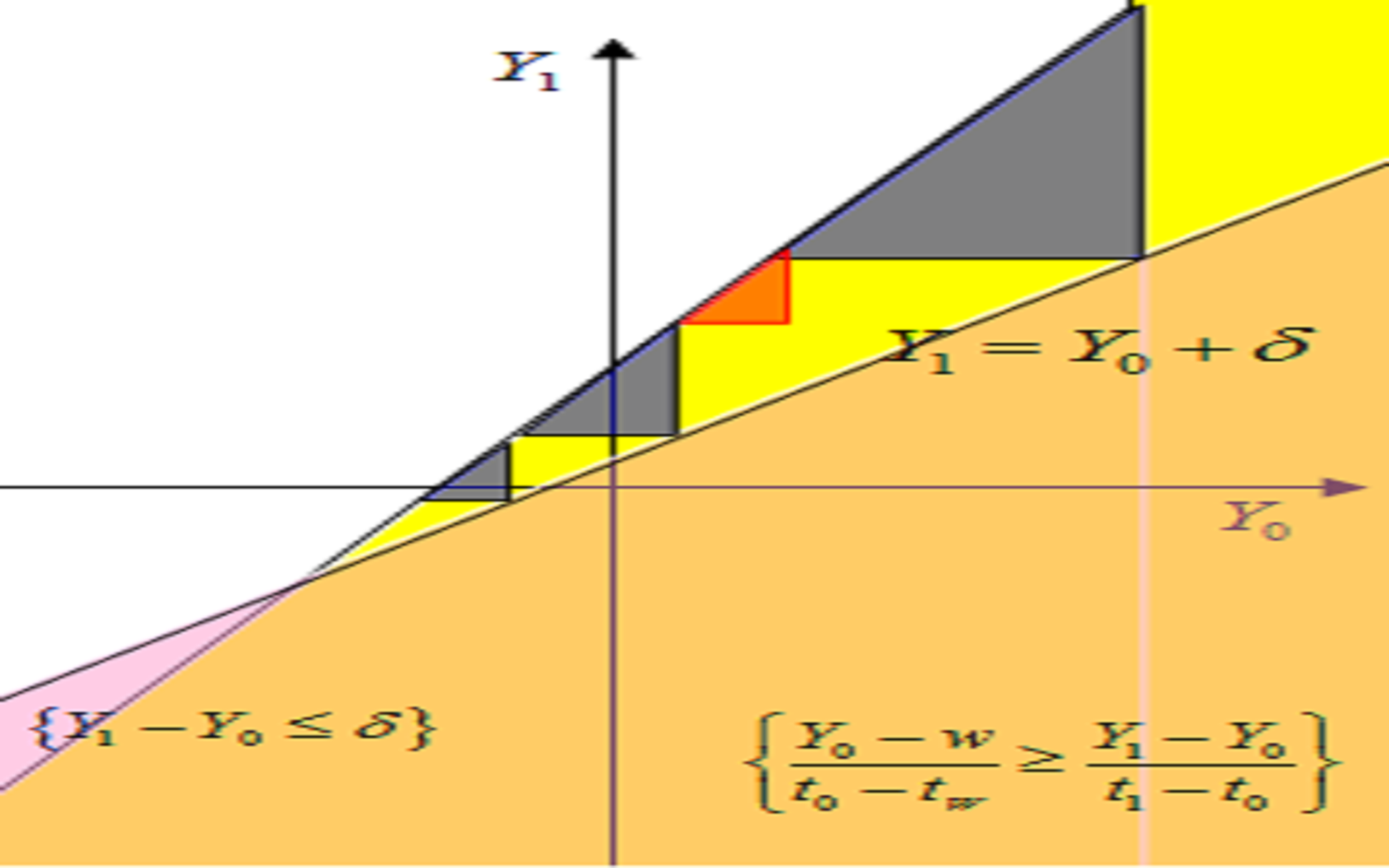}%
}
\\
(a) & (b)
\end{tabular}

Figure. A.7: $\sum\limits_{k=-\infty}^{\infty}\max\left\{  \mu_{0,W}\left(
B_{k}|w\right)  -\mu_{1,W}\left(  B_{k}^{D}|w\right)  ,0\right\}
{\small \leq}\sum\limits_{k=-\infty}^{\infty}\max\left\{  \mu_{0,W}\left(
\widetilde{B}_{k}|w\right)  -\mu_{1,W}\left(  \widetilde{B}_{k}^{D}|w\right)
,0\right\}  $
\end{center}

The inequality in (A.14) is illustrated in Figure A.7, which describes
\begin{align*}
& \sum\limits_{k=-\infty}^{\infty}\max\left\{  \mu_{0,W}\left(  B_{k}%
|w\right)  -\mu_{1,W}\left(  B_{k}^{D}|w\right)  ,0\right\}  ,\\
& \sum\limits_{k=-\infty}^{\infty}\max\left\{  \mu_{0,W}\left(  \widetilde
{B}_{k}|w\right)  -\mu_{1,W}\left(  \widetilde{B}_{k}^{D}|w\right)
,0\right\}
\end{align*}
in (a) and (b), respectively. Therefore, from consideration of Case a, b\ and
c,
\begin{align*}
& \underset{\left\{  B_{k}\right\}  _{k=-\infty}^{\infty}}{\sup}%
\sum\limits_{k=-\infty}^{\infty}\max\left\{  \mu_{0,W}\left(  B_{k}|w\right)
-\mu_{1,W}\left(  B_{k}^{D}|w\right)  ,0\right\} \\
& =\underset{\left\{  b_{k}\right\}  _{k=-\infty}^{\infty}}{\sup}%
\sum\limits_{k=-\infty}^{\infty}\max\left\{  F_{0,W}\left(  b_{k}|w\right)
-F_{1,W}\left(  \frac{t_{1}-t_{W}}{t_{0}-t_{W}}b_{k+1}-\frac{t_{1}-t_{0}%
}{t_{0}-t_{W}}w|w\right)  ,0\right\}
\end{align*}
where $\frac{t_{0}-t_{W}}{t_{1}-t_{0}}\delta+w\leq b_{k+1}\leq b_{k}%
.$\ Consequently, the sharp upper bound is written as follows: letting
$F_{\Delta,W}^{U}\left(  \delta|w\right)  $ be the sharp upper bound on
$\Pr\left(  Y_{1}-Y_{0}\leq\delta|W=w\right)  ,$
\begin{align*}
& F_{\Delta}^{U}\left(  \delta\right) \\
& =\int F_{\Delta,W}^{U}\left(  \delta|w\right)  dF_{W}\left(  w\right) \\
& =\int\left\{  1-\underset{\left\{  B_{k}\right\}  _{k=-\infty}^{\infty}%
}{\sup}\sum\limits_{k=-\infty}^{\infty}\max\left\{  \mu_{0,W}\left(
B_{k}|w\right)  -\mu_{1,W}\left(  B_{k}^{D}|w\right)  ,0\right\}  \right\}
dF_{W}\\
& =1+\int\underset{\left\{  b_{k}\right\}  _{k=-\infty}^{\infty}}{\inf}%
\sum\limits_{k=-\infty}^{\infty}\min\left\{  F_{1,W}\left(  \frac{t_{1}-t_{W}%
}{t_{0}-t_{W}}b_{k+1}-\frac{t_{1}-t_{0}}{t_{0}-t_{W}}w|w\right)
-F_{0,W}\left(  b_{k}|w\right)  ,0\right\}  dF_{W}%
\end{align*}
where $\frac{t_{0}-t_{W}}{t_{1}-t_{0}}\delta+w\leq b_{k+1}\leq b_{k}.$
${\small \blacksquare}$

\setcounter{figure}{0} \renewcommand{\thefigure}{B.\arabic{figure}}
\setcounter{table}{0} \renewcommand{\thetable}{B.\arabic{table}}
\setcounter{equation}{0} \renewcommand{\theequation}{B.\arabic{equation}}

\section*{Appendix B}

Appendix B presents the procedure used to compute the sharp lower bound under
MTR in Section 4 and Section 5. The following lemma is useful for reducing
computational costs:

\begin{description}
\item[Lemma B.1] Let
\begin{align*}
\left\{  a_{k}\right\}  _{k=-\infty}^{\infty}  & \in\underset{\left\{
a_{k}\right\}  _{k=-\infty}^{\infty}\in\mathcal{A}_{\delta}}{\arg\max}%
\sum\limits_{k=-\infty}^{\infty}\max\left\{  F_{1}\left(  a_{k+1}\right)
-F_{0}\left(  a_{k}\right)  ,0\right\}  ,\\
\text{where }\mathcal{A}_{\delta}  & =\left\{  \left\{  a_{k}\right\}
_{k=-\infty}^{\infty};0\leq a_{k+1}-a_{k}\leq\delta\text{ for each integer
}k\right\}  .
\end{align*}
It is innocuous to assume that $\left\{  a_{k}\right\}  _{k=-\infty}^{\infty}
$ satisfies $a_{k+2}-a_{k}>\delta$ for each integer $k$.
\end{description}

\begin{proof}
I\ will show that for any sequence $\left\{  a_{k}\right\}  _{k=-\infty
}^{\infty}\in\mathcal{A}_{\delta}$ satisfying $a_{k+2}-a_{k}\leq\delta$ for
some integer $k,$ one can construct $\left\{  \widetilde{a}_{k}\right\}
_{k=-\infty}^{\infty}\in\mathcal{A}_{\delta}$ with $\widetilde{a}%
_{k+2}-\widetilde{a}_{k}>\delta$ for each integer $k$ and$\ $
\[
\sum\limits_{k=-\infty}^{\infty}\max\left\{  F_{1}\left(  a_{k+1}\right)
-F_{0}\left(  a_{k}\right)  ,0\right\}  \leq\sum\limits_{k=-\infty}^{\infty
}\max\left\{  F_{1}\left(  \widetilde{a}_{k+1}\right)  -F_{1}\left(
\widetilde{a}_{k}\right)  ,0\right\}  .
\]
Suppose that there exists an integer $l$ s.t. $a_{l+2}-a_{l}\leq\delta.$ Let
\begin{align*}
\widetilde{a}_{k}  & =a_{k}\text{ for }k\leq l,\\
\widetilde{a}_{k}  & =a_{k+1}\text{ for }k\geq l+1.
\end{align*}
Then
\begin{align*}
& \sum\limits_{k=-\infty}^{\infty}\max\left\{  F_{1}\left(  a_{k+1}\right)
-F_{0}\left(  a_{k}\right)  ,0\right\} \\
& =\sum\limits_{k=-\infty}^{l-1}\max\left\{  F_{1}\left(  a_{k+1}\right)
-F_{0}\left(  a_{k}\right)  ,0\right\}  +\max\left\{  F_{1}\left(
a_{l+1}\right)  -F_{0}\left(  a_{l}\right)  ,0\right\} \\
& +\max\left\{  F_{1}\left(  a_{l+2}\right)  -F_{0}\left(  a_{l+1}\right)
,0\right\}  +\sum\limits_{k=l+2}^{\infty}\max\left\{  F_{1}\left(
a_{k+1}\right)  -F_{0}\left(  a_{k}\right)  ,0\right\} \\
& \leq\sum\limits_{k=-\infty}^{l-1}\max\left\{  F_{1}\left(  a_{k+1}\right)
-F_{0}\left(  a_{k}\right)  ,0\right\}  +\max\left\{  F_{1}\left(
a_{l+2}\right)  -F_{0}\left(  a_{l}\right)  ,0\right\} \\
& +\sum\limits_{k=l+2}^{\infty}\max\left\{  F_{1}\left(  a_{k+1}\right)
-F_{0}\left(  a_{k}\right)  ,0\right\} \\
& =\sum\limits_{k=-\infty}^{\infty}\max\left\{  F_{1}\left(  \widetilde
{a}_{k+1}\right)  -F_{0}\left(  \widetilde{a}_{k}\right)  ,0\right\}  .
\end{align*}
The inequality in the fourth line holds because MTR implies stochastic
dominance of $Y_{1}$ over $Y_{0}$. This is illustrated in Figure A.3(a) and
(b), where the sum of the lower bound on each triangle is equal to
$\sum\limits_{k=-\infty}^{\infty}\max\left\{  F_{1}\left(  a_{k+1}\right)
-F_{0}\left(  a_{k}\right)  ,0\right\}  $ and $\sum\limits_{k=-\infty}%
^{\infty}\max\left\{  F_{1}\left(  \widetilde{a}_{k+1}\right)  -F_{0}\left(
\widetilde{a}_{k}\right)  ,0\right\}  ,$ respectively.\newline%
\begin{tabular}
[c]{ll}%
$%
{\includegraphics[
natheight=2.885000in,
natwidth=3.510300in,
height=2.047in,
width=2.4855in
]%
{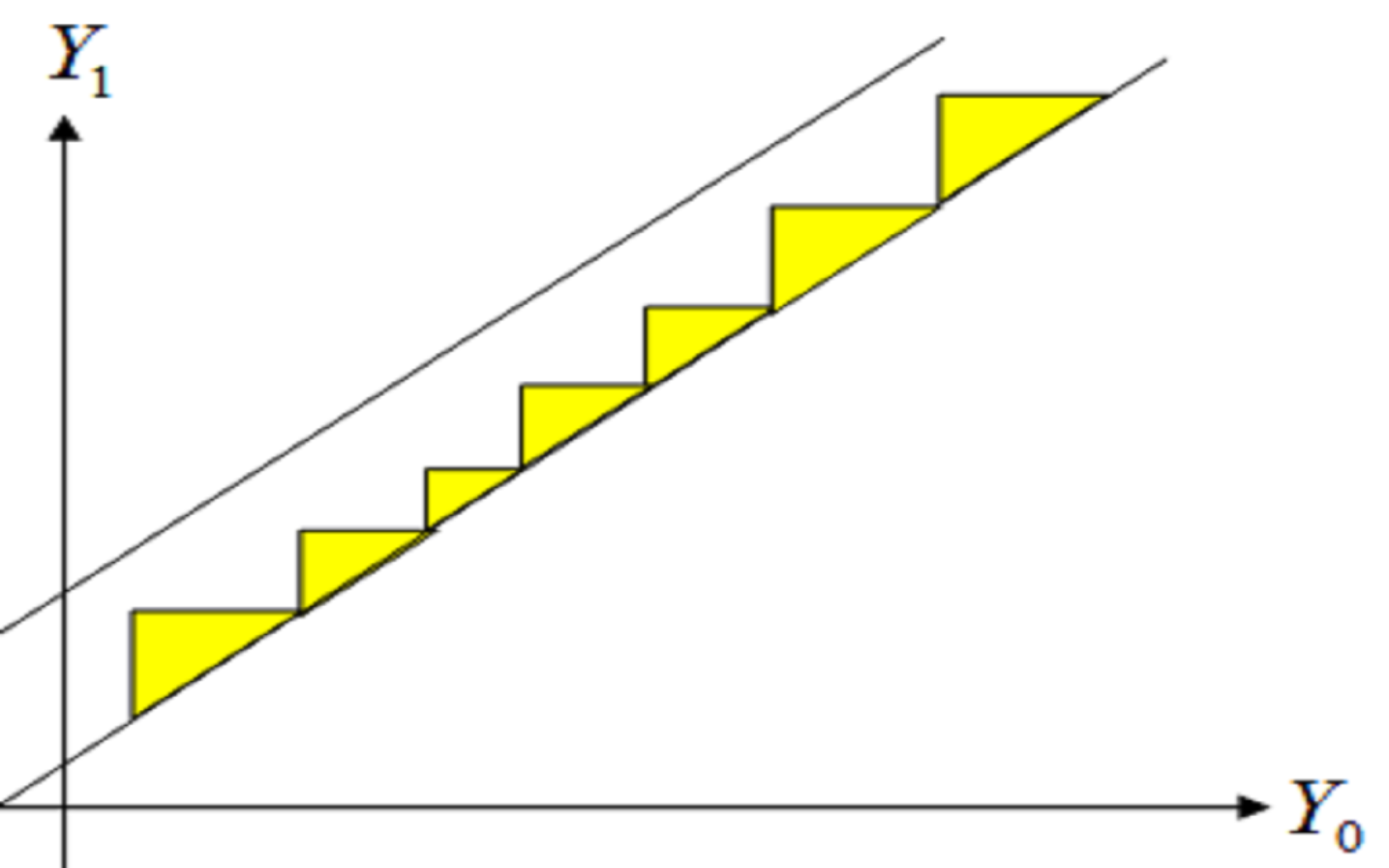}%
}
$ & $%
{\includegraphics[
natheight=2.958500in,
natwidth=3.551800in,
height=2.0989in,
width=2.514in
]%
{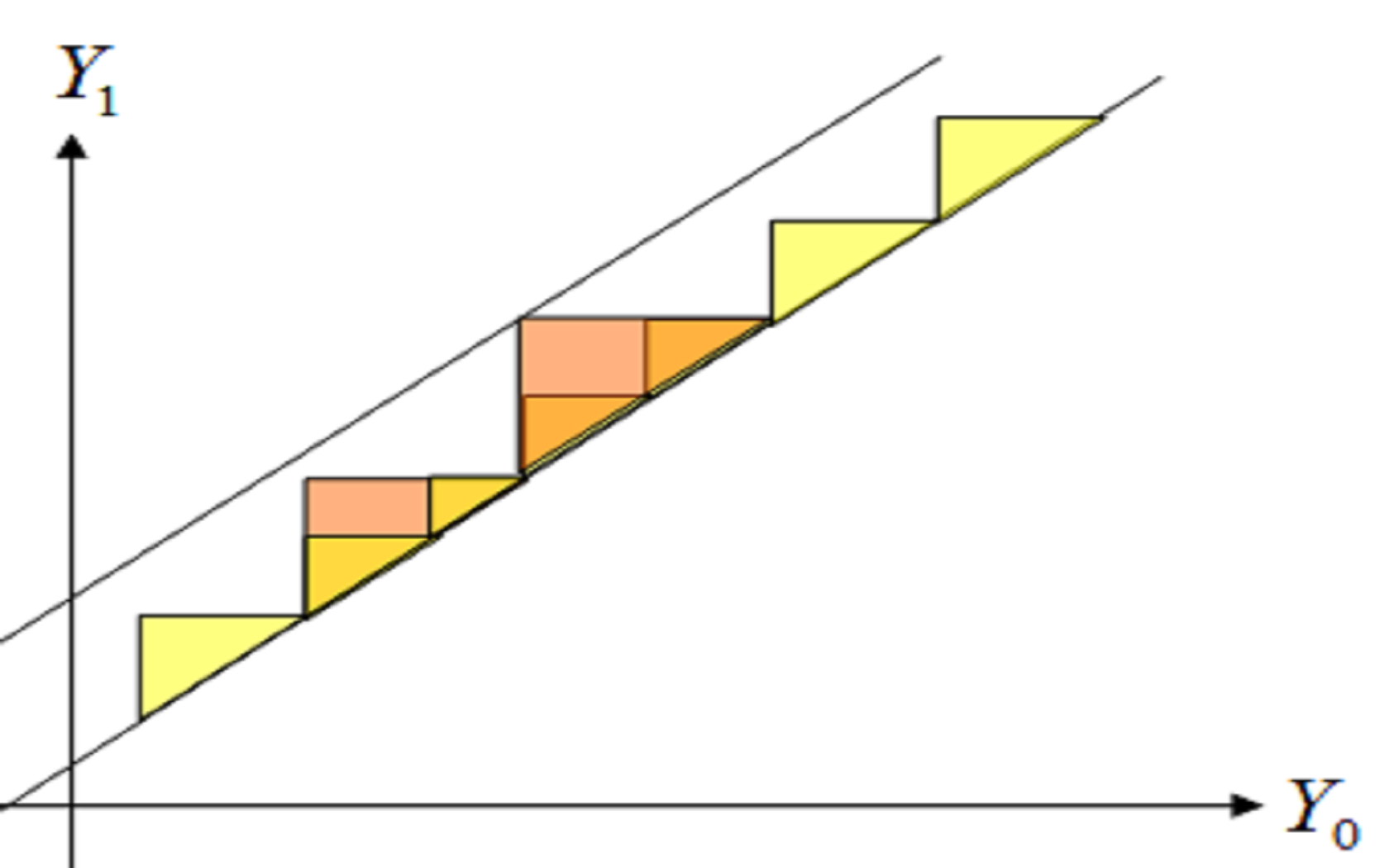}%
}
$\\
\multicolumn{1}{c}{(a)} & \multicolumn{1}{c}{(b)}\\
\multicolumn{2}{c}{Figure B.1: $a_{k+2}-a_{k}>\delta$ at the optimum}%
\end{tabular}
\newline Therefore, it is innocuous to assume $a_{k+2}-a_{k}>\delta$ for every
integer $k$ at the optimum.
\end{proof}

Now I\ present the constrained optimization procedure to compute the sharp
lower bound under MTR. I\ pay particular attention to the special case where
$a_{k+1}-a_{k}=\delta$ for each integer $k$ at the optimum. In this case, the
lower bound reduces to
\begin{equation}
\underset{0\leq y\leq\delta}{\sup}\sum_{k=-\infty}^{\infty}\max\left(
F_{1}\left(  y+\left(  k+1\right)  \delta\right)  -F_{0}\left(  y+k\delta
\right)  ,0\right)  ,\label{b.1}%
\end{equation}
and computation of (B.1) poses a simple one-dimensional optimization problem.

Let
\[
V\left(  \delta\right)  =\underset{0\leq y\leq\delta}{\sup}\sum_{k=-\infty
}^{\infty}\max\left(  F_{1}\left(  y+\left(  k+1\right)  \delta\right)
-F_{0}\left(  y+k\delta\right)  ,0\right)  ,
\]
and%
\[
V_{K}\left(  \delta\right)  =\underset{y\in\left\{  y^{\ast}+k\delta\right\}
_{k=-\infty}^{\infty}}{\max}\sum_{k=-K}^{K}\max\left(  F_{1}\left(  y+\left(
k+1\right)  \delta\right)  -F_{0}\left(  y+k\delta\right)  ,0\right)  ,
\]
where $y^{\ast}\in\underset{0\leq y\leq\delta}{\arg\max}\sum_{k=-\infty
}^{\infty}\max\left(  F_{1}\left(  y+\left(  k+1\right)  \delta\right)
-F_{0}\left(  y+k\delta\right)  ,0\right)  $ and $K$ is a nonnegative integer.

\begin{description}
\item[Step 1.] Compute $V\left(  \delta\right)  .$

\item[Step 2.] To further reduce computational costs, set $K$ to be a
nonnegative integer satisfying $\left\vert V\left(  \delta\right)
-V_{K}\left(  \delta\right)  \right\vert <\varepsilon$ for small
$\varepsilon>0.\footnote{I\ put $\varepsilon=10^{-5}$ for the implementation
in Section 4 and Section 5.}$

\item[Step 3.] For $J=K$, solve the following optimization problem:%
\begin{equation}
\underset{\left\{  a_{k}\right\}  _{k=-J}^{J}\in\mathcal{S}_{\delta}%
^{J,K}\left(  \widehat{y}\right)  }{\sup}\sum\limits_{k=-J}^{J}\max\left\{
F_{1}\left(  a_{k+1}\right)  -F_{0}\left(  a_{k}\right)  ,0\right\}
,\label{b.2}%
\end{equation}
where%
\begin{align*}
\mathcal{S}_{\delta}^{J,K}\left(  y\right)   & =\left\{
\begin{array}
[c]{c}%
\left\{  a_{k}\right\}  _{k=-J}^{J};a_{J}\leq y+K\delta,a_{-J}\geq
y-K\delta,\text{ }0\leq a_{k+1}-a_{k}\leq\delta,\\
\delta<a_{k+2}-a_{k}\text{ for each integer }k
\end{array}
\right\}  ,\\
\widehat{y}  & =\underset{y\in\left\{  y^{\ast}+k\delta\right\}  _{k=-\infty
}^{\infty}}{\arg\max}\sum_{k=-K}^{K}\max\left(  F_{1}\left(  y+\left(
k+1\right)  \delta\right)  -F_{0}\left(  y+k\delta\right)  ,0\right)  .
\end{align*}

\item[Step 4.] Repeat Step 3 for $J=K+1,\ldots,2K.\footnote{By Lemma B.1,
I$\ $considered $J=K,$ $K+1,\ldots,2K$ for the sequence $\left\{
a_{k}\right\}  _{k=-J}^{J}$ and compared the values of local maxima achieved
by $\left\{  a_{k}\right\}  _{k=-J}^{J}$ with $V_{K}\left(  \delta\right)  $}$
\end{description}

It is not straightforward to solve the problem (B.2) numerically in Step
3; the function $\max\{x,0\}$ is nondifferentiable. Furthermore in practice,
marginal distribution functions are often estimated in a complicated form to
compute their Jacobian and Hessian. To overcome this problem, I\ approximate
the nondifferentiable function $\max\{x,0\}$ with a smooth function $\frac
{x}{1+\exp\left(  -x/h\right)  }$ for small $h>0$ and marginal distribution
functions with finite normal mixtures $\sum\limits_{i}a_{i}\Phi\left(
\frac{x-\mu_{i}}{\sigma_{i}}\right)  ,$ which makes it substantially simple to
evaluate the Jacobian and Hessian of the objective function at any
point.\footnote{I\ used the Kolmogorov-Smirnov test to determine the number of
components in the mixture model. I\ increased the order of the mixture model
from one until the test does not reject the null that the two distribution
functions are identical. In the numerical example, I\ used one to three
components for $9$ different pairs of $(k_{1},k_{2})$ considered in Section 4
and I\ used three for the empirical application. For each mixture model that
I\ used to approximate the marginal distributions, the null hypothesis that
two distribution functions are identical was not rejected with pvalue$>0.99$.}%

\begin{tabular}
[c]{ll}%
{\includegraphics[
natheight=5.833200in,
natwidth=7.778100in,
height=1.6613in,
width=2.7501in
]%
{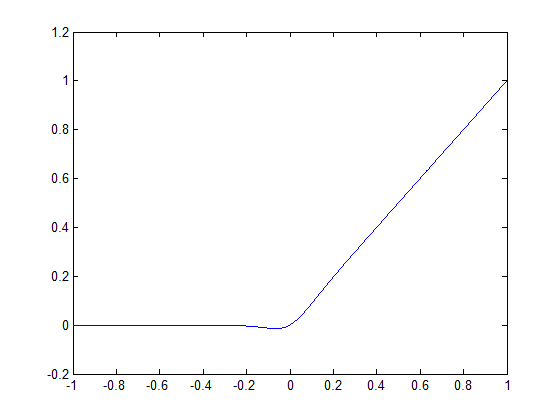}%
}
&
{\includegraphics[
natheight=5.833200in,
natwidth=7.778100in,
height=1.6613in,
width=2.7501in
]%
{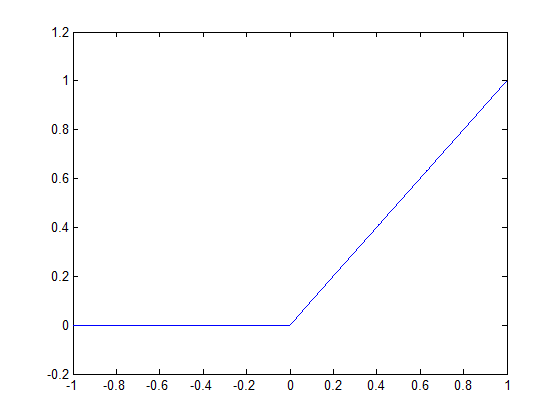}%
}
\\
\multicolumn{1}{c}{(a) $h=0.05$} & \multicolumn{1}{c}{(b) $h=0.01$}\\
\multicolumn{2}{c}{Figure B.2: Approximation of $\max\{x,0\}$ and $\frac
{x}{1+\exp\left(  -x/h\right)  }$}%
\end{tabular}

I\ used Knitro to solve the optimization problem using the smoothed functions.
Knitro is a constrained nonlinear optimization software.\footnote{Recently
Knitro has been often used to solve large-dimensional constrained optimization
problems in the literature including Conlon (2012), Dub\'{e} et al. (2012) and
Galichon and Salani\'{e} (2012). See Byrd et al. (2006) for details.} In
optimization, I\ considered the constraints that $0\leq a_{k+1}-a_{k}%
\leq\delta$ and $\delta<a_{k+2}-a_{k}$\ for each integer $k,$and I\ fed the
Jacobian and the Hessian of the Lagrangian into Knitro. Since the objective
function in the optimization is not convex,\ it is likely to have multiple
local maxima. I\ randomly generated initial values 90-200 times using the
"multistart" feature in Knitro.

The numerical optimization results substantially depend on the initial values,
which is the evidence of multiple local maxima and surprisingly, the values of
the objective function at all these local maxima were lower than $V_{K}\left(
\delta\right)  $ in both Section 4 and Section 5. Based on the numerical
evidence, it appears that the global maximum for both Section 4 and Section 5
is achieved or well approximated when $a_{k+1}-a_{k}=\delta$ for each integer
$k$. It remains to show under which conditions on the joint distribution or
marginal distributions the sharp lower bound is indeed achieved when
$a_{k+1}-a_{k}=\delta$\ for each integer $k$.

\setcounter{figure}{0} \renewcommand{\thefigure}{C.\arabic{figure}}
\setcounter{table}{0} \renewcommand{\thetable}{C.\arabic{table}}
\setcounter{equation}{0}
\renewcommand{\theequation}{C.\arabic{equation}}\clearpage

\section*{Appendix C}

Appendix C reports the empirical results which are not documented in Section
5. I report the regression tables for the estimation results for the equations
(15), (16) and (17).

\begin{center}
Table C.1: Probit estimation of the first stage regression%

\begin{tabular}
[c]{lc|lc}\hline\hline
\multicolumn{4}{l}{{\small Dependent Variable: nonsmoking indicator }$D$%
}\\\hline
{\small Tax increase in 1st trimester} & {\small 0.0331} & {\small Age 41+} &
{\small 0.0424}\\
& {\small (0.0013)} &  & {\small (0.0080)}\\
{\small Married} & {\small 0.1270} & {\small High school grad} &
{\small 0.0602}\\
& {\small (0.0019)} &  & {\small (0.0021)}\\
{\small Hispanic} & {\small 0.1214} & {\small Some college} & {\small 0.1361}%
\\
& {\small (0.0027)} &  & {\small (0.0024)}\\
{\small Black} & {\small 0.1551} & {\small College grad.} & {\small 0.2571}\\
& {\small (0.0026)} &  & {\small (0.0029)}\\
{\small Age 2125} & {\small -0.0483} & {\small Post grad.} & {\small 0.2870}\\
& {\small (0.0025)} &  & {\small (0.0035)}\\
{\small Age 2630} & {\small -0.3484} & {\small Adequate care} &
{\small 0.0375}\\
& {\small (0.0027)} &  & {\small (0.0039)}\\
{\small Age 3135} & {\small -0.0174} & {\small Intermediate care} &
\multicolumn{1}{l}{{\small 0.0188}}\\
& {\small (0.0030)} &  & \multicolumn{1}{l}{{\small (0.0041)}}\\
{\small Age 3640} & {\small 0.0078} &  & \multicolumn{1}{l}{}\\
& {\small (0.0037)} &  & \multicolumn{1}{l}{}\\\hline\hline
\end{tabular}

\end{center}

{\scriptsize Note: The table reports the change in the probit
response function due to a change in the indicator variable, with the rest of
the covariates evaluated at the mean. The specification also includes
indicators for birth orders, weight gains and medical risk factors. Robust
standard errors are reported in parentheses.}

\clearpage

\begin{center}
Table C.2: Series estimation of the second stage regression%

\begin{tabular}
[c]{lc|ll|lc}\hline\hline
\multicolumn{6}{l}{{\small Dependent Variable: birth weight (grams)}}\\\hline
$\widehat{p}$ & {\small 1106.07} & $\widehat{p}{\small \times}$%
{\small intermediate care} & {\small -289.97} & {\small Married} &
{\small 46.77}\\
& {\small (168.72)} &  & {\small (135.93)} &  & {\small (6.55)}\\
$\widehat{p}^{2}$ & {\small -647.97} & $\widehat{p}^{2}\times$%
{\small Hispanic} & {\small 295.44} & {\small Hispanic} & {\small -135.88}\\
& {\small (128.59)} &  & {\small (104.16)} &  & {\small (50.80)}\\
$\widehat{p}\times${\small Hispanic} & {\small -209.42} & $\widehat{p}%
^{2}\times${\small black} & {\small 253.63} & {\small Black} &
{\small -294.97}\\
& {\small (145.39)} &  & {\small (84.89)} &  & {\small (40.84)}\\
$\widehat{p}\times${\small black} & {\small -58.92} & $\widehat{p}^{2}\times
${\small age 2125} & {\small 206.27} & {\small Age 2125} & {\small 39.82}\\
& {\small (117.00)} &  & {\small (76.73)} &  & {\small (32.52)}\\
$\widehat{p}\times${\small age 2125} & {\small -179.04} & $\widehat{p}%
^{2}\times${\small age 2630} & {\small 280.49} & {\small Age 2630} &
{\small 25.60}\\
& {\small (100.02)} &  & {\small (77.35)} &  & {\small (32.72)}\\
$\widehat{p}\times${\small age 2630} & {\small -217.70} & $\widehat{p}%
^{2}\times${\small age 3135} & {\small 389.70} & {\small Age 3135} &
{\small 31.38}\\
& {\small (100.81)} &  & {\small (82.29)} &  & {\small (34.66)}\\
$\widehat{p}\times${\small age 3135} & {\small -327.82} & $\widehat{p}%
^{2}\times${\small age 3640} & {\small 311.94} & {\small Age 3640} &
{\small -11.43}\\
& {\small (107.27)} &  & {\small (108.41)} &  & {\small (46.52)}\\
$\widehat{p}\times${\small age 3640} & {\small -230.64} & $\widehat{p}%
^{2}\times${\small age 41+} & {\small -18.82} & {\small Age 41+} &
{\small -139.83}\\
& {\small (144.42)} &  & {\small (265.15)} &  & {\small (119.87)}\\
$\widehat{p}\times${\small age 41+} & {\small 198.31} & $\widehat{p}^{2}%
\times${\small high school grad.} & {\small -155.27} & {\small High school
grad.} & {\small 45.12}\\
& {\small (366.20)} &  & {\small (64.56)} &  & {\small (24.77)}\\
$\widehat{p}\times${\small high school grad.} & {\small 81.92} & $\widehat
{p}^{2}{\small \times}${\small some college} & {\small -197.64} & {\small Some
college} & {\small 90.66}\\
& {\small (79.87)} &  & {\small (81.64)} &  & {\small (30.41)}\\
$\widehat{p}{\small \times}${\small some college} & {\small 82.72} &
$\widehat{p}^{2}{\small \times}${\small college grad.} & {\small -16.43} &
{\small College grad.} & {\small 198.00}\\
& {\small (99.82)} &  & {\small (177.45)} &  & {\small (68.02)}\\
$\widehat{p}{\small \times}${\small college grad.} & {\small -174.91} &
$\widehat{p}^{2}{\small \times}${\small post grad.} & {\small -410.44} &
{\small Post grad.} & {\small 0.74}\\
& {\small (233.59)} &  & {\small (265.75)} &  & {\small (118.50)}\\
$\widehat{p}{\small \times}${\small post grad.} & {\small 392.50} &
$\widehat{p}^{2}{\small \times}${\small adequate care} & {\small 357.39} &
{\small Adequate care} & {\small 237.31}\\
& {\small (373.54)} &  & {\small (105.89)} &  & {\small (37.10)}\\
$\widehat{p}{\small \times}${\small adequate care} & {\small -520.47} &
$\widehat{p}^{2}{\small \times}${\small intermediate care} & {\small 198.06} &
{\small Intermediate care} & {\small 123.90}\\
& {\small (127.37)} &  & {\small (112.30)} &  & {\small (39.95)}\\\hline\hline
\end{tabular}

\end{center}

{\scriptsize Note :The table reports the second stage resgression estimates
for the effect of smoking cessation on infant birth weight. }$\widehat{p}%
${\scriptsize \ denotes the propensity score estimate in the first stage
probit regression. The specification also includes indicators for birth
orders, weight gains and medical risk factors. Robust standard errors are
reported in parentheses.}

\clearpage  

\begin{center}
Table C.3: Quantile regression

Dependent variable: birth weight (grams)%

\begin{tabular}
[c]{lccccc}\hline\hline
& \multicolumn{5}{c}{{\small Quantile}}\\\cline{2-6}
& {\small .15} & {\small .25} & {\small .50} & {\small .75} & {\small .85}%
\\\hline
{\small D (nonsmoking)} & {\small 444.87} & {\small 462.54} & {\small 673.80}
& {\small 259.75} & {\small 365.19}\\
& ({\small 3.32)} & ({\small 2.76)} & ({\small 2.87)} & ({\small 2.77)} &
({\small 3.29)}\\
{\small D*Hispanic} & {\small -72.78} & {\small -195.84} & {\small -502.80} &
{\small -357.08} & {\small -416.70}\\
& ({\small 0.17)} & ({\small 0.35)} & ({\small 0.25)} & ({\small 0.56)} &
(0{\small .29)}\\
{\small D*black} & {\small 259.65} & {\small 474.71} & {\small 46.40} &
{\small -244.92} & {\small -256.61}\\
& ({\small 0.64)} & ({\small 0.25)} & {\small (0.53)} & ({\small 2.14)} &
(0{\small .34)}\\
{\small D*high school grad} & {\small 158.91} & {\small 317.26} &
{\small 203.20} & {\small 34.25} & {\small 36.17}\\
& ({\small 0.24)} & ({\small 0.17)} & (0{\small .30)} & (0{\small .24)} &
(0{\small .38)}\\
{\small D*some college} & {\small 208.87} & {\small 365.21} & {\small 347.40}
& {\small 149.00} & {\small 85.70}\\
& ({\small 0.45)} & ({\small 0.35)} & {\small (0.52)} & (0{\small .38)} &
(0{\small .46)}\\
{\small D*college graduate} & {\small -34.09} & {\small 87.05} &
{\small 305.20} & {\small 542.25} & {\small 324.17}\\
& ({\small 0.82)} & ({\small 0.78)} & {\small (0.83)} & ({\small 0.93)} &
({\small 1.56)}\\
{\small D*post graduate} & {\small 97.57} & {\small 233.63} & {\small 260.60}
& {\small 29.1667} & {\small -209.25}\\
& ({\small 2.40)} & ({\small 1.32)} & ({\small 1.16)} & ({\small 1.08)} &
({\small 2.14)}\\
{\small D*age 2125} & {\small 276.78} & {\small 65.91} & {\small 296.00} &
{\small 22.25} & {\small 107.33}\\
& ({\small 0.24)} & ({\small 0.24)} & {\small (0.32)} & {\small (0.31)} &
{\small (0.39)}\\
{\small D*age 2630} & {\small -71.96} & {\small -224.44} & {\small -332.20} &
{\small -205.25} & {\small -199.94}\\
& {\small (0.40)} & {\small (0.32)} & {\small (0.47)} & {\small (0.32)} &
{\small (0.39)}\\
{\small D*age 3135} & {\small -35.91} & {\small -320.73} & {\small -325.40} &
{\small -304.25} & {\small 39.17}\\
& {\small (0.58)} & ({\small 0.62)} & {\small (0.64)} & {\small (0.52)} &
{\small (0.77)}\\
{\small D*age 3640} & {\small 16.78} & {\small -9.80} & {\small 245.00} &
{\small 293.58} & {\small 661.41}\\
& {\small (0.19)} & ({\small 0.22)} & {\small (0.16)} & ({\small 0.17)} &
{\small (0.24)}\\
{\small D*age 41+} & {\small 411.30} & {\small -238.42} & {\small -403.80} &
{\small 22.08} & {\small -117.84}\\
& ({\small 1.15)} & ({\small 2.99)} & ({\small 1.56)} & ({\small 4.52)} &
{\small (0.64)}\\
{\small Married} & {\small 109.26} & {\small 44.59} & {\small 40.54} &
{\small -34.00} & {\small -255.35}\\
& {\small (0.10)} & {\small (0.09)} & {\small (0.07)} & {\small (0.10)} &
{\small (0.10)}\\
{\small High school grad.} & {\small 160.78} & {\small 140.34} &
{\small 53.97} & {\small -28.00} & {\small -118.73}\\
& {\small (0.31)} & {\small (0.30)} & {\small (0.21)} & {\small (0.30)} &
{\small (0.30)}\\\hline
\end{tabular}

\clearpage Table C.3 - continued from previous page%

\begin{tabular}
[c]{lccccc}\hline\hline
& \multicolumn{5}{c}{{\small Quantile}}\\\cline{2-6}
& {\small .15} & {\small .25} & {\small .50} & {\small .75} & {\small .85}%
\\\hline
{\small Some college} & {\small 326.38} & {\small 196.12} & {\small 134.85} &
{\small 81.00} & {\small 37.58}\\
& {\small (0.64)} & {\small (0.48)} & {\small (0.31)} & {\small (0.69)} &
{\small (0.59)}\\
{\small College graduate} & {\small 412.73} & {\small 228.74} &
{\small 303.27} & {\small -25.00} & {\small 301.14}\\
& {\small (1.13)} & {\small (1.26)} & {\small (1.06)} & {\small (1.18)} &
{\small (1.75)}\\
{\small Post graduate} & {\small 400.21} & {\small 329.48} & {\small 169.92} &
{\small -113.00} & {\small 161.44}\\
& {\small (1.33)} & {\small (0.96)} & {\small (1.47)} & {\small (1.50)} &
{\small (1.89)}\\
{\small Hispanic} & {\small \ 9.68} & {\small 36.67} & {\small -131.40} &
{\small -98.00} & {\small -59.02}\\
& {\small (0.40)} & {\small (0.36)} & {\small (0.33)} & {\small (0.43)} &
{\small (0.41)}\\
{\small Black} & {\small -198.94} & {\small -242.27} & {\small -362.94} &
{\small -316.00} & {\small -201.26}\\
& {\small (1.25)} & {\small (0.30)} & {\small (0.58)} & {\small (0.32)} &
{\small (0.38)}\\
{\small Age 2125} & {\small -152.78} & {\small -4.30} & {\small -71.33} &
{\small 173.00} & {\small 140.92}\\
& {\small (.33)} & {\small (0.32)} & {\small (0.26)} & {\small (0.35)} &
{\small (0.31)}\\
{\small Age 2630 } & {\small -386.94} & {\small -100.38} & {\small -65.79} &
{\small 173.00} & {\small 194.59}\\
& {\small (0.55)} & {\small (0.49)} & {\small (0.31)} & {\small (0.64)} &
{\small (1.13)}\\
{\small Age 3135} & {\small -419.94} & {\small -158.99} & {\small -348.74} &
{\small 122.00} & {\small 132.61}\\
& {\small (0.83)} & {\small (0.74)} & {\small (0.56)} & {\small (0.86)} &
{\small (0.63)}\\
{\small Age 3640} & {\small -326.36} & {\small -295.01} & {\small -238.45} &
{\small -51.00} & {\small 322.54}\\
& {\small (1.12)} & {\small (0.51)} & {\small (0.26)} & {\small (0.68)} &
{\small (0.53)}\\
{\small Age 41+} & {\small -464.89} & {\small -184.90} & {\small -60.71} &
{\small 77.00} & {\small 183.04}\\
& {\small (1.37)} & {\small (4.96)} & {\small (2.76)} & {\small (0.99)} &
{\small \ (0.91)}\\
{\small Adequte care} & {\small -63.38} & {\small 193.99} & {\small -19.37} &
{\small 240.00} & {\small 253.58}\\
& {\small (0.61)} & {\small (0.63)} & {\small (0.59)} & {\small (0.47)} &
{\small (0.36)}\\
{\small Intermediate care} & {\small -188.28} & {\small 12.91} &
{\small -82.46} & {\small -1.00} & {\small 90.96}\\
& {\small (0.70)} & {\small (0.70)} & {\small (0.60)} & {\small (.53)} &
{\small (0.44)}\\\hline\hline
\end{tabular}

\end{center}

{\scriptsize Note: The table reports quantile regression estimates for the
effect of smoking on the quantiles of infant birth weight for compliers. The
tax increase is used as an instrument for smoking. The specification also
includes indicators for birth orders, weight gains and medical risk factors.
Robust standard errors are reported in parentheses.} \newline\newline

\end{document}